\documentclass{article}
\usepackage[T1]{fontenc}
\usepackage[utf8]{inputenc}
\usepackage{lmodern}
\usepackage[a4paper,margin=2.5cm]{geometry}

\usepackage{amsmath,amsfonts,amssymb,amsthm,mathtools}

\usepackage{bm}
\usepackage{mathrsfs}
\usepackage{stmaryrd}
\usepackage{bbm}

\usepackage{graphicx}
\usepackage{tikz}
\usepackage{tikz-cd}

\usepackage{enumitem}
\usepackage{xcolor}
\usepackage{booktabs}

\DeclareMathOperator{\Ad}{Ad}

\DeclareMathOperator{\Tr}{Tr}

\DeclareMathOperator{\Op}{Op}

\newcommand{\Cinfty}{\mathscr{C}^\infty}

\newcommand{\RR}{\mathbb{R}}

\newcommand{\dd}{\mathrm{d}}
\newcommand{\LieD}{\mathcal{L}}

\newcommand\restr[2]{{
  \left.\kern-\nulldelimiterspace 
  #1 
  \right|_{#2} 
}}

\newcommand*{\contr}[1]{\iota_{#1}}
\newcommand*{\liedv}[1]{\mathcal{L}_{#1}}
\newcommand{\T}{{\mathsf T}}
\newcommand{\cT}{\T^{\ast}}

\let\coloneqq\defeq

\usepackage{amsthm}

\numberwithin{equation}{section}

\newtheorem{theorem}{Theorem}[section]
\newtheorem{proposition}[theorem]{Proposition}
\newtheorem{lemma}[theorem]{Lemma}
\newtheorem{corollary}[theorem]{Corollary}
\newtheorem{problem}{Problem}
\newtheorem*{problem*}{Problem}

\theoremstyle{definition}
\newtheorem{definition}[theorem]{Definition}
\newtheorem{assumption}[theorem]{Assumption}
\newtheorem{example}[theorem]{Example}

\theoremstyle{remark}
\newtheorem{remark}[theorem]{Remark}

\title{\textbf{Egorov-Type Semiclassical Limits for Open Quantum Systems with a Bi-Lindblad Structure}}
\author{Leonardo Colombo\thanks{Centro de Automática y Robótica (CSIC-UPM), Carretera de Campo Real, km 0, 200, 28500 Arganda del Rey, Spain \\ email: \href{mailto:leonardo.colombo@csic.es}{leonardo.colombo@csic.es}; ORCID: \href{https://orcid.org/0000-0001-6493-6113}{0000-0001-6493-6113}.} \, and Asier L\'opez-Gord\'on\thanks{Instytut Matematyczny Polskiej Akademii Nauk (IM PAN), ul.~\'Sniadeckich 8, 00-656 Warsaw, Poland\\ email: \href{mailto:alopez-gordon@impan.pl}{alopez-gordon@impan.pl}; ORCID: \href{https://orcid.org/0000-0002-9620-9647}{0000-0002-9620-9647}.}}
\date{\vspace*{-1cm}}

\usepackage[style=ext-numeric,citestyle=numeric-comp,sorting=nyt,sortcites, backend=biber, giveninits=true, maxnames=99]{biblatex}
\bibliography{biblio}
\DeclareFieldFormat[article,periodical]{volume}{\mkbibbold{#1}}
\DeclareFieldFormat[article,periodical]{number}{\mkbibparens{#1}}
\renewbibmacro{in:}{%
  \ifentrytype{article}
    {}
    {\printtext{\bibstring{in}\intitlepunct}}}

\DeclareFieldFormat{issuedate}{#1}
\renewcommand{\jourvoldelim}{\addcomma\space}

\renewbibmacro*{journal+issuetitle}{%
  \usebibmacro{journal}%
  \setunit*{\jourvoldelim}%
  \iffieldundef{series}
    {}
    {\setunit*{\jourserdelim}%
     \printfield{series}%
     \setunit{\servoldelim}}%
  \usebibmacro{volume+number+eid}%
  \setunit{\addcolon\space}%
  \usebibmacro{issue}%
  \setunit{\bibpagespunct}%
  \printfield{pages}%
  \setunit{\space}%
  \usebibmacro{issue+date}%
  \newunit}

\renewbibmacro*{note+pages}{%
  \printfield{note}%
  \newunit}

\DeclareFieldFormat[article,periodical]{date}{\mkbibparens{#1}}

\AtEveryBibitem{\clearfield{month}}
\AtEveryBibitem{\clearfield{day}}

\AtEveryBibitem{
  \ifentrytype{online}{}{
    \ifentrytype{thesis}{}{
      \ifentrytype{unpublished}{}{ 
        \clearfield{url}
      }
    }
  }
  \ifentrytype{unpublished}{}{ 
  \clearfield{note}
  }
  \clearfield{language}

  \ifentrytype{book}{\clearfield{pages}}{}
}

\DeclareSourcemap{
  \maps[datatype=bibtex]{
    \map[overwrite=true]{
      \step[fieldset=language, null]
      \step[fieldset=address, null]
     \step[fieldset=pagetotal, null]
    }
  }
}

\DeclareSourcemap{
    \maps[datatype=bibtex]{
      \map{
        \step[fieldsource=eprint,final]
        \step[fieldset=url,null]
        \step[fieldset=doi,null]
      }  
    }
  }

\AtEveryBibitem{\clearfield{pubstate}} 
\AtEveryBibitem{\clearfield{version}} 

\usepackage[unicode=true, pdfusetitle, colorlinks=true, allcolors=blue, bookmarks=false]{hyperref}
\usepackage[hyphenbreaks]{breakurl}
\usepackage{xurl}
\hypersetup{breaklinks=true}
\usepackage[nameinlink,capitalize]{cleveref}
\usepackage{nameref}

\newcommand{\change}[1]{#1}

\begin{document}

\maketitle


\begin{abstract}
\noindent

\noindent This paper develops a bridge between bi-Hamiltonian structures of Poisson--Lie
type, contact Hamiltonian dynamics, and the Gorini--Kossakowski--Sudarshan--Lindblad
(GKSL) formalism for open quantum systems. On the classical side, we consider bi-Hamiltonian systems defined by a Poisson pencil with
non-trivial invariants.  Using an exact symplectic realization, these invariants are lifted and projected onto a contact manifold,
yielding a completely integrable contact Hamiltonian system and a Jacobi-commutative algebra of observables. On the quantum side, we introduce a class of \emph{contact-compatible Lindblad
generators}: GKSL evolutions whose dissipative part preserves a commutative
$C^\ast$-subalgebra generated by the quantizations of the classical dissipated
quantities, and whose Hamiltonian part admits an Egorov-type semiclassical limit
to the contact dynamics.  This construction provides a mathematical mechanism compatible with the semiclassical limit for pure dephasing, compatible with integrability and
contact dissipation. An explicit Poisson--Lie pencil, inspired by deformed Euler top models, is developed as a
fully worked-out example illustrating the resulting bi-Lindblad structure and its
semiclassical behavior.

\bigskip


\bigskip

\noindent\textbf{MSC\,2020 codes:}
81S22 (primary); 
37J35,  
37J55, 
46L60,  
53Z05, 
81Q70,  
81R12, 
81S10.  
\end{abstract}


\section{Introduction}

Bi-Hamiltonian geometry plays a central role in the theory of integrable
systems. A \emph{bi-Hamiltonian system} consists of a smooth manifold endowed
with two \emph{compatible} Poisson tensors $\Pi_0,\Pi_1$, meaning that every linear
combination
\[
\Pi_\lambda \coloneqq \Pi_1 - \lambda \Pi_0,\qquad \lambda\in\mathbb{R},
\]
is again a Poisson tensor. A vector field $X$ is called
bi-Hamiltonian if there exist Hamiltonian functions $H_0,H_1$ such that
\[
X = \Pi_0^\sharp(\dd H_0) = \Pi_1^\sharp(\dd H_1)\, .
\]
Compatibility ensures the existence of a recursion operator (on a symplectic leaf where one of the Poisson structures is non-degenerate) and hierarchies of
commuting Hamiltonians, which in turn give rise to rich algebraic and geometric
structures and, under suitable regularity assumptions, to complete integrability
in the sense of Liouville; see, for example,
\cite{M.M1984,C.M.P1993,M.C.F+1997,F.M.P2000,Fernandes1994,G.M.S2003}.

A natural and geometrically significant source of bi-Hamiltonian systems is
provided by \emph{Poisson--Lie groups}. A Poisson tensor $\pi$ on a Lie group
$G$ is called \emph{multiplicative} if the group multiplication
$m\colon G\times G \longrightarrow G$ is a Poisson map from
$(G\times G,\pi\oplus\pi)$ to $(G,\pi)$. Multiplicative Poisson structures and
their infinitesimal counterparts, Lie bialgebras, were introduced by Drinfel'd
in the context of quantum groups \cite{Drinfeld1988,Drinfeld1990} and further
developed in the foundational work of Semenov--Tian--Shansky on the classical
$r$-matrix formalism, Poisson pencils and factorization dynamics
\cite{STS1983,STS1985}. The global geometry of Poisson--Lie groups, including
dressing transformations and symplectic groupoids, was elucidated in
\cite{LuWeinstein1990}. 

When two multiplicative Poisson tensors $\pi_0,\pi_1$ on $G$ are compatible, the resulting Poisson pencil $\pi_\lambda=\pi_1-\lambda\pi_0,\, \lambda\in \RR$ encodes a rich class of
bi-Hamiltonian vector fields on $G$. In particular, if $H_0,H_1\in \Cinfty(G)$
satisfy
\[
X = \pi_0^\sharp(\dd H_0) = \pi_1^\sharp(\dd H_1),
\]
then $X$ is a bi-Hamiltonian vector field whose integrability properties reflect
the algebraic structure of the underlying Lie bialgebra and classical
$r$-matrix. This mechanism underlies many classical integrable models on Lie
groups and Poisson--Lie groups
\cite{STS1983,STS1985,KS2004,KS2008,LuWeinstein1990}.

On the other hand, contact geometry provides a natural framework for
describing non-Hamiltonian classical dynamics in a plethora of contexts, including certain dissipative systems and
thermodynamically relevant flows; see, e.g.,
\cite{d.L2019a,d.L2019b,G.G2022a,G.G2024a,Lopez-Gordon2024}. There is a one-to-one correspondence between contact manifolds and homogeneous symplectic manifolds, which has been known for a while \cite{Arnold1978,L.M1987}, and whose applications to contact Hamiltonian dynamics have been recently analyzed in detail by Grabowska and Grabowski \cite{Grabowski2013a, G.G2022a, G.G2023,G.G2024a}. We have exploited this correspondence in order to establish a novel characterisation of completely integrable contact Hamiltonian systems as completely integrable Hamiltonian systems (in the standard Liouville sense) that satisfy certain homogeneity conditions \cite{C.d.L+2023,Lopez-Gordon2024,Colombo2025}.

At this stage, it is natural to ask why one should restrict to homogeneous
symplectic realizations projectable to contact manifolds, rather than considering
a generic symplectic realization of an open quantum system. Although a classical
limit can be obtained from such generic symplectic realizations, these
constructions are highly non-unique and do not provide a canonical description of
dissipation. Requiring the realization to be homogeneous and projectable to a
contact structure selects a distinguished geometric class in which dissipation,
contraction of phase-space volume, and irreversibility are encoded intrinsically.
In this framework, the effective non-unitary dynamics arises as the projection of
a unitary flow on a higher-dimensional homogeneous symplectic manifold, yielding a
robust and universal semiclassical limit. The contact-compatible structure is
therefore not an additional assumption, but the minimal geometric setting
capturing the correct physics of open quantum systems.

A key point, which we emphasise
here, is that:

\begin{itemize}
\item the \emph{integrability} of the underlying symplectic dynamics is still
      controlled by Poisson (or Poisson--Nijenhuis) structures and their
      compatibility, exactly as in the standard bi-Hamiltonian theory;
\item the \emph{homogeneity} is an
      independent geometric ingredient that guarantees that the dynamics and its
      commuting integrals can be projected onto a contact manifold.
\end{itemize}

In particular, in \cite{Colombo2025}, it is shown, in concrete
examples, that one may take a Poisson structure on the symplectization
compatible with the canonical one, use the corresponding recursion operator to
construct a maximal family of homogeneous functions in involution, and then
``project'' these functions to a $(2n+1)$-dimensional contact manifold as
functions in Jacobi-involution. What turns out to be too rigid for
contact integrability is not Poisson compatibility on the symplectization, but
the use of \emph{compatible Jacobi structures} (see
\cite{NunesdaCosta1998a,M.M.P1999,P.N2003,Petalidou2002}) directly on the original contact
manifold: compatible Jacobi pairs do not yield a recursion operator with enough
independent eigenvalues to produce a full set of contact integrals.

This paper aims to unify bi-Hamiltonian Poisson geometry with the
homogeneous/contact correspondence and to connect the resulting classical
framework with the quantum realm. In doing so, we bring together three layers of
structure:
\begin{enumerate}
\item bi-Hamiltonian dynamics on Poisson--Lie groups and their symplectic (or
      exact symplectic) realizations;
\item the projection of homogeneous Hamiltonian systems to contact dynamics via
      Liouville-transverse hypersurfaces;
\item quantum open-system evolution of GKSL type.
\end{enumerate}
To avoid conceptual ambiguity, we stress a fundamental distinction:

\begin{itemize}
\item Systems defined on a Poisson--Lie group $G$, and on any of its symplectic or
      exact symplectic realizations, are entirely \emph{classical Hamiltonian
      systems}. Their integrability is governed by Poisson, Poisson--Nijenhuis
      and bi-Hamiltonian structures in the usual sense.
\item The Lindblad equation that we construct is a genuinely \emph{quantum}
      evolution equation (of GKSL type), acting on density matrices over a
      suitable Hilbert space obtained by quantizing $G$ or one of its
      symplectic realizations. Its dissipative part has no classical analog at
      the level of closed Hamiltonian dynamics, although it admits a meaningful
      semiclassical description.
\end{itemize}

A \emph{contact Hamiltonian system} is a triple $(C,\alpha,h)$ consisting of a manifold $C$ endowed with a contact form $\alpha$ and a fixed function \change{$h\in \Cinfty(C)$}.
In this context, by a \emph{dissipated quantity} we mean a classical observable
$I\in \Cinfty(C)$ such that
\[
\{I,h\}_\alpha = 0,
\]
where $\{\cdot, \cdot\}_\alpha$ is the Jacobi bracket defined by $\alpha$. The trivial symplectization of $(C, \alpha, h)$ is given by $(M=C\times (0, +\infty), \change{\omega=-\dd(r\alpha)}, H = r h)$, where $r$ denotes the global coordinate of $(0, +\infty)$. There is a Lie algebra isomorphism between the Jacobi bracket $\{\cdot, \cdot\}_\alpha$ and the Poisson bracket $\{\cdot, \cdot\}_\omega$, which implies that dissipated quantities for $(C, \alpha, h)$ are in one-to-one correspondence with conserved quantities for $(M, \omega, H)$. 

A symplectic form $\omega$ on a manifold $M$ is exact (i.e., there exists a one-form $\theta$ such that $\dd \theta = -\omega$) if and only if there exists a \emph{Liouville vector field} $\Delta$ on $M$ such that $\LieD_\Delta \omega = \omega$. Indeed, one can define $\theta = -\iota_\Delta \omega$. A tensor field $A$ on $M$ is called \emph{homogeneous of degree $w\in \RR$} if $\LieD_\Delta A = w\cdot A$. On each hypersurface $C$ of $M$ which is transverse to $\Delta$, the restriction $\alpha = \restr{\theta}{C}$ is a contact form, and homogeneous functions of degree $1$ on $M$ project onto functions on $C$. Conversely, given a function $g$ on $C$, one can extend it, via the flow $\phi_t$ of $\Delta$, to a homogeneous function $\hat g$ on $M$ given by $\hat{g}(\phi_t(x))=e^t g(x)$ for each $x\in C$.

By \emph{quantum dissipation} we refer to the irreversible, non-unitary component
of the reduced dynamics of an open quantum system, described at the Markovian
level by a GKSL (Lindblad) generator. This dissipative term arises from tracing
out environmental degrees of freedom and accounts for phenomena such as
decoherence, relaxation and thermalisation \cite{BreuerPetruccione}. While such
dissipation has no direct analog in classical Hamiltonian mechanics, it admits
effective classical limits. In the present work we exhibit a class of GKSL
generators whose semiclassical limit is naturally described by contact
Hamiltonian dynamics, with dissipated quantities playing the role of robust
invariants of the resulting reduced evolution.


Our main goal is to show that the classical dissipated quantities arising (via realization and projection) from
a bi-Hamiltonian Poisson--Lie system admit quantum counterparts that are
invariant under a suitably constructed GKSL generator, thereby establishing a
precise bridge between Poisson--Lie integrability, homogeneous/contact
Hamiltonian geometry, and physically admissible quantum dissipation.

Before introducing the specific quantization scheme, it is useful to explain
why such a quantum layer is natural in the present setting.  
A bi-Hamiltonian Poisson--Lie system produces, through an exact and homogeneous
symplectic realization, a completely integrable contact system whose dynamics is
intrinsically non-Hamiltonian, as contact vector fields do not preserve either the contact form or the associated ``Hamiltonian'' function. From the viewpoint of physics, contact systems appear precisely as
\emph{effective} or \emph{reduced} descriptions of systems interacting with an
environment, where certain degrees of freedom are not explicitly modelled.
Quantum open systems \cite{BreuerPetruccione, Lindblad1976}, governed by Lindblad master (GKSL) equations, arise from a conceptually analogous mechanism: a unitary evolution on a larger Hilbert space is
reduced to a non-unitary dynamics on the subsystem of interest.

This parallel suggests that the contact dynamics obtained from the homogeneous
projection of a Poisson--Lie system should admit a natural quantum analog.
However, making this correspondence precise requires a quantization of the
classical observables on the symplectic realization $(M,\omega)$ (or directly on
the Poisson--Lie group $G$), together with a quantum evolution whose
semiclassical limit reproduces the contact flow. To interpret the classical
dissipated quantities $I_0,\dots,I_n$ as quantum constants of motion, and to
construct a GKSL generator whose dissipative part reflects the geometric content
of the contact structure, one must assign to each classical observable
$f\in \Cinfty(M)$ a quantum operator acting on a Hilbert space.

In more physical terms, by ``quantizing'' the above classical data we mean choosing
a Hilbert space $\mathcal{H}$ and a family of maps
\[
Q_\hbar \colon \Cinfty(M) \;\longrightarrow\; \mathcal{B}(\mathcal{H}),
\qquad
f \longmapsto \widehat{f} = Q_\hbar(f),
\]
defined for sufficiently small $\hbar>0$, such that $\widehat{f}$ is interpreted
as the quantum observable corresponding to the classical observable $f$.
Here $\mathcal{H}$ denotes a (finite-dimensional) Hilbert space, and
$\mathcal{B}(\mathcal{H})$ the algebra of linear operators on $\mathcal{H}$.
Since $\mathcal{H}$ is finite-dimensional, every linear operator on
$\mathcal{H}$ is automatically bounded, so that
$\mathcal{B}(\mathcal{H})=\mathrm{End}(\mathcal{H})$ and no domain or
topological issues arise. In particular, all operator topologies coincide and
the algebraic formulation of the quantum dynamics is completely sufficient for
our purposes.

The assumption of finite-dimensionality is made for conceptual clarity and to
avoid technicalities related to unbounded operators; it does not impose any
restriction on the classical phase space $M$ itself. Rather, it should be
understood as a semiclassical or effective quantization of the Poisson algebra
of observables, capturing the contact dynamics of interest in a finite quantum
state space. All constructions below extend, under standard hypotheses, to
separable infinite-dimensional Hilbert spaces.


Depending on the context, $Q_\hbar$ may be realized by geometric quantization of a
symplectic realization of $G$, by (unitary) representation theory of $G$, or by
deformation quantization of the Poisson manifold $(G,\pi_0)$ or $(M,\omega)$. Quantum states are positive, trace-one operators $\rho$ on $\mathcal{H}$, and their
time evolution is described by a GKSL generator of the form
\begin{equation}
\label{eq:intro-GKSL}
\mathcal{L}_\hbar(\rho)
= -\frac{i}{\hbar}\big[\,\widehat{H}_\hbar,\rho\,\big]
  + \sum_{k=1}^{N(\hbar)} \Big(
      L_{k,\hbar}\,\rho\, L_{k,\hbar}^\dagger
      - \tfrac12 \,\big\{\,L_{k,\hbar}^\dagger L_{k,\hbar},\,\rho\,\big\}
    \Big),
\end{equation}
where $\widehat{H}_\hbar$ is the quantization of a classical Hamiltonian,
$[\cdot,\cdot]$ is the operator commutator, and $\{\cdot,\cdot\}$ denotes the
operator anti-commutator, $\{A,B\}\coloneqq AB+BA$. The family
$\{L_{k,\hbar}\}_{k=1}^{N(\hbar)}$ encodes the coupling to the environment; here
$N(\hbar)\in\mathbb{N}$ is finite for each fixed $\hbar$ and may depend on $\hbar$
in the chosen semiclassical regime. In the constructions considered in this paper, $\widehat{H}_\hbar$ is chosen
as the quantization of the homogeneous Hamiltonian $H$ on the exact symplectic realization
$(M,\omega=-\dd \theta)$ of a Poisson--Lie group, whose restriction to a Liouville-transverse hypersurface $C\subset M$ yields the
contact Hamiltonian $h=\restr{H}{C}$.


The \emph{semiclassical limit} $\hbar\to 0$ is understood either in the sense of
Wigner transforms or within deformation quantization. In this limit, commutators
of quantized observables reproduce the underlying classical brackets: for
observables $a,b$ on the symplectic realization $(M,\omega)$ one has
\[
\frac{1}{i\hbar}\big[\,Q_\hbar(a),Q_\hbar(b)\,\big]
\;\longrightarrow\;
\{a,b\}_\omega
\qquad\text{as }\hbar\to 0,
\]
and, after restriction to a Liouville-transverse hypersurface
$C\subset M$, this induces the corresponding Jacobi bracket
$\{\cdot,\cdot\}_\alpha$ on the contact manifold $(C,\alpha)$.

Accordingly, the GKSL evolution~\eqref{eq:intro-GKSL} reduces in the semiclassical
limit to the classical Hamiltonian dynamics on $(M,\omega)$ or, after projection,
to the contact Hamiltonian dynamics on $(C,\alpha,h)$.
In particular, the quantum
observables associated with the classical dissipated quantities
$I_0,\dots,I_n$ are required to become constants of motion in the Heisenberg
picture, and their semiclassical limit recovers the dissipated quantities of the
underlying contact integrable system.


At the technical level, the semiclassical passage from the Heisenberg-picture GKSL
evolution to the underlying classical contact dynamics will be formulated through an
Egorov-type requirement: quantized observables evolved by $\mathcal{L}_\hbar^\ast$
must be approximated, to leading order in $\hbar$, by classical transport along the
corresponding (contact) flow, where $\mathcal{L}_\hbar^\ast$ denotes the adjoint GKSL generator in the
Heisenberg picture, defined with respect to the trace duality. This viewpoint naturally leads to invariant commutative
$C^\ast$-subalgebras generated by quantizations of the classical dissipated quantities,
which will play a key role in our analysis of dephasing-type mechanisms.

The classical part of our construction proceeds as follows. Starting from a
bi-Hamiltonian Poisson--Lie system, we consider an exact symplectic realization
where the Hamiltonian is homogeneous of degree one with respect to the Liouville
vector field. On this exact symplectic manifold we exploit bi-Hamiltonian
techniques to build a family of commuting integrals. Homogeneity then allows us
to project the dynamics and its integrals to a contact
manifold, obtaining a completely integrable contact system in the sense of
\cite{C.d.L+2023,Lopez-Gordon2024}.


The main question we address is whether one can construct an open quantum dynamics
naturally associated with this classical framework, whose semiclassical limit
reproduces the contact Hamiltonian flow.

\begin{problem}
Let $G$ be a Poisson--Lie group endowed with two compatible multiplicative
Poisson bivectors $\pi_0,\pi_1$, and let $H_0,H_1\in \Cinfty(G)$ be Hamiltonian
functions such that the corresponding Hamiltonian vector field
\[
X = \pi_0^\sharp(\dd H_0) = \pi_1^\sharp(\dd H_1)
\]
is bi-Hamiltonian.

Suppose that $(G,\pi_0)$ admits an exact symplectic realization $(M,\theta)$ and
that $X$ lifts to a homogeneous Hamiltonian system $(M,\theta,H)$, with $H$ of
degree $1$ with respect to the Liouville vector field. Let $C\subset M$ be a
hypersurface transverse to the Liouville flow, and set
\[
\alpha \coloneqq \theta|_C,\qquad h \coloneqq H|_C,
\]
so that $(C,\alpha,h)$ is a $(2n+1)$-dimensional completely integrable contact
Hamiltonian system, with classical dissipated quantities
\[
I_0 \coloneqq h,\quad I_1,\dots,I_n \in \Cinfty(C)
\]
coming from the bi-Hamiltonian hierarchy on $G$.

Can one construct a Lindblad generator---that is, a physically admissible GKSL
generator $\mathcal{L}$ of quantum open-system dynamics on a suitable quantum
state space (for instance, on density matrices over a Hilbert space
$\mathcal{H}$ obtained by quantizing $G$ or $M$) such that:
\begin{enumerate}
\item its semiclassical limit (in the sense of $\hbar\to 0$ or deformation
      quantization) reproduces the contact Hamiltonian flow of $(C,\alpha,h)$;
\item its conserved observables $\widehat{I}_0,\dots,\widehat{I}_n$ are
      quantizations of the classical dissipated quantities
      $I_0,\dots,I_n$?
\end{enumerate}
\end{problem}

Our answer is \emph{affirmative} (see Theorem~\ref{prop:egorov-criterion}).
We introduce the notion of a \emph{contact-compatible Lindblad generator}: a
GKSL dynamics whose Hamiltonian and dissipative parts are organized so that

\begin{itemize}
\item the Heisenberg-picture evolution admits a well-defined semiclassical limit
      on the contact phase space $(C,\alpha)$, reproducing the contact flow of
      $h$;
\item the quantum constants of motion of the Lindblad evolution correspond, in
      the semiclassical limit, to the classical dissipated quantities of the
      contact integrable system.
\end{itemize}

We also introduce a notion of \emph{bi-Lindblad structure}, inspired by
bi-Hamiltonian geometry: a pair of Lindblad generators whose convex combinations
remain Lindbladian and whose common invariant observables encode the classical
integrals of a bi-Hamiltonian system on a Poisson--Lie group. This provides a
geometric bridge between Poisson--Lie geometry, contact integrable dynamics and
physically admissible quantum dissipation (see Proposition~\ref{prop:bi-Lindblad-implies-contact-compatible}).

An Euler-top-type Poisson--Lie system, inspired by the deformed Euler top
introduced in \cite{ballesteros2017poisson,G.I.M+2023}, is developed here as a
fully explicit example illustrating the classical--quantum correspondence
between bi-Hamiltonian contact dynamics and Lindblad open-system evolution.

\medskip

\noindent \textbf{Structure of the paper.}
Section~\ref{sec:bi-ham-contact} reviews bi-Hamiltonian systems, Poisson pencils
and homogeneous Hamiltonian dynamics, with particular emphasis on how Poisson
compatibility and homogeneity give rise to Jacobi-commutative families of
integrals.
Section~\ref{sec:PL-contact} develops the passage from Poisson--Lie
bi-Hamiltonian data to completely integrable contact systems via homogeneous
symplectization and Liouville reduction, introducing the notion of dissipated
quantities and their associated Jacobi algebras.
Section~\ref{sec:lindblad} introduces contact-compatible and bi-Lindblad
generators in the GKSL formalism, establishes an Egorov-type semiclassical
criterion, and analyzes the emergence of invariant commutative $C^\ast$-algebras
and pure dephasing mechanisms.
Section~\ref{sec:euler-example} illustrates the construction on an explicit
low-dimensional Poisson--Lie pencil of Euler-top type, inspired by deformed
Euler-top models, showing how integrability, dissipation and decoherence are
combined in a bi-Lindblad framework.
Finally, we conclude with a discussion of the physical interpretation and connections with integrable open quantum systems.

\medskip

\noindent \textbf{Notation and conventions.} All manifolds and maps are $\Cinfty$-smooth. Given a vector field $X$ and a $p$-form $\alpha$ on a manifold $M$, the exterior derivative of $\alpha$ is denoted by $\dd \alpha$, the contraction of $\alpha$ with $X$ is denoted by $\contr{X}\alpha$. The Lie derivative of a tensor field $A$ on $M$ with respect to $X$ is denoted by $\liedv{X} A$.
The positive real numbers will be denoted by $\RR_{>0}$. Given a vector bundle $E\to M$, the space of sections of $E$ is denoted by $\Gamma(E)$.
The Jacobi (resp.~Poisson) bracket defined by a contact form $\alpha$ (resp.~symplectic form $\omega$) is denoted by $\{\cdot, \cdot\}_\alpha$ (resp.~$\{\cdot, \cdot\}_\omega$). If $M\simeq \RR^{2d}$, the standard class of semiclassical symbols of order zero
on $M$ will be denoted by $S^0(M)$, and $\mathcal{S}(\mathbb{R}^d)$ will denote the Schwartz space of rapidly decreasing smooth functions on $\mathbb{R}^d$. For simplicity's sake, we will assume Hilbert spaces to be finite-dimensional. Given a Hilbert space $\mathcal{H}$, we will denote
the $C^\ast$-algebra of (bounded) linear operators on $\mathcal{H}$ by $\mathcal{B}(\mathcal{H})$. The operator anticommutator is $\{A,B\}\coloneqq AB+BA$ for each $A, B \in \mathcal{B}(\mathcal{H})$.

\section{Bi-Hamiltonian geometry and contact integrability}
\label{sec:bi-ham-contact}

This section reviews classical bi-Hamiltonian geometry and homogeneous Hamiltonian 
systems on exact symplectic manifolds, and recalls the notion of contact 
integrability introduced in \cite{C.d.L+2023,Lopez-Gordon2024}.

In this section we work entirely in the classical realm: Poisson, symplectic 
and contact structures live on smooth manifolds, and no quantization procedure 
is involved. The contact systems and homogeneous Hamiltonian systems constructed 
here will later serve as the classical (or semiclassical) phase-space models for 
the GKSL generators introduced in Section~\ref{sec:lindblad}.

On the one hand, bi-Hamiltonian geometry is entirely formulated in the Poisson (or
symplectic) category: compatibility of Poisson structures leads to recursion 
operators and hierarchies of commuting Hamiltonians 
\cite{M.M1984,M.C.F+1997,F.M.P2000,Fernandes1994,C.M.P1993,G.M.S2003}. 
On the other hand, homogeneity is a property with respect to the Liouville vector
field on an exact symplectic manifold. These two ingredients are independent: bi-Hamiltonian (or Poisson--Nijenhuis) structures will be
used to generate integrals in involution in the usual symplectic sense, while
homogeneity will be the geometric mechanism that allows us to project these
integrals to a contact manifold of one dimension less.

\subsection{Bi-Hamiltonian structures}

Let $M$ be a smooth manifold endowed with two Poisson tensors 
$\Pi_0,\Pi_1\in\Gamma(\wedge^2 \T M)$. 

\begin{definition}
The pair $(\Pi_0,\Pi_1)$ is \emph{compatible} if every linear combination
\[
\Pi_\lambda \coloneqq \Pi_1 - \lambda \Pi_0, \qquad \lambda\in\mathbb{R},
\]
is a Poisson tensor.
\end{definition}
The family $\Pi_\lambda$ of Poisson tensors is called a \emph{Poisson pencil}, and it defines a family of Poisson brackets
\[
\{f,g\}_\lambda \coloneqq \Pi_\lambda(\dd f,\dd g)\, , \quad \forall\, f, g\in \Cinfty(M).
\]

\begin{definition}
A vector field $X\in\mathfrak{X}(M)$ is \emph{bi-Hamiltonian} with respect to 
$(\Pi_0,\Pi_1)$ if there exist smooth functions $H_0,H_1\in \Cinfty(M)$ such that
\[
X = \Pi_0^\sharp(\dd H_0) = \Pi_1^\sharp(\dd H_1).
\]
\end{definition}

On any symplectic leaf where $\Pi_0$ is non-degenerate, one can introduce the 
\emph{recursion operator} (see, for instance, 
\cite{M.M1984,M.C.F+1997,F.M.P2000,Fernandes1994})
\[
N \coloneqq \Pi_1^\sharp \circ (\Pi_0^\sharp)^{-1}.
\]
The Nijenhuis torsion of $N$ vanishes, and its eigenvalues are commuting functions in involution
providing the standard bi-Hamiltonian 
framework for Liouville integrability \cite{C.M.P1993,G.M.S2003}. 


\subsection{Exact symplectic manifolds and homogeneity}

Throughout the paper we adopt the convention that an exact symplectic manifold
$(M,\theta)$ carries the symplectic form $\omega \coloneqq -\dd\theta$. Hamiltonian vector fields are defined by $\iota_{X_H}\omega = \dd H$, and the \emph{Liouville vector field} $\Delta$ is uniquely determined by $\iota_\Delta\omega = -\theta$.

Let $(M,\theta)$ be an exact symplectic manifold with symplectic form
$\omega=-\dd\theta$. The one-form $\theta$ is called a \emph{symplectic potential}.

\begin{definition}
A tensor field $A$ on $M$ is said to be \emph{homogeneous of degree $w\in\mathbb{R}$}
(or simply \emph{$w$-homogeneous}) if
\[
\mathcal{L}_\Delta A = w\,A .
\]
In particular, a function $H\in \Cinfty(M)$ is \emph{homogeneous of degree $1$}
if $\mathcal{L}_\Delta H = H$.
A triple $(M,\theta,H)$ with $H$ of degree $1$ is called a
\emph{homogeneous Hamiltonian system}.
\end{definition}

\begin{remark}
The symplectic potential $\theta$ and the symplectic form $\omega$ are
$1$-homogeneous with respect to the Liouville vector field $\Delta$.
Moreover, if $f\in \Cinfty(M)$ is a $1$-homogeneous function, then its
Hamiltonian vector field $X_f$, defined by $\iota_{X_f}\omega = \dd f$, is $0$-homogeneous, i.e.\ $\mathcal{L}_\Delta X_f = 0$.
In other words, the Hamiltonian flow of a $1$-homogeneous function is invariant
under the Liouville flow.
\end{remark}

Following \cite{C.d.L+2023,Lopez-Gordon2024}, we say that a homogeneous Hamiltonian 
system is integrable if it admits a full set of independent, homogeneous first 
integrals in involution.

\begin{definition}\label{def:homogeneous_integrable_system}
A \emph{homogeneous integrable system} $(M,\theta,H)$ consists of a 
$2n$-dimensional exact symplectic manifold $(M,\theta)$, a $1$-homogeneous 
Hamiltonian $H$, and functions $f_1,\dots,f_n\in \Cinfty(M)$ such that:
\begin{enumerate}
\item Each $f_i$ is a first integral of $X_H$, and homogeneous of degree $1$;
\item The $f_i$ are in involution for the Poisson bracket induced by $\omega$;
\item The differentials $\dd f_1,\dots,\dd f_n$ are independent on a dense open subset.
\end{enumerate}
\end{definition}

In our later construction, we will start from a Poisson (or Poisson--Lie) 
bi-Hamiltonian system and, when an exact symplectic realization is available, we 
will require the Hamiltonian and its integrals to be $1$-homogeneous with respect 
to the corresponding Liouville vector field. Compatibility of Poisson tensors is 
responsible for integrability in the Liouville sense, whereas homogeneity is what 
will allow us to project this symplectic picture to a contact Hamiltonian system. These are therefore two 
independent geometric requirements that must be arranged simultaneously in our 
applications.

Homogeneous integrable systems arise naturally as symplectizations of integrable 
contact systems, as recalled below.

\subsection{Contact Hamiltonian systems and integrability}

Let $(C,\alpha)$ be a $(2n+1)$-dimensional co-oriented contact manifold, i.e.
$\alpha\wedge(\dd\alpha)^n$ is a volume form on $C$.
The \emph{Reeb vector field} $R$ is uniquely determined by
\[
\iota_R \dd\alpha = 0,
\qquad
\alpha(R)=1.
\]

Given a function $h\in \Cinfty(C)$, the associated
\emph{contact Hamiltonian vector field} $X_h$ is defined by
(see, e.g., \cite{d.L2019b,G.G2022a,Lopez-Gordon2024})
\[
\begin{cases}
\alpha(X_h)=h,\\[2mm]
\iota_{X_h}\dd\alpha = (\dd h)(R)\,\alpha - \dd h.
\end{cases}
\]
The triple $(C,\alpha,h)$ is called a \emph{contact Hamiltonian system}.


The contact form $\alpha$ induces a Jacobi structure $(\Lambda,R)$ on $C$
\cite{Lichnerowicz1977a,Lichnerowicz1978,d.L.M+2003}, and hence a Jacobi
bracket on $\Cinfty(C)$ given by
\[
\{f,g\}_\alpha
\coloneqq X_f(g) - g\,R(f)
= \Lambda(\dd f,\dd g) + f\,R(g) - g\,R(f),
\qquad
f,g\in \Cinfty(C).
\]

\begin{definition}
Let $(C,\alpha,h)$ be a contact Hamiltonian system.
A function $f\in \Cinfty(C)$ is called a \emph{dissipated quantity}
(also called a \emph{dissipated integral}) if it is in Jacobi-involution with
$h$, i.e.
\[
\{f,h\}_\alpha = 0.
\]
\end{definition}

This notion was introduced in \cite{G.G.M+2020a} and further developed in
\cite{d.L2019a,C.d.L+2023,Lopez-Gordon2024,G.L.R2023} as the natural replacement, in contact
geometry, of conserved quantities in Hamiltonian dynamics.

\begin{definition}\label{def:contact-integrable}
A $(2n+1)$-dimensional contact Hamiltonian system $(C,\alpha,h)$ is said to be
\emph{completely integrable} if there exist $n+1$ dissipated quantities
$f_1,\dots,f_{n+1}$ such that:
\begin{enumerate}
\item $\{f_i,f_j\}_\alpha=0$ for all $i,j$;
\item the differentials $\dd f_1,\dots,\dd f_{n+1}$ have rank at least $n$ on a dense
      open subset of $C$.
\end{enumerate}
Reordering if necessary, we may and shall assume that $h$ itself is included in
the family $\{f_i\}$.
\end{definition}

In this language, the contact Liouville--Arnol'd theorem proved in
\cite{C.d.L+2023} extends the classical symplectic theory
\cite{Arnold1978,Audin2004} to the contact category. In particular, completely
integrable contact systems admit local action--angle type coordinates adapted
to the contact Hamiltonian dynamics and to a certain foliation defined by the dissipated quantities.

\subsection{Symplectization and the contact--homogeneous correspondence}

Let $(C,\alpha)$ be a co-oriented contact manifold. 
Its (trivial) symplectization is the exact symplectic manifold 
$(C^\Sigma,\theta)$ with
\[
C^\Sigma \coloneqq C\times\mathbb{R}_{>0},\qquad
\theta \coloneqq r\,\alpha,
\]
where $r$ is the coordinate on $\mathbb{R}_{>0}$. 
The projection 
\[
\Sigma\colon C^\Sigma\to C,\qquad \Sigma(x,r)=x,
\]
is a symplectization in the sense of 
\cite{D.L.M1991,I.L.M+1997,G.I.M+2004,Petalidou2002,Lopez-Gordon2024}: 
there exists a nowhere-vanishing function $\sigma$ (here $\sigma=r$) such that 
$\theta=\sigma\,\Sigma^\ast\alpha$.

Given a contact Hamiltonian $h\in \Cinfty(C)$, one associates a 
$1$-homogeneous Hamiltonian $H^\Sigma$ on $C^\Sigma$ by
\[
H^\Sigma(x,r) \coloneqq \change{\sigma(x,r)\,h(x) = r\,h(x)},
\]
following the conventions in \cite{C.d.L+2023,Lopez-Gordon2024}. 
This defines a homogeneous Hamiltonian system $(C^\Sigma,\theta,H^\Sigma)$. The Liouville vector field of $\theta$ in this case is given by $\Delta = r \partial_r$.

The following result summarizes the correspondence between integrability in the 
contact and homogeneous symplectic settings (see 
\cite{C.d.L+2023,Lopez-Gordon2024,G.I.M+2004,Petalidou2002} for proofs and variants)

\begin{proposition}[Contact--homogeneous correspondence]
\label{prop:contact-homogeneous-correspondence}
Let $(C,\alpha,h)$ be a contact Hamiltonian system and 
$(C^\Sigma,\theta,H^\Sigma)$ its symplectization as above. Then:
\begin{enumerate}
\item The map $f\mapsto f^\Sigma \coloneqq \change{\sigma\,\Sigma^\ast f}$ defines a bijection
      between $\Cinfty(C)$ and the space of $1$-homogeneous functions on
      $C^\Sigma$.  This bijection has an associated correspondence between the associated Hamiltonian vector fields given by
      $$X_{f^\Sigma} = \change{X_f - r R(f) \partial_r}\, ,$$
      where $X_{f^\Sigma}$ (resp.~$X_f$) denotes the Hamiltonian vector field of $f^\Sigma$  (resp.~$f$)  with respect to $\omega$ (resp.~$\alpha$). 
\item For all $f,g\in \Cinfty(C)$, the Poisson bracket on $(C^\Sigma,\theta)$
      and the Jacobi bracket on $(C,\alpha)$ satisfy
      \[
      \{f^\Sigma,g^\Sigma\}_\theta = \change{\{f,g\}_\alpha^\Sigma}.
      \]
\item A function $f$ is a dissipated quantity for $(C,\alpha,h)$ if and only if
      $f^\Sigma$ is a first integral of $(C^\Sigma,\theta,H^\Sigma)$.
\item The contact system $(C,\alpha,h)$ is completely integrable if and only if
      the homogeneous Hamiltonian system $(C^\Sigma,\theta,H^\Sigma)$ is a
      homogeneous integrable system.
\end{enumerate}
\end{proposition}


\begin{remark}
The previous proposition shows that the study of completely integrable contact 
Hamiltonian systems can be reduced, via symplectization, to the study of 
homogeneous integrable systems on exact symplectic manifolds. In particular,
one may try to import bi-Hamiltonian techniques to the contact setting by first 
constructing compatible Poisson structures on the symplectic realization and 
then projecting the resulting commuting integrals to the contact manifold. 

In contrast, starting directly from compatible Jacobi structures on $(C,\alpha)$
\cite{NunesdaCosta1998a,M.M.P1999,P.N2003,C.D.N2010} and using the associated 
Jacobi--Nijenhuis recursion operators does not lead to a maximal 
family of dissipated quantities in involution for a contact system.
In fact, in \cite{Colombo2025} we have proven a no-go theorem.
This dichotomy will be crucial when we pass from classical homogeneous and 
contact dynamics to quantum GKSL generators in the subsequent sections.
\end{remark}

\begin{example}
\label{ex:linear-contact-correct}

Let $C=\mathbb{R}^3(q,p,z)$ endowed with the standard contact form
\[
\alpha \coloneqq \dd z - p\,\dd q 
.
\]
Its Reeb vector field is $R=\partial_z$.

The contact Hamiltonian vector field $X_h$ is determined by
\[
\alpha(X_h)=h,
\qquad
\iota_{X_h}\dd\alpha=(\dd h)(R)\,\alpha - \dd h .
\]
This yields the coordinate expression
$$X_h = -\frac{\partial h}{\partial p} \frac{\partial }{\partial q}
+\left(\frac{\partial h}{\partial q}+p\frac{\partial h}{\partial z}\right) \frac{\partial }{\partial p} 
+ \left(h - p\frac{\partial h}{\partial p}\right)\frac{\partial }{\partial z}\, .$$
Consider the \emph{linear} contact Hamiltonian
\[
h(q,p,z)\coloneqq z - p .
\]
Then
the resulting contact dynamics is
$$X_h =  \frac{\partial }{\partial q}
+p\frac{\partial }{\partial p} 
+z\frac{\partial }{\partial z}\, ,$$
whose flow $\varphi_t$ is given by
$$\varphi_t\left(q, p, z\right) = \left(q+t, p e^t, z e^t\right)\, .$$

\medskip

Since $R(h)=1$, the condition for a function $I_i$ to be a dissipated quantity is

$$0 = \{h, I_i\} = X_h I_i - (Rh) I_i = X_h I_i - I_i = \frac{\partial I_i}{\partial q}
+p\frac{\partial I_i}{\partial p} 
+z\frac{\partial I_i}{\partial z} - I_i\, .$$

For instance, $I_1=p, I_2 = z$ and $I_3=e^q$ are dissipated quantities. Note that they are not constant along the flow of $X_h$. In fact, $I_i \circ \varphi_t = e^t I_i, \, i=1, 2, 3$. This
illustrates that dissipated quantities (Jacobi integrals) need not be
conserved along the flow, unless additional conditions such as $R(I_i)=0$ are
satisfied.

\medskip

Moreover, $I_0\coloneqq h$ is trivially dissipated since the Jacobi bracket is skew-symmetric.
The differentials
\[
\dd h = \dd z-\dd p,\qquad \dd I_2 = \dd z
\]
are linearly independent on $C$, so
\[
\operatorname{rank}\,\langle \dd h, \dd I_2\rangle =2 \ge n=1 ,
\]
as required for complete contact integrability in dimension $3$.

Therefore $(C,\alpha,h)$ is a completely integrable contact Hamiltonian system
with $n=1$, admitting the two commuting dissipated quantities
$I_0=h$ and $I_2=z$.

\medskip


The associated homogeneous integrable system $(M, \theta, H, F)$ is given by
$$M= C\times\mathbb{R}_{>0}\, , \quad \theta = r \alpha=r\dd z -rp \dd q\, , \quad H = \change{rh =rz-rp}\, , \quad F = \change{rI_1=rp}\, ,$$
where $r$ is the canonical coordinate of $\mathbb{R}_{>0}$. The Liouville vector field is given by $\Delta = r \partial_r$.
The Hamiltonian vector field of $H$ with respect to $\omega = - \dd \theta$ reads
$$X_{H} = \change{X_h - r R(h) \frac{\partial }{\partial r}
=  \frac{\partial }{\partial q}
+p\frac{\partial }{\partial p}
+z\frac{\partial }{\partial z} - r  \frac{\partial }{\partial r}}\, .$$
One can verify that it is $0$-homogeneous, i.e., $[\Delta, X_H]=0$.

\end{example}


\begin{remark}
    In order to simplify the exposition, we will consider only co-oriented contact manifolds and their trivial symplectizations. More generally, one can define a contact distribution on a manifold $M$ as a maximally non-integrable sub-bundle $D\subset \T M$ of corank $1$. A contact form is a one-form $\alpha$ such that $\ker \alpha = D$, which may only exist locally, and is never unique (as $f\alpha$ is also a contact form for any non-vanishing function $f$). 
    A more general notion of symplectization is an $\RR^{\times}$-principal bundle $P\to M$ endowed with a symplectic form which is homogeneous with respect to the principal action. For more details, see \cite{B.G.G2017a,Grabowski2013a,G.G2022a,G.G2024a}
\end{remark}

\section{Bi-Hamiltonian Poisson--Lie systems and contact realizations}
\label{sec:PL-contact}

In this section we explain how bi-Hamiltonian structures on Poisson--Lie groups 
naturally give rise to contact integrable systems via exact symplectic 
realizations and homogeneous Hamiltonians. This provides the classical geometric 
backbone for the Lindblad constructions that will be developed later on.

\subsection{Bi-Hamiltonian structures on Poisson--Lie groups}

Let $G$ be a Lie group with multiplication $m\colon G\times G\to G$. Recall that a smooth map $\Phi\colon (M_1,\pi_1)\to(M_2,\pi_2)$ between Poisson
manifolds is a \emph{Poisson map} if it preserves Poisson brackets, i.e.
\[
\{f\circ\Phi,\, g\circ\Phi\}_{\pi_1}
=
\{f,g\}_{\pi_2}\circ\Phi,
\qquad
\forall\, f,g\in \Cinfty(M_2).
\]

A Poisson tensor $\pi_0\in\Gamma(\wedge^2 \T G)$ is called \emph{multiplicative} 
if the multiplication map is a Poisson map from $(G\times G,\pi_0\oplus\pi_0)$ to 
$(G,\pi_0)$, i.e.
\[
m\colon (G\times G,\pi_0\oplus\pi_0)\longrightarrow (G,\pi_0)
\]
is Poisson. In this case $(G,\pi_0)$ is a \emph{Poisson--Lie group}; see \cite{ballesteros2017poisson} and references therein.

Given a Poisson tensor $\pi$ on $G$, we denote by 
\[
\pi^\sharp\colon \cT  G\longrightarrow \T G,\qquad
\pi^\sharp(\alpha)\coloneqq \pi(\cdot,\alpha),
\]
the associated bundle map. For $H\in \Cinfty(G)$, the Hamiltonian vector field 
$X_H$ is defined by $X_H=\pi^\sharp(\dd H)$.




Many classical integrable systems, such as the Euler top and its deformations, 
admit realizations of this type on Poisson--Lie groups 
\cite{M.P1996,NunesdaCosta1998,Petalidou2000}. In what follows, we shall assume 
that the bi-Hamiltonian dynamics is (Liouville) integrable on a suitable open 
dense subset of $G$, and concentrate on how this structure lifts to an exact 
symplectic manifold and then induces a contact integrable system via homogeneity.

\subsection{Exact symplectic realizations and homogeneous lifts}

Let $(G,\pi_0)$ be a Poisson manifold. A \emph{symplectic realization} of 
$(G,\pi_0)$ is a symplectic manifold $(M,\omega)$ together with a Poisson map 
$J\colon (M,\omega)\to (G,\pi_0)$ whose image is an open subset of $G$; see \cite{C.F.M2021}.
We are interested in the case in which $(M,\omega)$ is \emph{exact}.

\begin{definition}
An \emph{exact symplectic realization} of a Poisson manifold $(G,\pi_0)$ is a 
triple $(M,\theta,J)$ such that:
\begin{enumerate}
\item $\theta\in\Omega^1(M)$ is a one-form with $\omega\coloneqq -\dd\theta$ symplectic;
\item $J\colon (M,\omega)\to(G,\pi_0)$ is a Poisson map with open image, 
      where $(M,\omega)$ is regarded as a Poisson manifold via the Poisson
      bivector $\pi$ induced by $\omega$, i.e.\ $\pi^\sharp=\omega^{-1}$.
\end{enumerate}
The pair $(M,\theta)$ is then an exact symplectic manifold in the sense of 
Section~\ref{sec:bi-ham-contact}.
\end{definition}
\begin{remark}
Here and throughout, when we say that $J\colon (M,\omega)\to(G,\pi_0)$ is a
Poisson map, we mean that $M$ is equipped with the Poisson structure
$\pi=\omega^{-1}$ induced by the symplectic form $\omega$, so that $J$ satisfies
$J_*\pi=\pi_0$.
\end{remark}

\begin{remark}
Requiring a symplectic realization $(M,\omega)$ to be \emph{exact},
$\omega=-\dd\theta$, imposes topological restrictions on the realizing manifold
$M$ (in particular, $M$ cannot be compact and $[\omega]=0$ in de~Rham
cohomology). This requirement, however, does not constrain the underlying
Lie--Poisson or Poisson--Lie structure itself. For a large class of
Lie--Poisson systems, including linear and solvable cases, canonical exact
realizations are provided by cotangent bundles $\cT Q$, 
with the canonical symplectic structure $\omega_Q = - \dd \theta$ defined by the canonical one-form $\theta_Q$ (which is an exact symplectic structure by definition). In that case, the Liouville vector field is the generator of homotheties on the fibers of $\cT Q$.
The example in
Section~\ref{subsec:example-PN-contact} illustrates this situation explicitly.
More generally, since our construction only requires working on open dense
subsets where the homogeneous integrals are independent, the assumption of
exactness should be regarded as a natural geometric convenience rather than a
substantive restriction of the Poisson--Lie framework.
\end{remark}


To connect the Poisson--Lie data on $G$ with contact geometry, we make the 
following structural assumption.

\begin{assumption}
\label{assumption:homogeneous-lift}
Let $(G,\pi_0,\pi_1;H_0,H_1)$ be a bi-Hamiltonian Poisson--Lie system whose 
Hamiltonian flow $X$ is completely integrable (in the Liouville sense) on a dense 
open subset of $G$. We assume that there exist:
\begin{itemize}
\item an exact symplectic realization $(M,\theta,J)$ of $(G,\pi_0)$, with 
      $\dim M = 2(n+1)$;
\item a $1$-homogeneous Hamiltonian $H\in \Cinfty(M)$ with respect to the 
      Liouville vector field $\Delta$ of $(M,\theta)$, namely,
      $$\Delta H = H\, , \quad \contr{\Delta} \omega = - \theta\, ,$$
\end{itemize} such that:
\begin{enumerate}
\item $H_0\circ J = H$;
\item the Hamiltonian vector field $X_H$ on $(M,\omega=-\dd\theta)$ projects to 
      the bi-Hamiltonian vector field $X$ on $G$, i.e.
      \[
      \T J\circ X_H = X\circ J;
      \]
\item there exist $n+1$ independent, $1$-homogeneous first integrals in 
      involution for $(M,\theta,H)$, making $(M,\theta,H)$ a homogeneous 
      integrable system in the sense of Definition~\ref{def:homogeneous_integrable_system}.
\end{enumerate}
\end{assumption}

In many examples of interest, the Poisson pencil $(\pi_0,\pi_1)$ admits a 
symplectic realization $(M,\omega)$ on which the pullbacks of $\pi_0,\pi_1$ give 
rise to a Poisson--Nijenhuis structure $(\omega,N)$ on $M$ 
\cite{K.M1990,D.M.P2011,B.K.M2022,B.K.M2022a}. 
The eigenvalues of $N$ then provide a natural candidate for homogeneous integrals 
in involution, and the homogeneity of $H$ can often be arranged by a suitable 
choice of exact symplectic structure (or by considering an appropriate 
homogeneous extension).

In any case, Assumption~\ref{assumption:homogeneous-lift} encodes precisely the 
two independent geometric ingredients emphasised in the introduction:

\begin{itemize}
\item the bi-Hamiltonian Poisson data on $G$ and its integrability;
\item the existence of a homogeneous Hamiltonian lift on an exact symplectic 
      realization $(M,\theta)$, which will later allow us to pass to contact 
      geometry via a transverse hypersurface.
\end{itemize}

\subsection{Contact realizations from homogeneous lifts}

Under Assumption~\ref{assumption:homogeneous-lift}, we can apply the general 
homogeneous--contact correspondence of Section~\ref{sec:bi-ham-contact} to 
obtain a contact integrable system. We spell this out carefully.

We first recall a standard fact about exact symplectic manifolds (see, e.g., 
\cite{L.M1987,D.L.M1991,Lopez-Gordon2024}).

\begin{lemma}
\label{lem:Liouville-hypersurface-contact}
Let $(M,\theta)$ be an exact symplectic manifold of dimension $2(n+1)$, with 
$\omega=-\dd\theta$ and Liouville vector field $\Delta$ defined by 
$\iota_\Delta\omega=-\theta$. Let $C\subset M$ be a hypersurface transverse to 
$\Delta$. Then the restriction $\alpha\coloneqq \theta|_C$ is a contact form on $C$, 
i.e.
\[
\alpha\wedge (\dd\alpha)^n \neq 0
\]
everywhere on $C$.
\end{lemma}

\begin{proof}
Since $\omega$ is symplectic, the top form $\omega^{n+1}$ is a volume form on $M$.
Using $\iota_\Delta\omega=-\theta$ and the identity
$\iota_\Delta(\omega^{n+1})=(n+1)(\iota_\Delta\omega)\wedge\omega^n$, we obtain
\[
\iota_\Delta(\omega^{n+1})=-(n+1)\,\theta\wedge\omega^n
=-(n+1)(-1)^n\,\theta\wedge(\dd\theta)^n.
\]
Hence $\Theta\coloneqq \theta\wedge(\dd\theta)^n$ is non-vanishing at every point where
$\Delta\neq 0$.

Now let $C\subset M$ be a hypersurface transverse to $\Delta$.
Transversality implies $\Delta_x\neq 0$ and $\T_xM=\mathbb{R}\Delta_x\oplus \T_xC$
for every $x\in C$. Therefore the contraction $\iota_{\Delta}(\omega^{n+1})$
restricts to a nowhere-vanishing $(2n+1)$-form on $C$, and so does $\Theta$.

Finally, letting $\iota:C\hookrightarrow M$ be the inclusion and
$\alpha=\iota^*\theta$, we have
\[
\iota^*\Theta=\iota^*\bigl(\theta\wedge(\dd\theta)^n\bigr)
=\alpha\wedge(\dd\alpha)^n,
\]
which is therefore nowhere vanishing on $C$. This proves that $\alpha$ is a
contact form.
\end{proof}

We now combine this lemma with the homogeneous integrability of $(M,\theta,H)$.

\begin{proposition}
\label{prop:PL-to-contact}
Let $(G,\pi_0,\pi_1;H_0,H_1)$ and $(M,\theta,H)$ be as in 
Assumption~\ref{assumption:homogeneous-lift}. 
Let $\Delta$ be the Liouville vector field of $(M,\theta)$ and let 
$C\subset M$ be a hypersurface transverse to $\Delta$, contained in the open subset 
where $H\neq 0$. Define
\[
\alpha \coloneqq \theta|_C,\qquad h \coloneqq H|_C.
\]
Then $(C,\alpha,h)$ is a $(2n+1)$-dimensional completely integrable contact 
Hamiltonian system in the sense of Definition~\ref{def:contact-integrable}.
\end{proposition}

\begin{proof}
First, by Lemma~\ref{lem:Liouville-hypersurface-contact}, the restriction
$\alpha=\theta|_C$ is a contact form on $C$, so $(C,\alpha)$ is a co-oriented 
contact manifold.

By assumption, $H$ is $1$-homogeneous with respect to $\Delta$, i.e., 
$\mathcal{L}_\Delta H=H$. This 
implies that
\[
\frac{\dd}{\dd s}H(\varphi_s(x)) = (\mathcal{L}_\Delta H)(\varphi_s(x)) 
= H(\varphi_s(x)),
\]
where $\varphi_s$ denotes the flow of $\Delta$.
Hence,
\[
H(\varphi_s(x)) = e^{s} H(x).
\]
In particular, on the open subset where $H\neq 0$, the level sets of $H$ are 
transversal to the Liouville flow, and every Liouville orbit intersects each 
nonzero level set exactly once. By construction $C$ is a hypersurface transverse 
to $\Delta$ and contained in $\{H\neq 0\}$, so every Liouville orbit in the 
region of interest intersects $C$ exactly once.

By Assumption~\ref{assumption:homogeneous-lift}, the homogeneous Hamiltonian 
system $(M,\theta,H)$ is a homogeneous integrable system in the sense of 
Definition~\ref{def:homogeneous_integrable_system}: there exist $n+1$ first integrals
\[
F_0\coloneqq H,\; F_1,\dots,F_n\in \Cinfty(M)
\]
such that:
\begin{enumerate}
\item each $F_j$ is $1$-homogeneous, $\mathcal{L}_\Delta F_j=F_j$;
\item the $F_j$ are in involution with respect to the Poisson bracket induced by 
      $\omega=-\dd\theta$;
\item the differentials $\dd F_0,\dots,\dd F_n$ are independent on a dense open subset 
      $U\subset M$.
\end{enumerate}

Consider the restrictions
\[
I_j \coloneqq F_j|_C\in \Cinfty(C),\qquad j=0,\dots,n.
\]
We claim that the $(n+1)$-tuple $(I_0,\dots,I_n)$ defines a completely integrable 
contact system on $(C,\alpha)$ with contact Hamiltonian $h=I_0$.

\smallskip

\emph{(i) Dissipated quantities and Jacobi-involution.}
Let $X_H$ be the Hamiltonian vector field of $H$ on $(M,\omega)$, and let $X_h$ 
be the contact Hamiltonian vector field of $h$ on $(C,\alpha)$, defined by
\[
\alpha(X_h)=h,\qquad
\iota_{X_h}\dd\alpha = (\dd h)(R)\,\alpha - \dd h,
\]
where $R$ is the Reeb vector field of $(C,\alpha)$. 

Standard results on the contactization of homogeneous Hamiltonian systems
ensure that:
\begin{itemize}
\item Since $H$ is $1$-homogeneous, one has $\mathcal{L}_\Delta H = H$, and hence
      $X_H$ is $0$-homogeneous (equivalently $[\Delta,X_H]=0$). In particular,
      the flow of $X_H$ commutes with the Liouville flow.
\item Because $C$ is transverse to $\Delta$, a neighborhood of $C$ is
      identified with an open subset of $\mathbb{R}\times C$ via the Liouville
      flow. The $\Delta$-invariance of $X_H$ implies that $X_H$ descends to a
      well-defined vector field on the local quotient by $\Delta$, and choosing
      the section $i:C\hookrightarrow M$ yields an induced vector field on $C$.
\item Under this identification, the induced vector field on $C$ coincides with
      the contact Hamiltonian vector field $X_h$ associated with
      $\alpha=\theta|_C$ and $h=H|_C$. Moreover, if $F$ is a $1$-homogeneous
      first integral of $X_H$, then $I\coloneqq F|_C$ satisfies $\{I,h\}_\alpha=0$.
\end{itemize}

This is the same mechanism as in 
Proposition~\ref{prop:contact-homogeneous-correspondence}, but now the contact 
manifold $C$ is realized as a transversal hypersurface to the Liouville flow in 
an arbitrary exact symplectic manifold $(M,\theta)$, rather than as the slice 
$r=1$ of a trivial symplectization.

Applying this to $F_0,\dots,F_n$ we see that each $I_j$ is a dissipated quantity 
for the contact Hamiltonian system $(C,\alpha,h)$, i.e.\ $\{I_j,h\}_\alpha=0$ for 
all $j$.

Moreover, since the $F_j$ are in Poisson-involution on $M$ and the Jacobi 
bracket on $C$ is obtained by reduction along the Liouville flow (see 
\cite{Petalidou2002}), the restrictions 
$I_j$ are in Jacobi-involution:
\[
\{I_j,I_k\}_\alpha = 0,\qquad \forall\, j,k\in\{0,\dots,n\}.
\]

\smallskip

\emph{(ii) Functional independence.}
Let $U\subset M$ be the dense open subset where $\dd F_0,\dots,\dd F_n$ are independent. 
Since $C$ is transverse to the Liouville flow, the intersection $U\cap C$ is a 
dense open subset of $C$. At any $x\in U\cap C$ the differentials 
$\dd F_0,\dots,\dd F_n$ are linearly independent in $\cT_x M$. Restricting to 
$\T_xC\subset \T_xM$ we obtain the differentials $\dd I_0,\dots,\dd I_n$, which remain 
linearly independent on $\T_xC$: otherwise there would be a non-trivial linear 
relation among the $\dd F_j$ on $\T_xM$. Thus,
\[
\operatorname{rank}\langle \dd I_0,\dots,\dd I_n\rangle \ge n+1
\]
on a dense open subset of $C$. In particular, the rank is at least $n$, as 
required in Definition~\ref{def:contact-integrable} of contact integrability.

\smallskip

\emph{(iii) Dimension and completeness.}
By construction $\dim M = 2(n+1)$ and $C$ is a hypersurface, so
$\dim C = 2n+1$. Together with (i) and (ii), this shows that the contact
Hamiltonian system $(C,\alpha,h)$ is a $(2n+1)$-dimensional completely integrable
contact system in the sense of Definition~\ref{def:contact-integrable}.
\end{proof}

In particular, any bi-Hamiltonian Poisson--Lie system satisfying 
Assumption~\ref{assumption:homogeneous-lift} yields, in a canonical way, a 
contact integrable system $(C,\alpha,h)$ equipped with $n+1$ independent 
dissipated quantities in Jacobi-involution, obtained from the bi-Hamiltonian 
hierarchy on $G$ via an exact, homogeneous lift to $(M,\theta,H)$ and restriction 
to a Liouville-transverse hypersurface.

\begin{remark}
\label{rem:no-go-Jacobi}
The construction above crucially uses Poisson compatibility and homogeneity on 
the exact symplectic realization $(M,\theta)$, rather than compatibility of 
Jacobi structures directly on the contact manifold $(C,\alpha)$. As shown in 
our previous work on the Jacobi/bi-Hamiltonian no-go phenomenon 
\cite{C.d.L+2023,Lopez-Gordon2024} (building on 
\cite{NunesdaCosta1998a,M.M.P1999,N.P2002,C.D.N2010,P.N2003,Petalidou2002}), 
compatible Jacobi pairs on a contact manifold do not provide enough independent 
eigenvalues of the associated recursion operator to recover a full set of 
contact integrals. By passing through the homogeneous Poisson picture on 
$(M,\theta)$ we circumvent this obstruction: the bi-Hamiltonian structure lives 
in the Poisson category, while the contact dynamics is obtained by projection 
along the Liouville flow. This is summarized in the following diagram:
$$\begin{tikzcd}[column sep=large,row sep=large]
(M,\theta,H) \arrow[d, "J"'] 
  & C \arrow[l, hook', "i"'] 
      \arrow[dl, dashed, "{(C,\alpha,h;\, I_0,\dots,I_n)}" ] \\
(G,\pi_0,\pi_1;\, H_0,H_1) & 
\end{tikzcd}$$
Thus, we pass from a bi-Hamiltonian Poisson--Lie system on $(G,\pi_0,\pi_1)$ to a
contact integrable system $(C,\alpha,h)$ via an exact, homogeneous symplectic
realization $(M,\theta,H)$ and a Liouville-transverse hypersurface $C$.
\end{remark}


This is precisely the classical geometric framework in which, in the next 
section, we shall construct Lindblad generators whose semiclassical limit 
reproduces the contact flow of $h$ and whose conserved quantum observables 
quantize the classical dissipated quantities $I_0,\dots,I_n$.

\begin{corollary}
\label{cor:PL-contact-integrable}
Let $(G,\pi_0,\pi_1;H_0,H_1)$ be a bi-Hamiltonian Poisson--Lie system satisfying
Assumption~\ref{assumption:homogeneous-lift}. Then there exist a Liouville-transverse
hypersurface $C\subset M$ and a contact form $\alpha = \theta|_C$ such that the induced
contact Hamiltonian system $(C,\alpha,h)$, with $h = H|_C$, is completely integrable
and carries $n+1$ dissipated quantities $I_0\coloneqq h, I_1,\dots,I_n$ in Jacobi-involution.
\end{corollary}

\begin{proof}
This is an immediate consequence of Proposition~\ref{prop:PL-to-contact}, together with
the fact that Liouville-transverse hypersurfaces can be chosen inside any given non-zero
level set of $H$ in the region where the homogeneous integrals $(F_0,\dots,F_n)$ are
independent.
\end{proof}

\subsection{Example: a Poisson--Nijenhuis bi-Hamiltonian system with two non-trivial contact integrals}
\label{subsec:example-PN-contact}


Let $\mathrm{Aff}^+(1)$ denote the connected component of the identity of the 
group of orientation--preserving affine transformations of the real line, with
group law $(a,b)\cdot(a',b')=(aa',\,b+a\,b')$. Its Lie algebra $\mathfrak{aff}(1)$
has basis $(e_1,e_2)$ with $[e_1,e_2]=e_2$. For our purposes it suffices to view
the construction below as a concrete Poisson--Nijenhuis model that may be
realized as coming from a Poisson--Lie pencil on a $2$-dimensional solvable group;
the explicit realization map will not be needed.

Consider the cotangent bundle
\[
M=\cT \mathbb{R}^2\cong\mathbb{R}^4
\]
with coordinates $(x^1,x^2,p_1,p_2)$ equipped with the canonical one-form $\theta=p_i\,\dd x^i$ and the canonical symplectic form $\omega=-\dd\theta = \dd x^i \wedge \dd p_i$. The associated Poisson tensor is
\[
\Lambda
=\frac{\partial}{\partial x^1}\wedge\frac{\partial}{\partial p_1}
+\frac{\partial}{\partial x^2}\wedge\frac{\partial}{\partial p_2}.
\]
Define a second Poisson tensor
\[
\Lambda_1
=p_1\,\frac{\partial}{\partial x^1}\wedge\frac{\partial}{\partial p_1}
+p_2x^2\,\frac{\partial}{\partial x^2}\wedge\frac{\partial}{\partial p_2}.
\]
One checks that $[\Lambda,\Lambda_1]_{\mathrm{SN}}=0$, hence $(\Lambda,\Lambda_1)$
is a Poisson pencil.

On the open dense subset
\[
U\coloneqq M\setminus\big(\{p_2=0\}\cap\{x^2=0\}\big),
\]
the recursion operator
\[
N\coloneqq \sharp_{\Lambda_1}\circ\sharp_{\Lambda}^{-1}
\]
has two eigenvalues
\[
\lambda_1\coloneqq p_1,\qquad \lambda_2\coloneqq p_2x^2,
\]
which are functionally independent on $U$ and pairwise commuting for the two
Poisson structures.

The Liouville vector field reads
\[
\Delta=p_i\frac{\partial}{\partial p_i}\, ,
\]
and $\lambda_1,\lambda_2$ are both
$1$-homogeneous:
\[
\mathcal{L}_\Delta\lambda_i=\lambda_i,\qquad i=1,2.
\]
Define the $1$-homogeneous Hamiltonian
\[
H\coloneqq \lambda_1+\lambda_2=p_1+p_2x^2.
\]
such that
each $\lambda_i$ is a first integral of $X_H$. Thus $(M,\theta,H)$ is a
$1$-homogeneous integrable system in the sense of Definition~\ref{def:homogeneous_integrable_system}.

Moreover, $X_H$ is bi-Hamiltonian:
\[
X_H=\Lambda^\sharp(\dd H)=\Lambda_1^\sharp(\dd H_1),
\]
for instance with
\[
H_1=\log(\lambda_1\lambda_2)
\quad\text{on the region where }\lambda_1\lambda_2>0
\quad(\text{or }H_1=\log|\lambda_1\lambda_2|\text{ locally}).
\]

Introduce coordinates
\[
x^1=q,\qquad x^2=z,\qquad p_1=-rp,\qquad p_2=r,\qquad r>0.
\]
Then
\[
\theta=p_i\,\dd x^i=r\,\dd z-rp\,\dd q=r(\dd z-p\,\dd q),\qquad
H=r(z-p),\qquad
\Delta=r\,\partial_r.
\]
In particular, $(q^i, p_i)$ define coordinates on hypersurfaces which are transversal to the Liouville vector field. 
Let
\[
C\coloneqq \{r=1\}\subset M.
\]
Since $\Delta=r\partial_r$ is transverse to the level sets of $r$, the hypersurface
$C$ is transverse to the Liouville flow. By Lemma~\ref{lem:Liouville-hypersurface-contact}, the restriction
\[
\alpha\coloneqq \theta|_C=\dd z-p\,\dd q
\]
is a contact form. Indeed,
\[
\alpha\wedge \dd\alpha=(\dd z-p\,\dd q)\wedge(-\dd p\wedge \dd q)\neq 0.
\]
Thus $(C,\alpha)$ is a $3$-dimensional co-oriented contact manifold ($n=1$), and
\[
h\coloneqq H|_C=z-p
\]
is the induced contact Hamiltonian.

The eigenvalues restrict to
\[
\bar\lambda_1=\lambda_1|_C=-p,\qquad \bar\lambda_2=\lambda_2|_C=z.
\]
Up to a harmless sign convention we set
\[
I_1\coloneqq p,\qquad I_2\coloneqq z.
\]
By Proposition~\ref{prop:contact-homogeneous-correspondence},
\[
\{I_j,h\}_\alpha=0,\qquad \{I_1,I_2\}_\alpha=0,
\]
so $I_1$ and $I_2$ are dissipated quantities in Jacobi-involution. Since
\[
I_0\coloneqq h=z-p=I_2-I_1,
\]
the family $\{I_0,I_1,I_2\}$ contains exactly $n+1=2$ functionally independent
integrals for $n=1$; indeed,
\[
\operatorname{rank}\langle \dd I_0,\dd I_1,\dd I_2\rangle=2\ge n=1
\]
on a dense open subset (in fact everywhere). Therefore $(C,\alpha,h)$ is a
completely integrable contact system in the sense of Definition~\ref{def:contact-integrable}. As a matter of fact, it is the system from Example~\ref{ex:linear-contact-correct}.
\qedhere

\begin{remark}[The $3$-dimensional case $n=1$]
For a contact manifold of dimension $2n+1=3$, complete contact integrability
requires $n+1=2$ independent dissipated quantities in Jacobi-involution.
One may take the contact Hamiltonian itself as one of them; any additional
Jacobi integrals must then satisfy functional relations, as happens here with
$I_0=I_2-I_1$. What matters is that the rank condition in
Definition~\ref{def:contact-integrable} holds (here with rank $2$ on a dense set). This example fits the
general scheme of Proposition~\ref{prop:PL-to-contact}.
\end{remark}

\section{Lindblad dynamics compatible with contact geometry}
\label{sec:lindblad}

We now move to the quantum setting and introduce a class of Lindblad generators
that are compatible with the contact integrable systems constructed in 
Section~\ref{sec:PL-contact}.
For simplicity's sake,
we work with finite-dimensional Hilbert spaces; in this context,
$\mathcal{B}(\mathcal{H})$ denotes the full $C^\ast$-algebra of bounded
operators on $\mathcal{H}$ (which coincides with the space of all linear
operators), all operator topologies coincide, and domain issues do not arise.
Most of the algebraic definitions and arguments below extend to separable
infinite-dimensional Hilbert spaces under standard technical assumptions on
domains and continuity of the corresponding quantum dynamical semigroups.

\noindent Before proceeding, we briefly recall the algebraic notions that will be used
throughout this subsection:

\begin{itemize}
    \item An \emph{algebra} $\mathcal{A}$ over $\mathbb{C}$ is a complex vector space
    equipped with an associative bilinear product. A \emph{$^\ast$-algebra} is an
    algebra endowed with an involution $A\mapsto A^\dagger$ satisfying
    $(AB)^\dagger = B^\dagger A^\dagger$.

    \item A \emph{subalgebra} $\mathcal{A}\subset\mathcal{B}(\mathcal{H})$ is a linear
    subspace that is closed under products and adjoints. It is called
    \emph{unital} if it contains the identity operator
    $\mathbb{I}_{\mathcal{H}}$, and \emph{commutative} (or abelian) if
    $AB=BA$ for all $A,B\in\mathcal{A}$.

    \item A \emph{$C^\ast$-algebra} is a $^\ast$-algebra $\mathcal{A}$ that is complete
    with respect to a norm $\|\cdot\|$ satisfying the $C^\ast$-identity
    $\|A^\dagger A\|=\|A\|^2$. In finite-dimensional Hilbert spaces, every
    $^\ast$-subalgebra of $\mathcal{B}(\mathcal{H})$ is automatically a
    $C^\ast$-algebra with the operator norm.

    \item In particular, a unital commutative $C^\ast$-subalgebra of
    $\mathcal{B}(\mathcal{H})$ can always be identified, via the spectral theorem,
    with an algebra of functions on a finite classical configuration space. This
    interpretation will be crucial for understanding the emergence of effectively
    classical observables and decoherence in open quantum systems.
\end{itemize}

In the finite-dimensional setting, a state is described by a density matrix
$\rho$. Fixing an orthonormal basis
$\{|j\rangle\}_{j=1}^d$ (often selected by the measurement scheme or by the
spectral decomposition of a reference Hamiltonian), we write
$\rho_{jk}=\langle j|\rho|k\rangle$.
The \emph{coherences} are the off-diagonal entries $\rho_{jk}$ with $j\neq k$,
which encode relative phases and interference effects, while the diagonal
entries $\rho_{jj}$ represent \emph{populations}, i.e., classical
probabilities of outcomes in that basis (in the standard terminology of open
quantum systems, see e.g., \cite{BreuerPetruccione}). A dynamics exhibits \emph{decoherence} (in that basis)
when the off-diagonal entries are damped in time, typically driving $\rho(t)$
towards a state that is approximately diagonal in the chosen basis.
In open quantum systems, decoherence is not an ad hoc ingredient but a robust
consequence of the coupling to an environment: while the total evolution on
system-plus-environment may be unitary, the reduced dynamics on the system
alone is generally non-unitary and may suppress coherences.

There are two equivalent ways to represent the same quantum evolution.
In the \emph{Schr\"odinger picture}, the state evolves according to
\[
\dot{\rho}(t)=\mathcal{L}(\rho(t)),\qquad \rho(0)=\rho_0,
\]
while observables $A_0$ are fixed. Expectation values are then given by
\[
\mathrm{E}(A_0)(t)=\Tr\big(\rho(t)A_0\big),
\]
where $\mathcal{L}$ is the \emph{Gorini--Kossakowski--Sudarshan--Lindblad generator}
(or simply \emph{Lindblad generator}).

In the \emph{Heisenberg picture}, the state remains fixed and observables evolve as
\[
\dot{A}(t)=\mathcal{L}^\dagger(A(t)),\qquad A(0)=A_0,
\]
where $\mathcal{L}^\dagger$ denotes the adjoint of $\mathcal{L}$ with respect to
the trace pairing, i.e.
\[
\Tr\big(\mathcal{L}(\rho)\,A\big)=\Tr\big(\rho\,\mathcal{L}^\dagger(A)\big).
\]
Expectation values are given by
\[
\mathrm{E}(A)(t)=\Tr\big(\rho_0 A(t)\big).
\]

The two descriptions are equivalent by construction, since they yield identical
expectation values $\mathrm{E}(A)(t)$ for all observables and all times.


We will systematically use the Heisenberg picture in what follows.
The reason is that our compatibility conditions with contact geometry are most
naturally expressed as \emph{invariance statements for observables}, namely
\[
\mathcal{L}^\dagger(\widehat{I})=0
\]
for the quantized dissipated quantities $\widehat{I}$, and as a semiclassical
generator limit
\[
\mathcal{L}_\hbar^\dagger(\widehat{f}_\hbar)\;\longrightarrow\;
Q_\hbar(\{f,h\}_\alpha)
\]
on the commutative subalgebra generated by the
functions in involution.
This viewpoint allows one to identify conserved quantities and invariant
subalgebras directly at the level of the observable algebra
$\mathcal{B}(\mathcal{H})$, which is particularly well suited to the geometric
framework developed here.

By contrast, the Schr\"odinger picture is often preferable when addressing
state-preparation problems, entropy production, contractivity properties of the
quantum dynamical semigroup, or the detailed long-time behavior of the density
operator $\rho(t)$; see, for instance,
\cite{Lindblad1976,BreuerPetruccione} and the references therein.
Both pictures are, of course, mathematically equivalent, but emphasize different
structural aspects of open quantum dynamics.

\subsection{GKSL generators and Heisenberg picture}

Let $\mathcal{H}$ be a finite-dimensional Hilbert space and 
$\mathcal{B}(\mathcal{H})$ the $C^\ast$-algebra of linear operators on $\mathcal{H}$.
We denote by $\mathcal{S}(\mathcal{H})$ the convex set of density matrices on
$\mathcal{H}$, i.e.\ positive trace-one operators.

\begin{definition}
A linear map $\mathcal{L}\colon \mathcal{B}(\mathcal{H})\to\mathcal{B}(\mathcal{H})$
is called a \emph{Gorini--Kossakowski--Sudarshan--Lindblad generator} (or 
\emph{Lindblad generator}) in the Schrödinger picture if it can be written,
for all $\rho\in\mathcal{B}(\mathcal{H})$, as
\begin{equation}
\label{eq:GKSL}
\mathcal{L}(\rho)
=
-\frac{i}{\hbar}[H,\rho]
+ \sum_{k=1}^N\Big(
L_k\rho L_k^\dagger - \tfrac12\{L_k^\dagger L_k,\rho\}\Big),
\end{equation}
where $H=H^\dagger$ is self-adjoint (the \emph{Hamiltonian}),
$\{L_k\}_{k=1}^{N}\subset\mathcal{B}(\mathcal{H})$ is a finite family of
bounded operators (the \emph{Lindblad operators}), with $N\in\mathbb{N}$, and $\{\cdot,\cdot\}$ denotes the anticommutator
$\{A,B\}\coloneqq AB+BA$.
\end{definition}

It is well known that such $\mathcal{L}$ are precisely the generators of
one-parameter semigroups $\Phi_t = e^{t\mathcal{L}},\, t\ge 0$, of completely positive and trace-preserving (CPTP) maps on $\mathcal{S}(\mathcal{H})$.
Conversely, any norm-continuous CPTP semigroup on $\mathcal{S}(\mathcal{H})$ has
a generator of the form~\eqref{eq:GKSL} (see \cite{BreuerPetruccione,Lindblad1976} for instance). 

\medskip

We now recall the Heisenberg picture associated with a Lindblad generator.

\begin{definition}
Given a Lindblad generator $\mathcal{L}$ acting on density matrices, the 
\emph{Heisenberg adjoint} $\mathcal{L}^\dagger$ is the linear map on 
$\mathcal{B}(\mathcal{H})$ defined by
\[
\mathrm{Tr}\big( \mathcal{L}(\rho)\,A\big)
= \mathrm{Tr}\big( \rho\,\mathcal{L}^\dagger(A)\big)
\]
for all $\rho\in\mathcal{B}(\mathcal{H})$ and all $A\in\mathcal{B}(\mathcal{H})$.
\end{definition}

Since the trace pairing 
$\langle \rho, A\rangle = \mathrm{Tr}(\rho A)$
is non-degenerate, $\mathcal{L}^\dagger$ is uniquely determined by this relation.

\begin{proposition}
\label{prop:Ldagger-GKSL}
Let $\mathcal{L}$ be of GKSL form~\eqref{eq:GKSL}.
Then its Heisenberg adjoint $\mathcal{L}^\dagger$ is given, for all 
$A\in\mathcal{B}(\mathcal{H})$, by
\begin{equation}
\label{eq:GKSL-Heisenberg}
\mathcal{L}^\dagger(A)
= \frac{i}{\hbar}[H,A]
+ \sum_{k=1}^N\Big(
L_k^\dagger A L_k - \tfrac12\{L_k^\dagger L_k,A\}\Big).
\end{equation}
\end{proposition}

\begin{proof}
Let $A\in\mathcal{B}(\mathcal{H})$ and $\rho\in\mathcal{B}(\mathcal{H})$.
Using the cyclicity of the trace and the definition~\eqref{eq:GKSL}, we compute:
\begin{align*}
\mathrm{Tr}\big(\mathcal{L}(\rho)\,A\big)
&=
-\frac{i}{\hbar}\mathrm{Tr}\big([H,\rho]\,A\big)
+ \sum_{k=1}^N
\mathrm{Tr}\Big(
L_k\rho L_k^\dagger A - \tfrac12\{L_k^\dagger L_k,\rho\}A\Big)
\\[0.5ex]
&=
-\frac{i}{\hbar}\mathrm{Tr}\big(H\rho A - \rho H A\big)
+ \sum_{k=1}^N\Big(
\mathrm{Tr}(L_k\rho L_k^\dagger A)
- \tfrac12\mathrm{Tr}(L_k^\dagger L_k\rho A)
- \tfrac12\mathrm{Tr}(\rho L_k^\dagger L_k A)
\Big).
\end{align*}
By cyclicity of the trace,
\[
\mathrm{Tr}(H\rho A) = \mathrm{Tr}(\rho A H),
\qquad
\mathrm{Tr}(L_k\rho L_k^\dagger A)
= \mathrm{Tr}(\rho L_k^\dagger A L_k),
\qquad
\mathrm{Tr}(L_k^\dagger L_k\rho A)
= \mathrm{Tr}(\rho A L_k^\dagger L_k),
\]
so we obtain
\[
\mathrm{Tr}\big(\mathcal{L}(\rho)\,A\big)
=
\frac{i}{\hbar}\mathrm{Tr}\big(\rho[H,A]\big)
+ \sum_{k=1}^N\Big(
\mathrm{Tr}(\rho L_k^\dagger A L_k)
- \tfrac12\mathrm{Tr}(\rho A L_k^\dagger L_k)
- \tfrac12\mathrm{Tr}(\rho L_k^\dagger L_k A)
\Big).
\]
Factoring out $\rho$ under the trace gives
\[
\mathrm{Tr}\big(\mathcal{L}(\rho)\,A\big)
=
\mathrm{Tr}\Big(\rho\Big[
\frac{i}{\hbar}[H,A]
+ \sum_{k=1}^N\big(
L_k^\dagger A L_k - \tfrac12\{L_k^\dagger L_k,A\}
\big)\Big]\Big),
\]
and, by the defining relation of $\mathcal{L}^\dagger$ and the non-degeneracy of
the trace pairing, this shows~\eqref{eq:GKSL-Heisenberg}.
\end{proof}



We now formalise the notion of \emph{conserved observable} for a given Lindblad
dynamics.


\begin{proposition}\label{prop:expectation-constant}
    Let $\mathcal{L}$ be a Lindblad generator with adjoint $\mathcal{L}^\dagger$, and let $O\in\mathcal{B}(\mathcal{H})$. The following statements are equivalent:
    \begin{enumerate}
        \item $O$ has a constant expectation value, i.e., 
        $$\mathrm{Tr}\big(\rho(t)\,O\big) = \mathrm{Tr}\big(\rho(0)\,O\big)\, ,$$ 
        for
        each solution $\rho(t)$ of $\dot\rho=\mathcal{L}(\rho)$, 
        \item for
        any solution $\rho(t)$ of $\dot\rho=\mathcal{L}(\rho)$ one has
         \[\frac{\dd}{\dd t}\,\mathrm{Tr}\big(\rho(t)\,O\big) = 0\, ,\]
         \item $O\in \ker \mathcal{L}^\dagger$.
    \end{enumerate} 
\end{proposition}

\begin{proof}
The equivalence between the first two statements is obvious. For proving the equivalence between the second and third statements, note that, in the Schrödinger picture,
\[
\frac{\dd}{\dd t}\,\mathrm{Tr}(\rho(t)\,O)
= \mathrm{Tr}(\dot\rho(t)\,O)
= \mathrm{Tr}(\mathcal{L}(\rho(t))\,O)
= \mathrm{Tr}(\rho(t)\,\mathcal{L}^\dagger(O))
\, .
\]
Since the trace product is non-degenerate, the right-hand side vanishes for every $\rho$ if and only if $\mathcal{L}^\dagger(O)=0$.
\end{proof}

\noindent This motivates the following definition:
\begin{definition}
    Let $\mathcal{L}$ be a Lindblad generator and $\mathcal{L}^\dagger$ its
    Heisenberg adjoint. An operator $O\in\mathcal{B}(\mathcal{H})$ is called a
    \emph{(quantum) constant of motion}\footnote{This terminology is standard in the theory of quantum dynamical semigroups,
where such observables are precisely those whose expectation values are
preserved under the evolution for all initial states.} for $\mathcal{L}$ if it satisfies one (and hence all) of the equivalence statements in Proposition~\ref{prop:expectation-constant}. 
\end{definition}

\begin{remark}
    For any Lindblad operator, every matrix proportional to the identity is a constant of the motion, since $\mathbb{I}\in \ker \mathcal{L}^\dagger$ for any $\mathcal{L}^\dagger$ of the form \eqref{eq:GKSL-Heisenberg}. Hence, constants of the motion always form \change{a unital operator system in $\mathcal{B}(\mathcal{H})$ (a closed, self-adjoint linear subspace containing $\mathbb{I}$)}, of dimension at least $1$.
\end{remark}

We now give a sufficient condition ensuring that an observable is
conserved by the Lindblad dynamics.
\begin{lemma}
\label{lem:kernel-condition-contact-compatible}
Let $\mathcal{L}$ be a Lindblad generator in GKSL form~\eqref{eq:GKSL}, and
let $O\in\mathcal{B}(\mathcal{H})$ be self-adjoint. 
Assume that
\[
[H,O]=0,
\qquad
[L_k,O]=0,
\qquad
[L_k^\dagger,O]=0
\quad \text{for all }k.
\]
Then $\mathcal{L}^\dagger(O)=0$, i.e.\ $O$ is a quantum constant of motion.
\end{lemma}

\begin{proof}
Using~\eqref{eq:GKSL-Heisenberg},
\[
\mathcal{L}^\dagger(O)
= \frac{i}{\hbar}[H,O]
+ \sum_{k=1}^N\Big(
L_k^\dagger O L_k - \tfrac12\{L_k^\dagger L_k,O\}\Big).
\]
By hypothesis, $[H,O]=0$, so the Hamiltonian term vanishes.
For each $k$,
\[
L_k^\dagger O L_k
= O L_k^\dagger L_k,
\]
since $[L_k,O]=[L_k^\dagger,O]=0$ implies that $O$ commutes with $L_k^\dagger L_k$.
Similarly,
\[
\{L_k^\dagger L_k,O\}
= L_k^\dagger L_k O + O L_k^\dagger L_k
= 2 O L_k^\dagger L_k.
\]
Therefore
\[
L_k^\dagger O L_k - \tfrac12\{L_k^\dagger L_k,O\}
= O L_k^\dagger L_k - O L_k^\dagger L_k = 0.
\]
Summing over $k$ we obtain $\mathcal{L}^\dagger(O)=0$.
\end{proof}


\begin{remark}
Lemma~\ref{lem:kernel-condition-contact-compatible} shows that any observable
$O$ commuting with the Hamiltonian and with all Lindblad operators (and their
adjoints) is invariant under the Heisenberg evolution.

In particular, let $\{O_0,\dots,O_n\}$ be a family of commuting observables
generating a commutative $C^\ast$-subalgebra
\[
\mathcal{A} \coloneqq C^\ast(O_0,\dots,O_n) \subset \mathcal{B}(\mathcal{H}),
\]
and denote by
\[
\mathcal{A}' \coloneqq \{\,B\in\mathcal{B}(\mathcal{H}) \mid [B,A]=0 \ \text{for all } A\in\mathcal{A}\,\}
\]
its commutant.
If the Hamiltonian $H$ and all Lindblad operators $L_k$ belong to $\mathcal{A}'$,
then every observable $O\in\mathcal{A}$ satisfies
\[
\mathcal{L}^\dagger(O)=0,
\]
and is therefore a constant of motion for the Lindblad dynamics.

This will be our basic mechanism to quantize the classical dissipated quantities
of a contact integrable system: we will construct Lindblad generators whose
Hamiltonian and dissipative parts lie in the commutant of the algebra generated
by the quantum  constants of the motion.
\hfill$\diamond$
\end{remark}

\begin{example}[dephasing Lindblad dynamics with conserved observables]\label{example:Lindblad_dynamics_1}

We conclude this subsection with a simple but paradigmatic example in dimension
$2$, which illustrates how genuine quantum dissipation may coexist with
nontrivial constants of motion. This example will serve as a quantum toy model
for the notion of ``dissipated'' yet conserved quantities that later appear as
quantizations of contact integrals.

Let $\mathcal{H}=\mathbb{C}^2$ and denote by $\sigma_x,\sigma_y,\sigma_z$ the
Pauli matrices, which form a basis of the space of traceless self-adjoint
operators on $\mathcal{H}$. Moreover, $\{\sigma_x,\sigma_y,\sigma_z, \mathbb{I}\}$ is a basis of $\mathcal{B}(\mathcal{H})$.
Physically, $\sigma_z$ is diagonal in the energy
basis and measures the population imbalance between the two energy levels. We consider the Lindblad generator $\mathcal{L}$ defined by
\eqref{eq:GKSL} with
\[
H \coloneqq \frac{\hbar\omega}{2}\,\sigma_z,
\qquad
L \coloneqq \sqrt{\gamma}\,\sigma_z,
\]
where $\omega\in\mathbb{R}$ is the energy splitting and $\gamma>0$ is a
dissipation rate. There is a single Lindblad operator $L_1=L$.

For any density matrix $\rho$, the Schrödinger-picture evolution reads
\[
\mathcal{L}(\rho)
=
-\frac{i}{\hbar}[H,\rho]
+ L\rho L^\dagger - \tfrac12\{L^\dagger L,\rho\}.
\]
Since $\sigma_z^2=\mathbb{I}$, where $\mathbb{I}$ denotes the identity operator
on $\mathcal{H}$, we have
\[
L^\dagger L = \gamma \sigma_z^2 = \gamma \mathbb{I}.
\]
As a consequence, the dissipative part simplifies to
\[
L\rho L^\dagger - \tfrac12\{L^\dagger L,\rho\}
=
\gamma(\sigma_z \rho \sigma_z - \rho),
\]
and therefore
\[
\mathcal{L}(\rho)
=
-\frac{i\omega}{2}[\sigma_z,\rho]
+ \gamma(\sigma_z \rho \sigma_z - \rho).
\]

This dynamics is known as \emph{pure dephasing}: it suppresses the off-diagonal
matrix elements of $\rho$ in the eigenbasis of $\sigma_z$ (i.e.\ quantum
coherences) while leaving the diagonal entries (the populations) unchanged.
Although dissipation is present, no energy exchange with the environment occurs, i.e., $H$ is a constant of the motion.

While dissipation is most naturally described in the Schrödinger picture at the
level of states, the analysis of conserved quantities is more transparent in
the Heisenberg picture. The adjoint generator acts as
\[
\mathcal{L}^\dagger(A)
=
\frac{i}{\hbar}[H,A]
+ L^\dagger A L - \tfrac12\{L^\dagger L,A\}
=
\frac{i\omega}{2}[\sigma_z,A]
+ \gamma(\sigma_z A \sigma_z - A).
\]

We now determine the constants of motion. Take $O=\sigma_z$. Then
\[
[H,\sigma_z]=0,
\qquad
[L,\sigma_z]=\sqrt{\gamma}[\sigma_z,\sigma_z]=0,
\]
so the hypotheses of Lemma~\ref{lem:kernel-condition-contact-compatible} are satisfied and $\mathcal{L}^\dagger(\sigma_z)=0$.
Hence, $\sigma_z$ is a quantum constant of the motion.
Consequently, the \emph{population difference} between the two energy levels is conserved,
even though the system undergoes irreversible dephasing.

In contrast, for $O=\sigma_x$ or $O=\sigma_y$ one has
\[
[H,O]\neq 0,
\qquad
[L,O]\neq 0,
\]
and a direct computation yields
\[
\mathcal{L}^\dagger(\sigma_x)
=
-2\gamma\sigma_x - \omega\sigma_y,
\qquad
\mathcal{L}^\dagger(\sigma_y)
=
-2\gamma\sigma_y + \omega\sigma_x.
\]
Hence neither $\sigma_x$ nor $\sigma_y$ is a constant of motion: their
expectation values decay (and rotate) in time, reflecting the loss of quantum
coherence.

Finally, recall that the identity operator $\mathbb{I}$ always satisfies
$\mathcal{L}^\dagger(\mathbb{I})=0$, expressing trace preservation of the
Lindblad dynamics. Therefore, this system
has a non-trivial subalgebra $\operatorname{span}\{\sigma_z, \mathbb{I}\} \subset \mathcal{B}(\mathcal{H})$ of conserved observables,
despite the presence of
irreversible dissipation.
\end{example}


\subsection{Invariant commutative subalgebras and constants of motion}
\label{subsec:invariant-subalgebras}

We now highlight a basic structural feature
of GKSL dynamics that will play a central role in our construction. In the
following sections, we will require that invariant commutative subalgebras arise
as quantizations of Jacobi-commutative algebras of dissipated quantities
associated with contact integrable systems.

Such invariant commutative subalgebras encode families of observables whose
expectation values are preserved by the open-system dynamics and which therefore
behave as robust, effectively classical quantities in the presence of
dissipation. As we show below, the GKSL structure provides a simple and natural
mechanism ensuring the existence of such subalgebras under elementary
commutation assumptions.


In the Heisenberg picture, the generator $\mathcal{L}^\dagger$ defines a
generally \emph{non-Hamiltonian} evolution on the algebra of observables
$\mathcal{B}(\mathcal{H})$, in the sense that it is not a derivation generated by
a commutator with a self-adjoint operator, but includes dissipative
terms of Lindblad type. As a consequence, the resulting dynamics is not an
automorphism of the observable algebra and may exhibit contraction and loss of
information.


\begin{theorem}
[Invariant commutative subalgebras for GKSL dynamics]
\label{thm:invariant-abelian-subalgebra}
Let $\mathcal{H}$ be a finite-dimensional Hilbert space and let
$\mathcal{L}$ be a Lindblad generator on $\mathcal{B}(\mathcal{H})$ of the form
\eqref{eq:GKSL}, with Heisenberg adjoint $\mathcal{L}^\dagger$ given by
\eqref{eq:GKSL-Heisenberg}.
Let $\mathcal{A}\subset\mathcal{B}(\mathcal{H})$ be a unital, commutative
$C^\ast$-subalgebra. Suppose that
\[
[H,A]=0,
\qquad
[L_k,A]=[L_k^\dagger,A]=0
\quad
\text{for all }A\in\mathcal{A}\text{ and all }k.
\]
Then, $\mathcal{A}\subseteq \ker \mathcal{L}^\dagger$ and every observable $A\in\mathcal{A}$ is a quantum constant of motion. In particular, $\mathcal{A}$ is invariant under the Heisenberg evolution generated by $\mathcal{L}^\dagger$, i.e. $\mathcal{L}^\dagger(\mathcal{A})\subset\mathcal{A}$.
\end{theorem}
\begin{proof}
Let $A\in\mathcal{A}$. By assumption $[H,A]=0$, so the Hamiltonian contribution
to $\mathcal{L}^\dagger(A)$ vanishes. For each Lindblad operator $L_k$, the
commutation relations imply
\[
L_k^\dagger A L_k = A L_k^\dagger L_k,
\qquad
\{L_k^\dagger L_k,A\} = 2 A L_k^\dagger L_k.
\]
Hence each dissipative term satisfies
\[
L_k^\dagger A L_k - \tfrac12\{L_k^\dagger L_k,A\} = 0.
\]
Summing over $k$ yields $\mathcal{L}^\dagger(A)=0$.
Thus, $A$ is a constant of the motion.
Since $\mathcal{A}$ is a
linear subspace, this shows both invariance and pointwise vanishing of
$\mathcal{L}^\dagger$ on $\mathcal{A}$.
\end{proof}

\begin{remark}
\label{rem:coherence-decoherence}

The existence of an invariant commutative $C^\ast$-subalgebra for a Lindblad
generator admits a natural physical interpretation in terms of coherence and
decoherence.  

In the Schrödinger picture, decoherence refers to the irreversible suppression
of off-diagonal matrix elements of the density operator in certain preferred
bases, typically induced by the dissipative part of the GKSL generator.
Equivalently, in the Heisenberg picture, decoherence manifests itself as the
contraction of the algebra of observables toward a smaller, dynamically stable
subalgebra.  

Invariant commutative subalgebras play a distinguished role in this process.
Observables belonging to such a subalgebra have constant expectation values for
all initial states and are insensitive to the dissipative dynamics. In this
sense, they behave as robust, effectively classical quantities selected by the
open-system evolution.  

From a geometric viewpoint, this selection mechanism parallels the emergence of
dissipated quantities in contact Hamiltonian systems: although the full dynamics
is non-Hamiltonian and generically dissipative, certain observables retain a
privileged status encoded by algebraic invariance conditions. This analogy will
be made precise in the following sections by requiring that the invariant
quantum observables coincide, in the semiclassical limit, with the dissipated
quantities of a completely integrable contact system.



For the purposes of this work, the Heisenberg picture is particularly natural,
since conserved observables are precisely characterized as elements of
$\ker\mathcal{L}^\dagger$, and invariant commutative subalgebras can be directly
identified. 
\hfill$\diamond$
\end{remark}

\begin{example}[Invariant commutative subalgebras in dephasing Lindblad dynamics]

The dephasing example above (Example~\ref{example:Lindblad_dynamics_1}) can be
reinterpreted in purely algebraic terms as the invariance of a commutative
subalgebra under the Heisenberg evolution generated by a GKSL operator. This
perspective is particularly useful for understanding the role of conserved
observables in open quantum systems and anticipates the contact-compatible
constructions developed later.

Consider again the dephasing Lindblad generator on $\mathcal{H}=\mathbb{C}^2$
defined by
\[
\mathcal{L}^\dagger(A)
=
\frac{i\omega}{2}[\sigma_z,A]
+ \gamma(\sigma_z A \sigma_z - A),
\qquad \omega\in\mathbb{R},\ \gamma>0.
\]
Let
\[
\mathcal{A}
\coloneqq 
\mathrm{span}\{\mathbb{I},\sigma_z\}
\;\subset\;
\mathcal{B}(\mathcal{H})
\]
be the unital $^\ast$-subalgebra generated by $\sigma_z$. This is a commutative
$C^\ast$-subalgebra, canonically isomorphic to $\mathbb{C}^2$, corresponding to
the algebra of functions on a two-point classical configuration space.

A direct computation shows that $\mathcal{A}$ is invariant under
$\mathcal{L}^\dagger$. Indeed,
\[
\mathcal{L}^\dagger(\mathbb{I})=0,
\qquad
\mathcal{L}^\dagger(\sigma_z)=0,
\]
and therefore
\[
\mathcal{L}^\dagger(A)=0
\quad\text{for all }A\in\mathcal{A}.
\]
Equivalently,
\[
\mathcal{L}^\dagger(\mathcal{A}) \subset \mathcal{A},
\qquad
\mathcal{L}^\dagger|_{\mathcal{A}} = 0.
\]

In contrast, the complementary subspace
\[
\mathcal{A}^\perp \coloneqq \mathrm{span}\{\sigma_x,\sigma_y\}
\]
is not invariant. Using the Pauli-matrix identities
\[
[\sigma_z,\sigma_x]=2i\sigma_y,\qquad
[\sigma_z,\sigma_y]=-2i\sigma_x,\qquad
\sigma_z\sigma_x\sigma_z=-\sigma_x,\qquad
\sigma_z\sigma_y\sigma_z=-\sigma_y,
\]
one finds
\[
\mathcal{L}^\dagger(\sigma_x)
= -\omega\sigma_y - 2\gamma\sigma_x,
\qquad
\mathcal{L}^\dagger(\sigma_y)
= \omega\sigma_x - 2\gamma\sigma_y.
\]
Thus, writing a generic observable in $\mathcal{A}^\perp$ as
\(
A(t)=a_x(t)\sigma_x+a_y(t)\sigma_y,
\)
the Heisenberg equation $\dot A(t)=\mathcal{L}^\dagger(A(t))$ reduces to the
linear system
\[
\binom{\dot a_x}{\dot a_y}
=
\begin{pmatrix}
-2\gamma & \ \omega\\
-\omega & -2\gamma
\end{pmatrix}
\binom{a_x}{a_y}.
\]
The eigenvalues of this matrix are $\lambda_\pm=-2\gamma\pm i\omega$, and the
solution is explicitly
\[
\binom{a_x(t)}{a_y(t)}
=
e^{-2\gamma t}
\begin{pmatrix}
\cos(\omega t) & \ \sin(\omega t)\\
-\sin(\omega t) & \ \cos(\omega t)
\end{pmatrix}
\binom{a_x(0)}{a_y(0)}.
\]
Equivalently, the evolved observables are
\[
\sigma_x(t)
=
e^{-2\gamma t}\big(\cos(\omega t)\,\sigma_x-\sin(\omega t)\,\sigma_y\big),
\qquad
\sigma_y(t)
=
e^{-2\gamma t}\big(\sin(\omega t)\,\sigma_x+\cos(\omega t)\,\sigma_y\big).
\]


For any fixed initial state $\rho_0$, the corresponding expectation values
satisfy
\begin{align*}
\mathrm{E}(\sigma_x)(t)
&\coloneqq \mathrm{Tr}(\rho_0\,\sigma_x(t))
=
e^{-2\gamma t}\Big(
\cos(\omega t)\mathrm{E}(\sigma_x)(0)
-\sin(\omega t)\mathrm{E}(\sigma_y)(0)
\Big),\\
\mathrm{E}(\sigma_y)(t)
&=
e^{-2\gamma t}\Big(
\sin(\omega t)\mathrm{E}(\sigma_x)(0)
+\cos(\omega t)\mathrm{E}(\sigma_y)(0)
\Big).
\end{align*}
In particular,
\[
\mathrm{E}(\sigma_x)(t)^2+\mathrm{E}(\sigma_y)(t)^2
=
e^{-4\gamma t}
\Big(
\mathrm{E}(\sigma_x)(0)^2+\mathrm{E}(\sigma_y)(0)^2
\Big),
\]
showing an exponential decay of the coherences with rate $2\gamma$.

\medskip

The statement that ``$\mathcal{A}^\perp$ is contracted by the dissipative part''
should be understood at the level of the induced semigroup
$e^{t\mathcal{L}^\dagger}$, not in terms of a negative operator norm.
Indeed, the restriction of $\mathcal{L}^\dagger$ to $\mathcal{A}^\perp$ has
spectrum $\{-2\gamma\pm i\omega\}$, and hence
\[
\|e^{t\mathcal{L}^\dagger}|_{\mathcal{A}^\perp}\|
\;\le\; C\,e^{-2\gamma t}
\]
for some $C>0$. This exponential bound is the precise sense in which the
dissipative dynamics suppresses the non-commutative (coherence) sector while
leaving $\mathcal{A}$ pointwise invariant. \medskip
Here, the statement that ``$\mathcal A$ is pointwise invariant'' means that the
Heisenberg generator vanishes identically on $\mathcal A$, namely
\[
\mathcal L^\dagger(A)=0
\qquad
\text{for all } A\in\mathcal A.
\]
Equivalently, every observable in $\mathcal A$ is a fixed point of the quantum
dynamical semigroup generated by $\mathcal L^\dagger$, that is,
\[
e^{t\mathcal L^\dagger}(A)=A
\qquad
\text{for all } t\ge 0,\ A\in\mathcal A.
\]
In particular,
\[
\mathcal A \subset \ker \mathcal L^\dagger,
\]
which is a stronger property than mere invariance of $\mathcal A$ as a
subspace.

\end{example}


From an algebraic viewpoint, the dephasing dynamics induces a dynamical
splitting
\[
\mathcal{B}(\mathcal{H})
=
\mathcal{A}
\;\oplus\;
\mathcal{A}^\perp,
\]
where $\mathcal{A}$ is a maximal commutative $C^\ast$-subalgebra that is invariant
under the Heisenberg evolution generated by $\mathcal{L}^\dagger$, and
$\mathcal{A}^\perp$ is its orthogonal complement with respect to the
Hilbert--Schmidt inner product.

The key point is that $\mathcal{L}^\dagger$ acts trivially on $\mathcal{A}$,
\[
\mathcal{L}^\dagger|_{\mathcal{A}} = 0,
\]
while it generates a strict contraction semigroup on $\mathcal{A}^\perp$.
Consequently, for any observable $A\in\mathcal{B}(\mathcal{H})$ one has
\[
A(t)=e^{t\mathcal{L}^\dagger}(A)
\;\longrightarrow\;
\Pi_{\mathcal{A}}(A)
\qquad (t\to+\infty),
\]
where $\Pi_{\mathcal{A}}$ denotes the conditional expectation (orthogonal
projection) onto $\mathcal{A}$.

Thus, in the long-time limit, the Heisenberg evolution effectively reduces the
observable algebra to $\mathcal{A}$: all components orthogonal to $\mathcal{A}$
are dynamically suppressed, while observables in $\mathcal{A}$ remain unchanged.

In the present example, $\mathcal{A}=C^\ast(\sigma_z)$, so the asymptotic
observable algebra consists precisely of functions of $\sigma_z$.  This means
that the Lindblad dynamics dynamically selects $\sigma_z$ as the only
nontrivial observable whose expectation values remain stable under the
open-system evolution.

More generally, this example illustrates a robust structural phenomenon: for
suitable GKSL generators, there exist non-trivial commutative $C^\ast$-subalgebras
$\mathcal{A}\subset\mathcal{B}(\mathcal{H})$, with
$\mathcal{A}\neq\operatorname{span}(\mathbb{I})$, that are invariant under the
Heisenberg evolution and lie entirely in $\ker\mathcal{L}^\dagger$.

In the geometric framework developed in this paper, such invariant commutative
subalgebras provide the natural quantum counterparts of the Jacobi-commutative
algebras generated by dissipated quantities in contact integrable systems.

\begin{remark}
The terms ``projection onto a classical observable algebra'' and
``selection of a preferred observable'' are to be understood in the precise
algebraic sense described above: the dynamics induces a contraction of the full
observable algebra toward an invariant commutative subalgebra, without invoking
additional postulates beyond the GKSL evolution itself.\hfill$\diamond$
\end{remark}

\subsection{Contact-compatible Lindblad generators}
\label{subsec:contact-compatible}

The goal of this subsection is to introduce a notion of Lindblad dynamics that
is compatible, in a precise semiclassical sense, with a given completely
integrable contact Hamiltonian system. More precisely, we want to formalize
when a quantum open-system evolution preserves, at the level of expectation
values and in the semiclassical limit, the distinguished family of
\emph{dissipated quantities} associated with contact integrability.

Let \((C,\alpha,h)\) be a completely integrable contact Hamiltonian system with
dissipated quantities \(I_0\coloneqq h,I_1,\dots,I_n\) as in
Definition~\ref{def:contact-integrable} and
Proposition~\ref{prop:PL-to-contact}. Throughout this subsection we assume that
\((C,\alpha,h;I_0,\dots,I_n)\) arises from a bi-Hamiltonian Poisson--Lie system as
described in Section~\ref{sec:PL-contact}.

\medskip


Consider the algebra $(\Cinfty(C), \{\cdot,\cdot\}_\alpha)$, where $\{\cdot,\cdot\}_\alpha$ denotes the Jacobi bracket defined by $\alpha$. Let $\mathcal{A}_{\mathrm{cl}}\subset \Cinfty(C)$ denote the commutative
subalgebra generated by \(\{I_0,\dots,I_n\}\). We refer to $\mathcal{A}_{\mathrm{cl}}$ as the
\emph{Jacobi commutative algebra} of the integrable system. This terminology reflects the
fact that $\mathcal{A}_{\mathrm{cl}}$ plays, for contact dynamics, the same role
as the Poisson-commutative algebra of integrals in Liouville-integrable
Hamiltonian systems: it encodes the distinguished observables whose mutual
relations are governed by the underlying Jacobi structure.

\medskip

We shall work with an abstract semiclassical quantization of the contact data.
The reason for doing so is twofold. First, the geometric constructions leading
to contact integrable systems are largely independent of any particular choice
of quantization scheme. Second, our aim is to identify structural conditions on
Lindblad generators that are robust under different realizations
(e.g.\ Weyl quantization, Toeplitz quantization, or finite-dimensional
representations), rather than tied to a specific model.

\bigskip

\begin{definition}[Semiclassical quantization of the Jacobi commutative algebra]
\label{def:semiclassical-quantization}
A \emph{semiclassical quantization} $(\mathcal{H}_\hbar, Q_\hbar)$ of \((C,\alpha,h;I_0,\dots,I_n)\) consists of a family of Hilbert spaces \(\{\mathcal{H}_\hbar\}_{\hbar\in(0,\hbar_0]}\) and a family of linear maps $Q_\hbar \colon \mathcal{A}_{\mathrm{cl}} \longrightarrow
\mathcal{B}(\mathcal{H}_\hbar)$,
 such that for all \(f,g\in\mathcal{A}_{\mathrm{cl}}\) one has:
\begin{enumerate}
\item (\emph{Reality})
      \[
      Q_\hbar(\overline{f}) = Q_\hbar(f)^\dagger.
      \]
      In particular, real-valued classical observables are mapped to
      self-adjoint operators.

\item (\emph{Normalization})
      \[
      Q_\hbar(1)=\mathbf{1}_{\mathcal{H}_\hbar}.
      \]
      Here \(1\) denotes the constant function on \(C\), while
      \(\mathbf{1}_{\mathcal{H}_\hbar}\) is the identity operator on the Hilbert
      space \(\mathcal{H}_\hbar\). This condition ensures that the quantization
      preserves units and expectation values of constant observables.

\item (\emph{Dirac condition for the Jacobi bracket})
      \[
      \lim_{\hbar\to 0}\left\|
      \frac{1}{i\hbar}\big[Q_\hbar(f),Q_\hbar(g)\big]
      - Q_\hbar\big(\{f,g\}_\alpha\big)
      \right\| = 0.
      \]
      The factor \(1/(i\hbar)\) reflects the standard correspondence principle:
      commutators of quantum observables rescaled by \(1/(i\hbar)\) converge, in
      the semiclassical limit, to the classical Jacobi bracket. This condition is
      the natural contact analog of the usual Dirac condition in symplectic
      quantization.
\end{enumerate}

For brevity we set
\[
\widehat{f}_\hbar \coloneqq Q_\hbar(f),\qquad
\widehat{I}_{j,\hbar} \coloneqq Q_\hbar(I_j),
\qquad
\widehat{h}_\hbar \coloneqq Q_\hbar(h) = \widehat{I}_{0,\hbar}\, , \quad \forall\, f \in \Cinfty(C)\, , \quad \forall\,  I_j \in \mathcal{A}_{\mathrm{cl}}\, .
\]
\end{definition}

In what follows, we fix a semiclassical quantization
\((\mathcal{H}_\hbar,Q_\hbar)\) of \((C,\alpha,h;I_0,\dots,I_n)\). For each \(\hbar\) we consider
Lindblad generators $\mathcal{L}_\hbar \colon \mathcal{B}(\mathcal{H}_\hbar)\to
\mathcal{B}(\mathcal{H}_\hbar)$ in GKSL form,
\[
\mathcal{L}_\hbar(\rho)
=
-\frac{i}{\hbar}\big[\widehat{H}_\hbar,\rho\big]
+ \sum_{k=1}^{N(\hbar)} \Big(
L_{k,\hbar}\,\rho\, L_{k,\hbar}^\dagger
- \tfrac12\{L_{k,\hbar}^\dagger L_{k,\hbar},\rho\}
\Big).
\]
Here \(N(\hbar)\in\mathbb{N}\) denotes the (possibly \(\hbar\)-dependent) number
of Lindblad operators. Allowing \(N(\hbar)\) to vary reflects the fact that the
effective dissipative description of an open quantum system may depend on the
semiclassical scale. We denote by \(\mathcal{L}_\hbar^\dagger\) the Heisenberg adjoint of
\(\mathcal{L}_\hbar\) (see Proposition~\ref{prop:Ldagger-GKSL}).

\begin{remark}
\label{rem:symplectize-quantize}

A contact system is naturally encoded by its symplectization (see Proposition~\ref{prop:contact-homogeneous-correspondence}).  
In particular, the contact Hamiltonian dynamics on $C$ is the reduction (or,
equivalently, the restriction to a transverse section $r=\mathrm{const}$) of the
homogeneous Hamiltonian dynamics on $M$.

In the present work, the semiclassical quantization map $Q_\hbar$ on the Jacobi
commutative algebra $\mathcal{A}_{\mathrm{cl}}$ is precisely obtained by
\emph{symplectizing first} and then applying a semiclassical quantization on $M$
to the homogeneous lifts:
\[
Q_\hbar(f) \;\coloneqq \; \Op_\hbar(f^\Sigma), \qquad f\in \mathcal{A}_{\mathrm{cl}}.
\]
With this choice, ``symplectization'' and ``quantization'' commute at the level
relevant for our compatibility conditions: the Dirac/Jacobi correspondence and
the Egorov-type limit are inherited from the standard semiclassical calculus on
$(M,\omega)$ (up to the usual $o(1)$ or $\mathcal{O}(\hbar)$ remainders).

It is also worth mentioning that the underlying contact dynamics depends only on the conformal class of
the contact form.  If $f$ is nowhere vanishing and $\alpha'\coloneqq \alpha/f$, then the
rescaled data $(C,\alpha',\,fI_0,\dots,fI_n)$ describe the same integrable
geometry (same invariant foliations and contact vector fields, $X_{I_i}^\alpha = X_{f I_i}^{\alpha'}$).  In the symplectization picture this
corresponds to a change of the radial normalization/section, and therefore one
expects the resulting quantizations to be equivalent (up to unitary
equivalence and/or time reparametrization within the same semiclassical class).
This ``gauge'' freedom is the classical counterpart of the inevitable unital
sector in GKSL dynamics, since $\mathcal{L}_\hbar^\dagger(\mathbf 1)=0$ holds for
every Lindblad generator.
\hfill$\diamond$
\end{remark}

We can now formulate the central compatibility notion.

\begin{definition}[Contact-compatible Lindblad generator]
\label{def:contact-compatible-Lindblad}
Let \((C,\alpha,h;I_0,\dots,I_n)\) and \((\mathcal{H}_\hbar,Q_\hbar)\) be as above.
A family of Lindblad generators
\(\{\mathcal{L}_\hbar\}_{\hbar\in(0,\hbar_0]}\) is said to be
\emph{contact-compatible} if:

\begin{enumerate}[label={(CC\arabic*)}]
\item \label{assumption_CC1} (\emph{Quantum constants of motion})
For each \(j=0,\dots,n\),
\[
\mathcal{L}_\hbar^\dagger\big(\widehat{I}_{j,\hbar}\big) = 0.
\]

\item \label{assumption_CC2} (\emph{Semiclassical contact limit})
For every \(f\in\mathcal{A}_{\mathrm{cl}}\),
\[
\lim_{\hbar\to 0}\left\|
\mathcal{L}_\hbar^\dagger\big(\widehat{f}_\hbar\big)
- Q_\hbar\big(\{f,h\}_\alpha\big)
\right\| = 0.
\]
\end{enumerate}
\end{definition}

\begin{remark}
\label{rem:weak-Egorov}
Condition {\rm\ref{assumption_CC2}}, together with the Dirac condition in
Definition~\ref{def:semiclassical-quantization}, can be viewed as a weak
Egorov-type property for contact dynamics: the Heisenberg evolution of quantum
observables converges, in the semiclassical limit, to the classical Jacobi
evolution generated by \(h\). This provides the physical interpretation of
contact-compatible Lindblad generators as quantum open-system evolutions whose
classical shadow is the contact Hamiltonian flow.
\hfill$\diamond$
\end{remark}

\begin{remark}
A conceptual difference between the classical contact setting and its quantum
counterpart concerns the role of the unit observable.

On the classical side, the Jacobi algebra $(\Cinfty(C),\{\cdot,\cdot\}_\alpha)$
is not unital in the Poisson sense: although the constant function $1$ belongs
to the algebra, it is not central in general, since
\[
\{1,f\}_\alpha = R(f),
\]
where $R$ is the Reeb vector field. Consequently, the inclusion of $1$ in a
Jacobi-commutative subalgebra $\mathcal{A}_{\mathrm{cl}}$ imposes a genuine
geometric condition, namely the invariance of all generators under the Reeb
flow.

By contrast, in the quantum setting the identity operator $\mathbf{1}$ is
automatically present and invariant under any GKSL evolution:
\[
\mathcal{L}^\dagger(\mathbf{1})=0
\]
holds identically for every Lindblad generator. As a result, the quantum notion
of an invariant unital commutative $C^\ast$-subalgebra does not impose any
additional constraint associated with the unit.

For this reason, the correct correspondence between classical and quantum
structures should be understood as follows: the quantization map $Q_\hbar$
assigns operators to the \emph{nonconstant} generators of the Jacobi-commutative
algebra $\mathcal{A}_{\mathrm{cl}}$, while the unit is added automatically on
the quantum side. The non-trivial content of contact compatibility and of quantum
constants of motion therefore lies entirely in the quantization of the
dissipated quantities, not in the unit element.

This distinction explains why, in minimal examples, the classical condition
$\{1,I\}_\alpha=0$ has no direct quantum analog, and clarifies the precise
sense in which invariant commutative $C^\ast$-subalgebras quantize
Jacobi-commutative algebras of dissipated quantities.
\hfill$\diamond$
\end{remark}

\paragraph{Example: minimal contact-compatible algebra.}


Let $\mathcal{A}_{\mathrm{cl}}=\mathrm{span}\{I\}$ be the Jacobi-commutative
subalgebra generated by a single dissipated quantity $I$.
Any semiclassical quantization $Q_\hbar$ maps the non-trivial generator $I$ to a
self-adjoint operator $\widehat{I}_\hbar$. The corresponding quantum algebra is
then taken to be the unital $C^\ast$-algebra
\[
\mathcal{A}_\hbar \coloneqq \mathrm{span}\{\mathbf{1},\widehat{I}_\hbar\},
\]
where the identity operator is included automatically.

A contact-compatible Lindblad generator is precisely one whose Heisenberg
adjoint leaves $\mathcal{A}_\hbar$ pointwise invariant and satisfies
$\mathcal{L}^\dagger(\widehat{I}_\hbar)=0$.
In this minimal setting, contact compatibility reduces to the existence of a non-trivial invariant commutative subalgebra generated by a single dissipated quantity beyond the scalars. This example highlights a general structural feature: while the constant function
$1$ plays a non-trivial role in the Jacobi algebra at the classical level, its
quantum counterpart is automatically invariant. The content of contact
compatibility therefore lies entirely in the quantization of the non-constant
dissipated quantities.

\medskip 

For a set of functions $S\subseteq \Cinfty(C)$, we write $\mathrm{Jacobi\text{-}alg}\langle S\rangle$ for the smallest Jacobi subalgebra of $(\Cinfty(C),\{\cdot,\cdot\}_\alpha)$
containing $S$, i.e. the smallest linear
subspace of $\Cinfty(C)$ which contains $S$ and is closed under pointwise multiplication and the Jacobi bracket.

The following corollary shows that such Jacobi-commutative algebras of
dissipated quantities admit a natural quantization in terms of invariant
commutative $C^\ast$-subalgebras of quantum observables.

\begin{corollary}[Quantization of Jacobi-commutative algebras by invariant quantum observables]
\label{cor:quantization-jacobi-algebra}


Let $(C,\alpha)$ be a contact manifold with Jacobi bracket
$\{\cdot,\cdot\}_\alpha$, and let
\[
\mathcal{A}_{\mathrm{cl}}
=
\mathrm{Jacobi\text{-}alg}\langle I_1,\dots,I_n\rangle
\subset \Cinfty(C)
\]
be a Jacobi-commutative algebra of dissipated quantities, i.e.
\[
\{I_i,I_j\}_\alpha = 0
\quad \text{for all } i,j,
\qquad
\{h,I_i\}_\alpha = 0
\]
for some contact Hamiltonian $h$.

Assume that $Q_\hbar$ is a semiclassical quantization of
$\mathcal{A}_{\mathrm{cl}}$ in the sense of
Definition~\ref{def:semiclassical-quantization}, and that there exists a family
of Lindblad generators $\{\mathcal{L}_\hbar\}$ such that
\begin{enumerate}
\item the Heisenberg adjoint $\mathcal{L}_\hbar^\dagger$ satisfies the Egorov-type
      condition
      \[
      \mathcal{L}_\hbar^\dagger(\widehat{f}_\hbar)
      =
      Q_\hbar(\{f,h\}_\alpha)
      + o(1)
      \qquad (\hbar\to 0)
      \]
      for all $f\in\mathcal{A}_{\mathrm{cl}}$;
\item the Lindblad data are chosen inside the commutant of the $C^\ast$-algebra
      generated by the quantized integrals,
      \[
      [\widehat{H}_\hbar,\widehat{I}_{i,\hbar}]=0,
      \qquad
      [L_{k,\hbar},\widehat{I}_{i,\hbar}]
      =
      [L_{k,\hbar}^\dagger,\widehat{I}_{i,\hbar}]=0 .
      \]
\end{enumerate}

Then the operators $\widehat{I}_{i,\hbar}\coloneqq Q_\hbar(I_i)$ generate a unital,
commutative, invariant $C^\ast$-subalgebra
\[
\mathcal{A}_\hbar
=
C^\ast(\widehat{I}_{1,\hbar},\dots,\widehat{I}_{n,\hbar})
\subset \mathcal{B}(\mathcal{H}_\hbar),
\]
and every $A\in\mathcal{A}_\hbar$ is a quantum constant of motion:
\[
\mathcal{L}_\hbar^\dagger(A)=0 .
\]
\end{corollary}

\bigskip
\noindent
In this sense, invariant commutative $C^\ast$-subalgebras of observables provide
the natural quantum counterparts of Jacobi-commutative algebras of dissipated
quantities.


\begin{proof}
Since the generators $I_i$ Jacobi-commute, the Dirac condition for the
quantization map $Q_\hbar$ implies that their quantizations commute modulo
$\mathcal{O}(\hbar)$:
\[
[\widehat{I}_{i,\hbar},\widehat{I}_{j,\hbar}]
=
\mathcal{O}(\hbar).
\]

Assumption (2) implies that the Hamiltonian and all Lindblad operators (and
their adjoints) belong to the commutant $\mathcal{A}_\hbar'$. By
Lemma~\ref{lem:kernel-condition-contact-compatible}, it follows that
\[
\mathcal{L}_\hbar^\dagger(\widehat{I}_{i,\hbar}) = 0
\quad \text{for all } i,
\]
and hence, by linearity and norm continuity of $\mathcal{L}_\hbar^\dagger$,
\[
\mathcal{L}_\hbar^\dagger(A)=0
\quad \text{for all } A\in\mathcal{A}_\hbar.
\] Therefore, the operators $\widehat{I}_{i,\hbar}$ generate
a commutative $C^\ast$-subalgebra $\mathcal{A}_\hbar \subset \mathcal{B}(\mathcal{H}_\hbar)$, which is unital \change{since $\mathbf{1}_{\mathcal{H}_\hbar} = Q_\hbar(1) \in \mathcal{A}_\hbar$ by the normalization condition~(2) of Definition~\ref{def:semiclassical-quantization}}.

Finally, the Egorov-type condition (1) ensures that, on the level of principal
symbols, the Heisenberg generator reduces to the contact Hamiltonian derivation
associated with $h$ when restricted to $\mathcal{A}_{\mathrm{cl}}$. Therefore,
$\mathcal{A}_\hbar$ can be interpreted as a quantization of the Jacobi-commutative
algebra $\mathcal{A}_{\mathrm{cl}}$.\end{proof}

\begin{remark}
Corollary~\ref{cor:quantization-jacobi-algebra} states that the $C^*$-algebra obtained by quantizing the algebra of functions in involution with respect to the Jacobi bracket is \change{unital} by construction. This \change{unitality} should be understood as a structural consequence of the quantization and $C^*$-completion procedure, rather than as a property that must be verified explicitly on the level of concrete generators.

More precisely, the quantization yields a $*$-algebra $\mathcal{A}_0$, equipped with an involution $^*$ corresponding to the adjoint operation at the operator level. The associated $C^*$-algebra $\mathcal{A}$ is obtained as the completion of $\mathcal{A}_0$ with respect to the $C^*$-norm
\[
\|A\| = \sup_{\pi} \|\pi(A)\|,
\]
where $\pi$ ranges over all continuous $*$-representations on Hilbert spaces. 
By construction, the resulting $C^\ast$-algebra is realized as a unital
$C^\ast$-subalgebra of $\mathcal{B}(\mathcal{H}_\hbar)$, and therefore contains the identity operator $\mathbf{1}$ as well as a well-defined group of unitary elements.

This structural \change{unitality} plays a crucial role in the subsequent Lindblad description. Indeed, the generators entering the Lindblad equation are required to be bounded operators on a Hilbert space, and the dissipative dynamics is formulated in terms of completely positive, trace-preserving semigroups acting on the underlying $C^*$-algebra of observables. The $C^*$-algebraic framework therefore provides the natural setting in which the Lindblad operators, the Hamiltonian part, and the associated quantum dynamical semigroup are consistently defined.

In this sense, the \change{unitality} of the quantized algebra is not an auxiliary assumption, but a necessary ingredient ensuring that the Lindblad evolution constructed later is mathematically well-defined. The example of the Pauli matrices illustrates this mechanism in a finite-dimensional setting, but the same reasoning applies in full generality.

This remark should be read in conjunction with the discussion preceding
Corollary~\ref{cor:quantization-jacobi-algebra}, where the distinct roles of the
unit in the classical Jacobi algebra and in its quantum counterpart are
clarified.

\end{remark}

\subsection{Bi-Lindblad structures}
\label{subsec:bi-Lindblad}

The purpose of this subsection is to introduce a quantum analog of
bi-Hamiltonian pencils: instead of a Poisson pencil
\(\pi_\lambda=\pi_1-\lambda\pi_0\) (closed under affine combination), we work
with a \emph{convex pencil} of GKSL generators
\(\mathcal{L}^{(\lambda)}_\hbar=(1-\lambda)\mathcal{L}^{(0)}_\hbar+\lambda\mathcal{L}^{(1)}_\hbar\),
and we require that the \emph{same} family of quantized integrals
\(\widehat I_{j,\hbar}\) lies in the kernel of both Heisenberg adjoints.
This is the structural mechanism by which a single classical bi-Hamiltonian
hierarchy produces two (or more) compatible open quantum evolutions sharing a
common invariant commutative sector.

We now formulate a quantum counterpart of bi-Hamiltonian structures, in which
the role of compatible Poisson tensors \((\pi_0,\pi_1)\) is played by a pair of
Lindblad generators whose convex combinations remain Lindbladian and whose
common quantum integrals quantize the classical dissipated quantities
\(I_j\).

Let \((C,\alpha,h;I_0,\dots,I_n)\) and \((\mathcal{H}_\hbar,Q_\hbar)\) be as in
Definition~\ref{def:semiclassical-quantization}.  Assume in addition that
the contact system arises from a bi-Hamiltonian Poisson--Lie system
\((G,\pi_0,\pi_1;H_0,H_1)\) as in Section~\ref{sec:PL-contact}, so that the
classical vector field \(X\) is Hamiltonian with respect to both \(\pi_0\) and \(\pi_1\),
and the hierarchy of integrals \(\{I_j\}\) ultimately comes from the common
bi-Hamiltonian invariants on \(G\).

At the quantum level, the two compatible Poisson--Lie Hamiltonians
$H_0,H_1$ do not give rise directly to Lindblad generators, but rather induce,
via the homogeneous lift and restriction procedure of
Section~\ref{sec:PL-contact}, a pair of contact Hamiltonians on $C$.
We denote by $H_0^{\mathrm{eff}},\; H_1^{\mathrm{eff}} \in \mathcal{A}_{\mathrm{cl}}$ the corresponding effective contact Hamiltonians governing the classical
dynamics on $C$ associated with the two Poisson structures.

\begin{definition}[Bi-Lindblad structure]
\label{def:bi-Lindblad}
A \emph{bi-Lindblad structure} associated with the contact integrable system
\((C,\alpha,h;I_0,\dots,I_n)\) and its semiclassical quantization
\((\mathcal{H}_\hbar,Q_\hbar)\) consists of a pair of families of Lindblad
generators $\{\mathcal{L}^{(0)}_\hbar\}_{\hbar\in(0,\hbar_0]}$ and $\{\mathcal{L}^{(1)}_\hbar\}_{\hbar\in(0,\hbar_0]}$ on \(\mathcal{B}(\mathcal{H}_\hbar)\) such that

\begin{enumerate}[label={(BL\arabic*)}]
\item\label{assumption_BL1} (\emph{Convex compatibility})
For every \(\lambda\in[0,1]\) and every \(\hbar\in(0,\hbar_0]\), the convex
combination
\[
\mathcal{L}^{(\lambda)}_\hbar
\coloneqq (1-\lambda)\,\mathcal{L}^{(0)}_\hbar
   + \lambda\,\mathcal{L}^{(1)}_\hbar
\]
is again a GKSL generator on \(\mathcal{B}(\mathcal{H}_\hbar)\).  In other
words, the convex cone generated by \(\mathcal{L}^{(0)}_\hbar\) and
\(\mathcal{L}^{(1)}_\hbar\) is contained in the space of Lindblad generators.

\item\label{assumption_BL2} (\emph{Common quantum integrals})
For each \(j=0,\dots,n\) and each \(\hbar\in(0,\hbar_0]\) one has
\[
\mathcal{L}^{(0)\,\dagger}_\hbar\big(\widehat{I}_{j,\hbar}\big)
=
\mathcal{L}^{(1)\,\dagger}_\hbar\big(\widehat{I}_{j,\hbar}\big)
= 0.
\]
Equivalently, the quantized dissipated quantities \(\widehat{I}_{j,\hbar}\) lie
in the joint kernel
\[
\ker\mathcal{L}^{(0)\,\dagger}_\hbar \cap \ker\mathcal{L}^{(1)\,\dagger}_\hbar.
\]

\item\label{assumption_BL3} (\emph{Semiclassical bi-Hamiltonian limit})
For every \(f\in\mathcal{A}_{\mathrm{cl}}\) and \(a\in\{0,1\}\),
\[
\lim_{\hbar\to 0}\left\|
\mathcal{L}^{(a)\,\dagger}_\hbar\big(\widehat{f}_\hbar\big)
- Q_\hbar\big(\{f,H_a^{\mathrm{eff}}\}_\alpha\big)
\right\| = 0,
\]
where \(H_0^{\mathrm{eff}},H_1^{\mathrm{eff}}\in\mathcal{A}_{\mathrm{cl}}\) are
effective contact Hamiltonians whose classical flows on \(C\) are obtained from
the underlying bi-Hamiltonian Poisson--Lie data \((G,\pi_0,\pi_1;H_0,H_1)\)
via the homogeneous lift and restriction of Section~\ref{sec:PL-contact}.
\end{enumerate}


\end{definition}

Heuristically, {\rm\ref{assumption_BL1}} mirrors the classical requirement that any linear
combination \(\pi_\lambda = \pi_1 - \lambda \pi_0\) of compatible Poisson
bivectors is still a Poisson tensor: here convex combinations of the two
generators remain physically admissible Lindblad generators. Condition {\rm\ref{assumption_BL2}}
captures the fact that the same family of classical integrals \(I_j\) arises
from the bi-Hamiltonian hierarchy and should therefore be conserved by both
quantum evolutions. Finally, {\rm\ref{assumption_BL3}} encodes the requirement that, in the
semiclassical limit, each branch \(\mathcal{L}^{(a)}_\hbar\) recovers the
contact dynamics associated with a different classical Hamiltonian \(H_a\), in
analogy with the way a bi-Hamiltonian vector field is Hamiltonian with respect to both
\(\pi_0\) and \(\pi_1\).

\begin{remark}
Condition {\rm\ref{assumption_BL1}} is natural from the open-systems viewpoint: the set of GKSL
generators forms a convex cone (sums and positive scalar multiples remain of
GKSL type). Thus, in many concrete constructions, {\rm\ref{assumption_BL1}} is essentially a
bookkeeping requirement ensuring that two admissible dissipative mechanisms can
be blended continuously without leaving the Lindblad class.
\hfill$\diamond$
\end{remark}

The following simple observation shows that any bi-Lindblad structure in the
sense of Definition~\ref{def:bi-Lindblad} automatically yields a
contact-compatible Lindblad generator for each \(\lambda\in[0,1]\).

\begin{proposition}
\label{prop:bi-Lindblad-implies-contact-compatible}
Let \(\{\mathcal{L}^{(0)}_\hbar\}\) and \(\{\mathcal{L}^{(1)}_\hbar\}\) define a
bi-Lindblad structure in the sense of Definition~\ref{def:bi-Lindblad}, and fix
\(\lambda\in[0,1]\). Then, for each \(\hbar\in(0,\hbar_0]\), the generator
\[
\mathcal{L}^{(\lambda)}_\hbar
= (1-\lambda)\,\mathcal{L}^{(0)}_\hbar
  + \lambda\,\mathcal{L}^{(1)}_\hbar
\]
is contact-compatible with \((C,\alpha,h;I_0,\dots,I_n)\) in the sense of
Definition~\ref{def:contact-compatible-Lindblad}, with respect to the effective
Hamiltonian
\[
H_\lambda^{\mathrm{eff}} \coloneqq (1-\lambda)\,H_0^{\mathrm{eff}} + \lambda\,H_1^{\mathrm{eff}}
\in \mathcal{A}_{\mathrm{cl}}.
\]
\end{proposition}

\begin{proof}
By {\rm\ref{assumption_BL1}}, \(\mathcal{L}^{(\lambda)}_\hbar\) is a GKSL generator for each
\(\lambda\in[0,1]\) and \(\hbar\in(0,\hbar_0]\).

\smallskip

\noindent\emph{Verification of {\rm\ref{assumption_CC1}}.}
Since the Heisenberg adjoint is linear,
\[
\mathcal{L}^{(\lambda)\,\dagger}_\hbar
= (1-\lambda)\,\mathcal{L}^{(0)\,\dagger}_\hbar
  + \lambda\,\mathcal{L}^{(1)\,\dagger}_\hbar,
\]
and by {\rm\ref{assumption_BL2}} we have \(\mathcal{L}^{(a)\,\dagger}_\hbar(\widehat I_{j,\hbar})=0\) for
\(a\in\{0,1\}\), it follows that
\[
\mathcal{L}^{(\lambda)\,\dagger}_\hbar(\widehat I_{j,\hbar})=0
\qquad\text{for all }j=0,\dots,n.
\]
Thus {\rm\ref{assumption_CC1}} holds for \(\mathcal{L}^{(\lambda)}_\hbar\).

\smallskip

\noindent\emph{Verification of {\rm\ref{assumption_CC2}}.}
Fix \(f\in\mathcal{A}_{\mathrm{cl}}\). Using bilinearity of the Jacobi bracket and
linearity of \(Q_\hbar\), we have
\[
Q_\hbar\!\big(\{f,H_\lambda^{\mathrm{eff}}\}_\alpha\big)
=
(1-\lambda)\,Q_\hbar\!\big(\{f,H_0^{\mathrm{eff}}\}_\alpha\big)
+
\lambda\,Q_\hbar\!\big(\{f,H_1^{\mathrm{eff}}\}_\alpha\big).
\]
Hence, by the triangle inequality,
\begin{align*}
&\Big\|
\mathcal{L}^{(\lambda)\,\dagger}_\hbar(\widehat f_\hbar)
- Q_\hbar\!\big(\{f,H_\lambda^{\mathrm{eff}}\}_\alpha\big)
\Big\|
\\
&\le
(1-\lambda)\,
\Big\|
\mathcal{L}^{(0)\,\dagger}_\hbar(\widehat f_\hbar)
- Q_\hbar\!\big(\{f,H_0^{\mathrm{eff}}\}_\alpha\big)
\Big\|
+
\lambda\,
\Big\|
\mathcal{L}^{(1)\,\dagger}_\hbar(\widehat f_\hbar)
- Q_\hbar\!\big(\{f,H_1^{\mathrm{eff}}\}_\alpha\big)
\Big\|.
\end{align*}
Each term tends to \(0\) as \(\hbar\to0\) by {\rm\ref{assumption_BL3}}, so the left-hand side
also tends to \(0\). This is precisely {\rm\ref{assumption_CC2}} for \(\mathcal{L}^{(\lambda)}_\hbar\)
with effective Hamiltonian \(H_\lambda^{\mathrm{eff}}\).
\end{proof}

\paragraph{Example: a bi-Lindblad pencil interpolating between unitary and dephasing dynamics.}
We illustrate Definition~\ref{def:bi-Lindblad} on a two-level system
\(\mathcal{H}=\mathbb{C}^2\), where the invariant commutative sector is generated
by \(\sigma_z\).  Fix parameters $\omega_0,\omega_1\in\mathbb{R}$ and $\gamma>0$, and define two
GKSL generators on $\mathcal{B}(\mathbb{C}^2)$ (not forming a semiclassical
family in $\hbar$, but sufficient to illustrate the algebraic content of
Definition~\ref{def:bi-Lindblad}):
\[
\mathcal{L}^{(0)}(\rho)
\coloneqq -\frac{i}{\hbar}\big[H_0,\rho\big],
\qquad
H_0\coloneqq \frac{\hbar\omega_0}{2}\sigma_z,
\]
and
\[
\mathcal{L}^{(1)}(\rho)
\coloneqq -\frac{i}{\hbar}\big[H_1,\rho\big]
+ L\rho L^\dagger-\tfrac12\{L^\dagger L,\rho\},
\qquad
H_1\coloneqq \frac{\hbar\omega_1}{2}\sigma_z,
\quad
L\coloneqq \sqrt{\gamma}\,\sigma_z\, .
\]
Note that $\mathcal{L}^{(0)}$ describes a purely unitary evolution, while $\mathcal{L}^{(1)}$ corresponds to a unitary evolution plus pure dephasing in the \(\sigma_z\) basis.
For every
\(\lambda\in[0,1]\), the operator
\[
\mathcal{L}^{(\lambda)}\coloneqq (1-\lambda)\mathcal{L}^{(0)}+\lambda\mathcal{L}^{(1)}
\]
is again of GKSL form, with Hamiltonian
\(H_\lambda=\frac{\hbar}{2}\big((1-\lambda)\omega_0+\lambda\omega_1\big)\sigma_z\)
and a single Lindblad operator \(\sqrt{\lambda}\,L\); hence {\rm\ref{assumption_BL1}} holds.

In the Heisenberg picture, the adjoints satisfy
\[
\mathcal{L}^{(a)\,\dagger}(\sigma_z)=0,
\qquad a\in\{0,1\},
\]
since \([H_a,\sigma_z]=0\) and \([L,\sigma_z]=0\). Therefore \(\sigma_z\) (and
trivially \(\mathbb{I}\)) lies in the joint kernel
\(\ker\mathcal{L}^{(0)\,\dagger}\cap\ker\mathcal{L}^{(1)\,\dagger}\), i.e.\
{\rm\ref{assumption_BL2}} holds for the commutative algebra \(C^\ast(\sigma_z)\).

This example captures the basic geometric idea behind bi-Lindblad structures:
two physically distinct open-system mechanisms (here ``purely Hamiltonian''
versus ``Hamiltonian + dephasing'') share a common invariant commutative sector,
and convex interpolation preserves complete positivity and the same quantum
constants of motion. In our later semiclassical constructions, the role of
\(\sigma_z\) will be played by the quantized dissipated quantities
\(\widehat{I}_{j,\hbar}\), and the two branches \(a=0,1\) will correspond to two
different effective classical Hamiltonians in the underlying bi-Hamiltonian
hierarchy.

\medskip

Thus a bi-Lindblad structure provides, for each \(\lambda\in[0,1]\), a specific
choice of contact-compatible Lindblad generator whose semiclassical limit
reflects a different point on the classical bi-Hamiltonian pencil.  In the Euler-top-type example of Section~\ref{sec:euler-example}, we shall see
explicitly how the quantum generators \(\mathcal{L}^{(0)}_\hbar\) and
\(\mathcal{L}^{(1)}_\hbar\) can be chosen so that their common quantum integrals
\(\widehat{I}_{j,\hbar}\) quantize the contact dissipated quantities
\(I_j\) obtained in Section~\ref{sec:PL-contact}, thereby realizing a concrete
bi-Lindblad structure in the above sense.

\medskip

\paragraph{Weyl quantization.} Before stating the Egorov-type criterion, we briefly recall what we mean by
\emph{semiclassical Weyl quantization} in the standard pseudodifferential setting.
Assume that the exact symplectic manifold $(M,\omega)$ is globally symplectomorphic
to $\mathbb{R}^{2d}$ (with $d=n+1$) with canonical coordinates $(x,\xi)\in\mathbb{R}^\dd\times\mathbb{R}^d$
and symplectic form $\displaystyle{\omega=\sum_{j=1}^d \dd\xi_j\wedge \dd x_j}$.

In the following, $S^0(M)$ will denote the standard class of semiclassical symbols of order zero
on $M\simeq\mathbb{R}^{2d}$, that is, the space of smooth functions
$a(x,\xi;\hbar)$ satisfying
\[
|\partial_x^\alpha \partial_\xi^\beta a(x,\xi;\hbar)|
\le C_{\alpha\beta}
\qquad
\text{for all multi-indices } \alpha,\beta,
\]
uniformly in $(x,\xi)\in\mathbb{R}^{2d}$ and $\hbar\in(0,\hbar_0]$. In addition, $\mathcal{S}(\mathbb{R}^d)$ will denote the Schwartz space of rapidly decreasing
smooth functions on $\mathbb{R}^d$, that is, the space of all functions
$\psi\in \Cinfty(\mathbb{R}^d)$ such that
\[
\sup_{x\in\mathbb{R}^d}
|x^\alpha \partial_x^\beta \psi(x)| < \infty
\qquad
\text{for all multi-indices } \alpha,\beta.
\]

Given a (bounded) semiclassical symbol $a\in S^0(M)$, the Weyl quantization
$\Op_\hbar^{\mathrm{w}}(a)$ is the pseudodifferential operator acting on
$L^2(\mathbb{R}^d)$ by
\[
\big(\Op_\hbar^{\mathrm{w}}(a)\psi\big)(x)
\coloneqq 
\frac{1}{(2\pi\hbar)^d}
\int_{\mathbb{R}^d}\int_{\mathbb{R}^d}
a\!\left(\frac{x+y}{2},\xi\right)\,
e^{\frac{i}{\hbar}(x-y)\cdot\xi}\,\psi(y)\,\dd y\,\dd\xi,
\qquad \psi\in\mathcal{S}(\mathbb{R}^d),
\]
and extends by density to a bounded operator on $L^2(\mathbb{R}^d)$ when $a\in S^0$.

The defining feature of Weyl quantization is its symmetric placement of the phase-space
variables, which ensures (i) good reality properties (real-valued symbols yield
essentially self-adjoint operators) and (ii) a clean semiclassical symbolic calculus.

In particular, the semiclassical Moyal product $\sharp$ satisfies
\[
\Op_\hbar^{\mathrm{w}}(a)\,\Op_\hbar^{\mathrm{w}}(b)
=\Op_\hbar^{\mathrm{w}}(a\sharp b),
\qquad
a\sharp b = ab + \frac{\hbar}{2i}\{a,b\}_\omega + \mathcal{O}(\hbar^2),
\]
and consequently the commutator admits the expansion
\[
\frac{1}{i\hbar}\big[\Op_\hbar^{\mathrm{w}}(a),\Op_\hbar^{\mathrm{w}}(b)\big]
=
\Op_\hbar^{\mathrm{w}}\big(\{a,b\}_\omega\big) + \mathcal{O}(\hbar^2),
\]
in operator norm for $a,b\in S^0(M)$ under standard symbol assumptions.
This is the pseudodifferential backbone behind the Dirac condition and the Egorov-type
limit used below. We refer to \cite{Zworski2012,DimassiSjoestrand1999} for comprehensive accounts
of semiclassical Weyl quantization, symbol classes $S^m$, the Moyal product, and
the Egorov theorem. See also \cite{Folland1989,Martinez2002} for the Weyl calculus and its
microlocal properties. From the perspective of deformation quantization, the Moyal product and the
Dirac condition are discussed for instance in \cite{Gutt2000}.

In our contact setting, the quantization map $Q_\hbar$ on the Jacobi subalgebra
$\mathcal{A}_{\mathrm{cl}}\subset \Cinfty(C)$ is obtained by first taking the
homogeneous lift $f\mapsto f^\Sigma$ (so $f^\Sigma|_C=f$) and then applying Weyl
quantization:
\[
Q_\hbar(f)=\Op_\hbar^{\mathrm{w}}(f^\Sigma),
\qquad f\in\mathcal{A}_{\mathrm{cl}}.
\]
This realizes the contact Dirac condition in Definition~\ref{def:semiclassical-quantization}
as a direct consequence of the Weyl commutator expansion and the Poisson--Jacobi
correspondence of Proposition~\ref{prop:contact-homogeneous-correspondence}.

\paragraph{Lindblad operators of Weyl type.}

In Theorem~\ref{prop:egorov-criterion} we assume that the Lindblad operators
$L_{k,\hbar}$ are \emph{of Weyl type}.  This means that each dissipative operator
admits a semiclassical pseudodifferential representation whose leading symbol is
a well-defined classical observable on the exact symplectic manifold $(M,\dd \theta)$. More precisely, saying that
\[
L_{k,\hbar} = \Op_\hbar^{\mathrm{w}}(\ell_k) + \mathcal{O}(\hbar)
\quad\text{in operator norm}
\]
means that:
\begin{itemize}
\item $\ell_k\in S^0(M)$ is a bounded classical symbol (the \emph{principal symbol}
      of $L_{k,\hbar}$);
\item the remainder $\mathcal O(\hbar)$ is controlled in operator norm, in the
      sense that
      \[
      \big\|L_{k,\hbar}-\Op_\hbar^{\mathrm{w}}(\ell_k)\big\|
      \le C\,\hbar \quad \text{for all } k\le N(\hbar),\ \hbar\in(0,\hbar_0]\]
      with $C$ independent of $k$ and $\hbar$, ensuring a uniform semiclassical expansion compatible with the
      pseudodifferential calculus;
\item the leading-order action of the dissipator can be analyzed directly at the
      level of classical phase space via Poisson brackets.
\end{itemize}

This assumption is the natural open-system analog of taking the Hamiltonian
$\widehat{H}_\hbar$ to be a Weyl quantization of a classical Hamiltonian function.
It ensures that dissipation is not introduced in an ad hoc operator-theoretic
way, but rather descends from classical data on $M$ and admits a meaningful
semiclassical interpretation.

In particular, when the symbols $\ell_k$ are chosen as functions of the classical
integrals $(F_0,\dots,F_n)$, the corresponding Lindblad operators act (approximately) diagonally, up to $\mathcal{O}(\hbar)$ corrections,
on the joint eigenspaces of the quantized integrals
$\widehat{I}_{j,\hbar}=\Op_\hbar^{\mathrm{w}}(I_j^\Sigma)$.
This is precisely the mechanism behind \emph{pure dephasing}: the dissipator
suppresses coherences between different eigenvalues of the integrals, while
leaving their expectation values invariant.

Thus, the Weyl-type assumption on the Lindblad operators is the technical bridge
that allows one to:
\begin{enumerate}
\item control the semiclassical limit of the dissipative dynamics;
\item prove the Egorov-type statement {\rm\ref{assumption_CC2}};
\item interpret decoherence geometrically in terms of the classical integrals
      and their associated invariant foliations.
\end{enumerate}

\begin{theorem}[Pseudodifferential Egorov criterion for {\rm\ref{assumption_CC2}}]
\label{prop:egorov-criterion}
Assume that the contact integrable system $(C,\alpha,h;I_0,\dots,I_n)$ arises
from a homogeneous Hamiltonian system $(M,\theta,H)$ as in
Proposition~\ref{prop:PL-to-contact}, with $C$ realized as a Liouville-transverse
hypersurface in the exact symplectic manifold $(M,\theta)$.

Suppose that:
\begin{enumerate}
\item $M$ is globally symplectomorphic to $\mathbb{R}^{2(n+1)}$ (or to a
      cotangent bundle $\cT  Q$) and the quantization maps $Q_\hbar$ in
      Definition~\ref{def:semiclassical-quantization} are obtained by restricting
      a semiclassical Weyl quantization
      \[
      \Op_\hbar^{\mathrm{w}}\colon S^0(M)\longrightarrow\mathcal{B}(\mathcal{H}_\hbar)
      \]
      of order-zero symbols on $M$ to the subalgebra
      $\mathcal{A}_{\mathrm{cl}}\subset \Cinfty(C)$ via the homogeneous lift
      $f\mapsto f^\Sigma$ of Proposition~\ref{prop:contact-homogeneous-correspondence}.

\item The Hamiltonian operators $\widehat{H}_\hbar$ are of the form
      \[
      \widehat{H}_\hbar = \Op_\hbar^{\mathrm{w}}(H) + R_\hbar\, , \quad \|R_\hbar\| = \mathcal{O}(\hbar^2)\, ,
      \]
      where $H$ is the $1$-homogeneous classical Hamiltonian
      on $M$ whose restriction to $C$ gives the contact Hamiltonian $h$.

\item The Lindblad operators are of Weyl type,
      \[
      L_{k,\hbar} = \Op_\hbar^{\mathrm{w}}(\ell_k) + \mathcal{O}(\hbar),
      \qquad k=1,\dots,N(\hbar),
      \]
      with symbols $\ell_k\in S^0(M)$ such that for each $k$ and every
$f\in\mathcal{A}_{\mathrm{cl}}$ one has
\[
\{\ell_k,f^\Sigma\}_\omega = 0
\qquad\text{and}\qquad
\{| \ell_k |^2,f^\Sigma\}_\omega = 0
\quad\text{on } M,
\] and the number of Lindblad operators satisfies $N(\hbar)=\mathcal{O}(1)$ as
$\hbar\to 0$.
     
\end{enumerate}
Then the family of Lindblad generators $\{\mathcal{L}_\hbar\}$ satisfies {\rm\ref{assumption_CC2}}
with respect to the contact Hamiltonian $h$, i.e.\ for every $f\in\mathcal{A}_{\mathrm{cl}}$,
\[
\lim_{\hbar\to 0}\Big\|
\mathcal{L}_\hbar^\dagger\big(\widehat{f}_\hbar\big)
- Q_\hbar\big(\{f,h\}_\alpha\big)
\Big\| = 0.
\]
Moreover, if in addition the Lindblad data are chosen so that
\begin{equation}
\label{eq:exact-commutant-choice}
[\widehat{H}_\hbar,\widehat{I}_{j,\hbar}]=0,
\qquad
[L_{k,\hbar},\widehat{I}_{j,\hbar}]=[L_{k,\hbar}^\dagger,\widehat{I}_{j,\hbar}]=0
\quad\text{for all }j,k,
\end{equation}
then {\rm\ref{assumption_CC1}} holds as well, i.e.\ $\mathcal{L}_\hbar^\dagger(\widehat{I}_{j,\hbar})=0$
for $j=0,\dots,n$.
\end{theorem}

\begin{proof}
We first prove {\rm\ref{assumption_CC2}} (Steps~1--4), and then discuss {\rm\ref{assumption_CC1}} (Step~5).

\smallskip

\noindent\emph{Step 1: Dirac condition for $Q_\hbar$ and basic properties.}
By hypothesis, $(M,\omega)$ is symplectomorphic to $\mathbb{R}^{2(n+1)}$ (or a
cotangent bundle $\cT  Q$), and $Q_\hbar$ is obtained by composing the
homogeneous lift $f\mapsto f^\Sigma$ of
Proposition~\ref{prop:contact-homogeneous-correspondence} with a semiclassical
Weyl quantization
$\Op_\hbar^{\mathrm{w}}\colon S^0(M)\longrightarrow\mathcal{B}(\mathcal{H}_\hbar)$. Thus, for $f\in\mathcal{A}_{\mathrm{cl}}$,
\[
\widehat{f}_\hbar \coloneqq Q_\hbar(f) = \Op_\hbar^{\mathrm{w}}(f^\Sigma),
\]
and the standard semiclassical symbolic calculus gives, for $a,b\in S^0(M)$,
\[
\frac{1}{i\hbar}\big[\Op_\hbar^{\mathrm{w}}(a),\Op_\hbar^{\mathrm{w}}(b)\big]
= \Op_\hbar^{\mathrm{w}}\big(\{a,b\}_\omega\big)
+ \mathcal{O}(\hbar^2)
\]
in operator norm. Applying this with $a=f^\Sigma$, $b=g^\Sigma$ and using the
Poisson--Jacobi correspondence
\(
\{f^\Sigma,g^\Sigma\}_\omega = (\{f,g\}_\alpha)^\Sigma,
\)
we obtain the Dirac condition of Definition~\ref{def:semiclassical-quantization}
(with an explicit $\mathcal{O}(\hbar^2)$ remainder).

\smallskip

\noindent\emph{Step 2: Hamiltonian part and {\rm\ref{assumption_CC2}}.}
Fix $f\in\mathcal{A}_{\mathrm{cl}}$ and write $\widehat{f}_\hbar=\Op_\hbar^{\mathrm{w}}(f^\Sigma)$.
By hypothesis (2),
\begin{align*}
\frac{i}{\hbar}\big[\widehat{H}_\hbar,\widehat{f}_\hbar\big]
&= \frac{i}{\hbar}\big[\Op_\hbar^{\mathrm{w}}(H),\Op_\hbar^{\mathrm{w}}(f^\Sigma)\big]
 + \frac{i}{\hbar}\big[R_\hbar,\Op_\hbar^{\mathrm{w}}(f^\Sigma)\big].
\end{align*}
The second term is $\mathcal{O}(\hbar)$ in operator norm since $\|\Op_\hbar^{\mathrm{w}}(f^\Sigma)\|$
is uniformly bounded for $f^\Sigma\in S^0$ and $\|R_\hbar\| = \mathcal{O}(\hbar^2)$. For the first term, the commutator expansion yields
\[
\frac{i}{\hbar}\big[\Op_\hbar^{\mathrm{w}}(H),\Op_\hbar^{\mathrm{w}}(f^\Sigma)\big]
= \Op_\hbar^{\mathrm{w}}\big(\{H,f^\Sigma\}_\omega\big)+\mathcal{O}(\hbar^2).
\]
Using $\{H,f^\Sigma\}_\omega=(\{f,h\}_\alpha)^\Sigma$, we obtain
\[
\frac{i}{\hbar}\big[\widehat{H}_\hbar,\widehat{f}_\hbar\big]
= Q_\hbar\big(\{f,h\}_\alpha\big)+\mathcal{O}(\hbar)
\]
in operator norm.

\smallskip

\smallskip

\noindent\emph{Step 3: Estimate of the dissipative part.} Fix $f\in\mathcal{A}_{\mathrm{cl}}$ and set $\widehat f_\hbar=\Op_\hbar^{\mathrm{w}}(f^\Sigma)$.
Let $\mathcal{D}_\hbar$ denote the dissipator in the Heisenberg picture,
\[
\mathcal{D}_\hbar(A)
\coloneqq \sum_{k=1}^{N(\hbar)}\Big(
L_{k,\hbar}^\dagger A L_{k,\hbar}
- \tfrac12\{L_{k,\hbar}^\dagger L_{k,\hbar},A\}
\Big).
\]
For each $k$, let
\[
\mathcal{D}_{k,\hbar}(A)\coloneqq 
L_{k,\hbar}^\dagger A L_{k,\hbar}
-\tfrac12\{L_{k,\hbar}^\dagger L_{k,\hbar},A\},
\]
so that $\displaystyle{\mathcal{D}_\hbar=\sum_{k=1}^{N(\hbar)}\mathcal{D}_{k,\hbar}}$ and hence
\[
\mathcal{D}_\hbar(\widehat f_\hbar)
=\sum_{k=1}^{N(\hbar)}\mathcal{D}_{k,\hbar}(\widehat f_\hbar).
\]
By the Weyl calculus, for each fixed $k$ the summand
$\mathcal{D}_{k,\hbar}(\widehat f_\hbar)$ is a Weyl operator
$\Op_\hbar^{\mathrm{w}}(d_{k,\hbar})$ up to $\mathcal{O}(\hbar)$ in operator norm (collecting the remainders coming
from $L_{k,\hbar}-\Op_\hbar^{\mathrm w}(\ell_k)=\mathcal O(\hbar)$),
where the symbol admits an expansion
\[
d_{k,\hbar}=d_k^{(0)}+\hbar d_k^{(1)}+\mathcal{O}(\hbar^2).
\]
The order-$\hbar^0$ term cancels identically,
\[
d_k^{(0)}=
\overline{\ell_k} f^\Sigma \ell_k
-\tfrac12\big(|\ell_k|^2 f^\Sigma+f^\Sigma|\ell_k|^2\big)=0.
\]
Moreover, $d_k^{(1)}$ is a linear combination of Poisson brackets of the form
$\{\ell_k,f^\Sigma\}_\omega$ and $\{| \ell_k |^2,f^\Sigma\}_\omega$.
Assumption~(3) therefore implies $d_k^{(1)}\equiv 0$ on $M$, and hence
\[
\mathcal{D}_{k,\hbar}(\widehat f_\hbar)=\mathcal{O}(\hbar)
\quad\text{in operator norm.}
\]
Summing over $k$ we obtain
\[
\mathcal{D}_{\hbar}(\widehat f_\hbar)
=\sum_{k=1}^{N(\hbar)}\mathcal{D}_{k,\hbar}(\widehat f_\hbar)
= \mathcal{O}\!\big(N(\hbar)\hbar\big),
\]
and since $N(\hbar)=\mathcal{O}(1)$ we conclude that
\[
\mathcal{D}_{\hbar}(\widehat f_\hbar)=\mathcal{O}(\hbar)
\quad\text{as }\hbar\to 0.
\]
Consequently, for each fixed $f\in\mathcal{A}_{\mathrm{cl}}$,
\[
\|\mathcal{D}_\hbar(\widehat f_\hbar)\|=\mathcal{O}(\hbar)
\qquad\text{as }\hbar\to0.
\]
This estimate is sufficient for establishing the semiclassical limit required in
{\rm\ref{assumption_CC2}}.

\smallskip

\noindent\emph{Step 4: Conclusion of {\rm\ref{assumption_CC2}}.}
Combining Steps~2 and~3 we obtain
\[
\mathcal{L}_\hbar^\dagger(\widehat{f}_\hbar)
= \frac{i}{\hbar}\big[\widehat{H}_\hbar,\widehat{f}_\hbar\big]
 + \mathcal{D}_\hbar(\widehat{f}_\hbar)
= Q_\hbar\big(\{f,h\}_\alpha\big) + \mathcal{O}(\hbar),
\]
hence
\[
\lim_{\hbar\to 0}\Big\|
\mathcal{L}_\hbar^\dagger(\widehat{f}_\hbar)
 - Q_\hbar\big(\{f,h\}_\alpha\big)
\Big\| = 0,
\]
which is {\rm\ref{assumption_CC2}}.

\smallskip

\noindent\emph{Step 5: Exact quantum constants of motion.}
Finally, {\rm\ref{assumption_CC1}} follows from the commutation relations
\eqref{eq:exact-commutant-choice} and Lemma~\ref{lem:kernel-condition-contact-compatible}.



\end{proof}


\begin{example}
\label{ex:egorov-dephasing}
We illustrate Theorem~\ref{prop:egorov-criterion} on the linear contact model of
Example~\ref{ex:linear-contact-correct}, with
\[
(C,\alpha,h)=(\mathbb{R}^3(q,p,z),\,\dd z-p\,\dd q,\,h=z-p),
\qquad
I_1\coloneqq z.
\]
This system arises by restriction from the exact symplectic homogeneous system
\[
M=C\times\mathbb{R}_{>0}\cong \mathbb{R}^4(q,p,z,r),\qquad
\theta=r(\dd z-p\,\dd q),\qquad \omega\coloneqq -\dd\theta,
\qquad
H^\Sigma\coloneqq \change{r\,h=r(z-p)},
\]
with Liouville-transverse hypersurface $C=\{r=1\}$.
The Jacobi subalgebra is
\[
\mathcal{A}_{\mathrm{cl}}
=\mathrm{Jacobi\text{-}alg}\langle I_0=h,\,I_1=z\rangle
\subset \Cinfty(C).
\]

\smallskip

Assume that $M$ is quantized by semiclassical Weyl quantization on $\mathbb{R}^4$, so that
\[
Q_\hbar(f)=\Op_\hbar^{\mathrm{w}}(f^\Sigma),
\qquad f\in\mathcal{A}_{\mathrm{cl}}.
\]
Choose the quantum Hamiltonian
\[
\widehat{H}_\hbar=\Op_\hbar^{\mathrm{w}}(H^\Sigma)+\mathcal{O}(\hbar^2).
\]

For the dissipative part, take a single Lindblad operator
\[
L_\hbar=\Op_\hbar^{\mathrm{w}}(\ell)+\mathcal{O}(\hbar),
\qquad
\ell=\sqrt{\gamma}\,\varphi\!\big(H^\Sigma,I_1^\Sigma\big),
\qquad \gamma>0,
\]
where $\varphi$ is a smooth real-valued function and $I_1^\Sigma$ denotes the
$1$-homogeneous lift of $I_1$.

Since $\ell$ depends only on the commuting homogeneous integrals
$(H^\Sigma,I_1^\Sigma)$, one has
\[
\{\ell,f^\Sigma\}_\omega=0
\qquad\text{and}\qquad
\{|\ell|^2,f^\Sigma\}_\omega=0
\quad\text{on } M
\qquad \forall f\in\mathcal{A}_{\mathrm{cl}}.
\]
Moreover, here $N(\hbar)=1=\mathcal{O}(1)$.

\smallskip

Therefore, all the hypotheses of Theorem~\ref{prop:egorov-criterion} are satisfied,
and the Lindblad generator $\mathcal L_\hbar$ fulfills~{\rm\ref{assumption_CC2}}.

If, in addition, the operators $\widehat{I}_{0,\hbar}$ and $\widehat{I}_{1,\hbar}$
are chosen so that $L_\hbar$ is implemented via a bounded functional calculus of
these commuting observables, then the commutation relations
\eqref{eq:exact-commutant-choice} hold and {\rm\ref{assumption_CC1}} is satisfied exactly.

\smallskip

Heuristically, this provides a semiclassical model of \emph{pure dephasing}:
the dissipator is diagonal with respect to the joint spectral decomposition of the
commuting quantum integrals, suppressing off-diagonal coherences while preserving
their expectation values.
\end{example}

\begin{remark}

We now give a concrete physical interpretation of the Egorov-type criterion
(Theorem~\ref{prop:egorov-criterion}) in the language of decoherence, in a way
tailored to the Weyl-type and commutant-based framework adopted above.

Fix $\hbar$ and consider the commutative $C^\ast$-algebra
\[
\mathcal{A}_\hbar \coloneqq C^\ast(\widehat{I}_{0,\hbar},\dots,\widehat{I}_{n,\hbar})
\subset \mathcal{B}(\mathcal{H}_\hbar),
\]
generated by the quantized dissipated quantities.  By construction,
$\mathcal{A}_\hbar$ is unital and consists of bounded operators.

In all dephasing-type examples relevant for us, the Lindblad operators are chosen
to be of Weyl type and, more specifically, as bounded functions of the commuting
quantum integrals, i.e.
\[
L_{k,\hbar}
=
\sqrt{\gamma_k}\,
\varphi_k(\widehat{I}_{0,\hbar},\dots,\widehat{I}_{n,\hbar}),
\qquad \gamma_k\ge 0,
\]
with a number of channels satisfying $N(\hbar)=\mathcal{O}(1)$ as $\hbar\to0$.
At the level of principal symbols, this corresponds to choosing
$\ell_k\in S^0(M)$ depending only on the homogeneous lifts of the classical
integrals, so that
\[
\{\ell_k,f^\Sigma\}_\omega
=
\{|\ell_k|^2,f^\Sigma\}_\omega
=0
\quad\text{on } M
\]
for all $f\in\mathcal{A}_{\mathrm{cl}}$.

As a consequence, the Lindblad operators belong to $\mathcal{A}_\hbar$, and hence
\[
[L_{k,\hbar},A]=[L_{k,\hbar}^\dagger,A]=0
\quad
\forall\,A\in\mathcal{A}_\hbar.
\]
If, in addition, the Hamiltonian part is chosen in the commutant of
$\mathcal{A}_\hbar$, i.e.
\(
[\widehat{H}_\hbar,\widehat{I}_{j,\hbar}]=0
\)
for all $j$, then the exact commutation relations
\eqref{eq:exact-commutant-choice} hold.  By the argument of Step~5 in the proof of
Theorem~\ref{prop:egorov-criterion}, this implies that every observable in
$\mathcal{A}_\hbar$ is an exact quantum constant of motion:
\[
\mathcal{L}_\hbar^\dagger(A)=0
\qquad \forall\,A\in\mathcal{A}_\hbar.
\]
Thus $\mathcal{A}_\hbar$ defines an invariant commutative sector of the Heisenberg
evolution.

This algebraic invariance has a direct decoherence interpretation.  Since the
operators $\widehat{I}_{j,\hbar}$ commute, they can be simultaneously
diagonalized, yielding a decomposition of the Hilbert space into joint
eigenspaces,
\[
\mathcal{H}_\hbar=\bigoplus_{\nu} \mathcal{H}_{\hbar,\nu},
\qquad
\widehat{I}_{j,\hbar}\big|_{\mathcal{H}_{\hbar,\nu}}
=\lambda_{j,\nu}\,\mathbb{I}.
\]
Every operator in $\mathcal{A}_\hbar$ is block-diagonal with respect to this
decomposition, acting as a scalar on each $\mathcal{H}_{\hbar,\nu}$.
The joint eigenspaces $\mathcal{H}_{\hbar,\nu}$ therefore play the role of
\emph{pointer sectors} selected by the dissipative dynamics.

Denoting by $P_\nu$ the projection onto $\mathcal{H}_{\hbar,\nu}$, the action of
the dissipator on off-diagonal density-matrix blocks
$\rho_{\nu\mu}=P_\nu\rho P_\mu$ ($\nu\neq\mu$) takes the form
\[
\frac{\dd}{\dd t}\rho_{\nu\mu}
=
\sum_k\Big(
\ell_{k,\nu}\overline{\ell_{k,\mu}}
-\tfrac12|\ell_{k,\nu}|^2
-\tfrac12|\ell_{k,\mu}|^2
\Big)\rho_{\nu\mu}
\;+\;(\text{Hamiltonian commutator term}),
\]
where $\ell_{k,\nu}$ denotes the scalar eigenvalue of $L_{k,\hbar}$ on
$\mathcal{H}_{\hbar,\nu}$, i.e.\ the value of the functional calculus
$\sqrt{\gamma_k}\varphi_k$ evaluated on the joint eigenvalues
$(\lambda_{0,\nu},\dots,\lambda_{n,\nu})$.  Whenever these values distinguish
different sectors, the real part of the prefactor is strictly negative and the
off-diagonal blocks decay exponentially, while the diagonal blocks
$\rho_{\nu\nu}$ remain unaffected.  This is precisely the mechanism of pure
dephasing.

Finally, Theorem~\ref{prop:egorov-criterion} explains why this
decoherence-induced pointer structure is compatible with the classical contact
dynamics in the semiclassical limit.  For every classical observable
$f\in\mathcal{A}_{\mathrm{cl}}$, the Heisenberg generator satisfies
\[
\mathcal{L}_\hbar^\dagger(\widehat{f}_\hbar)
=
Q_\hbar(\{f,h\}_\alpha)+o(1),
\qquad \hbar\to0,
\]
so that the evolution of expectation values converges to transport along the
contact Hamiltonian flow of $h$. In this sense, the present construction realizes a clean separation of roles:
\emph{decoherence} selects a robust commutative algebra of observables, while the
\emph{Egorov property} guarantees that, within this sector, the quantum evolution
reduces to the intended classical contact dynamics.
\hfill$\diamond$
\end{remark}

\begin{example}[Dephasing and pointer sectors for the linear contact model]
\label{ex:remark-dephasing}
We illustrate the content of the above remark on the linear contact system
introduced in Example~\ref{ex:egorov-dephasing}.

Fix $\hbar>0$ and consider the commuting family of quantized observables
\[
\widehat{I}_{0,\hbar}=\widehat{h}_\hbar,
\qquad
\widehat{I}_{1,\hbar}=\Op_\hbar^{\mathrm w}(I_1^\Sigma),
\qquad
I_1=z.
\]
Let
\[
\mathcal{A}_\hbar\coloneqq C^\ast(\widehat{I}_{0,\hbar},\widehat{I}_{1,\hbar})
\subset\mathcal{B}(\mathcal{H}_\hbar)
\]
be the commutative $C^\ast$-algebra generated by these operators.

Choose a Lindblad operator of dephasing type,
\[
L_\hbar=\sqrt{\gamma}\,\varphi(\widehat{I}_{0,\hbar},\widehat{I}_{1,\hbar}),
\qquad \gamma>0,
\]
where $\varphi$ is a bounded real-valued function.  Then
$L_\hbar\in\mathcal{A}_\hbar$ and therefore
\[
[L_\hbar,A]=[L_\hbar^\dagger,A]=0
\qquad\forall\,A\in\mathcal{A}_\hbar.
\]
If, in addition, the Hamiltonian is chosen so that
\(
[\widehat{H}_\hbar,\widehat{I}_{j,\hbar}]=0
\)
for $j=0,1$, then every observable in $\mathcal{A}_\hbar$ is an exact quantum
constant of motion:
\[
\mathcal{L}_\hbar^\dagger(A)=0
\qquad\forall\,A\in\mathcal{A}_\hbar.
\]

Since the generators $\widehat{I}_{0,\hbar}$ and $\widehat{I}_{1,\hbar}$ commute,
the Hilbert space admits a decomposition into their joint eigenspaces,
\[
\mathcal{H}_\hbar=\bigoplus_{\nu}\mathcal{H}_{\hbar,\nu},
\qquad
\widehat{I}_{j,\hbar}\big|_{\mathcal{H}_{\hbar,\nu}}
=\lambda_{j,\nu}\,\mathbb{I}.
\]
In this representation, the Lindblad operator acts diagonally,
\(
L_\hbar\big|_{\mathcal{H}_{\hbar,\nu}}
=\ell_\nu\,\mathbb{I}
\),
and the dissipative dynamics leaves each sector $\mathcal{H}_{\hbar,\nu}$
invariant.

A direct computation in the Schr\"odinger picture shows that the off-diagonal
density-matrix blocks
\(
\rho_{\nu\mu}=P_\nu\rho P_\mu
\)
($\nu\neq\mu$) decay exponentially in time, while the diagonal blocks
$\rho_{\nu\nu}$ are unaffected by the dissipator.  Thus the Lindblad evolution
implements pure dephasing with respect to the joint spectral decomposition of
$(\widehat{I}_{0,\hbar},\widehat{I}_{1,\hbar})$, and the subspaces
$\mathcal{H}_{\hbar,\nu}$ play the role of pointer sectors.

Finally, Theorem~\ref{prop:egorov-criterion} ensures that this quantum
superselection structure is compatible with the underlying classical contact
dynamics.  For any classical observable
$f\in\mathcal{A}_{\mathrm{cl}}=\mathrm{Jacobi\text{-}alg}\langle h,I_1\rangle$,
the Heisenberg generator satisfies
\[
\mathcal{L}_\hbar^\dagger(\widehat f_\hbar)
=
Q_\hbar(\{f,h\}_\alpha)+o(1),
\qquad \hbar\to0,
\]
so that the evolution of expectation values converges to transport along the
contact Hamiltonian flow of $h$. This example shows explicitly how, in dephasing-compatible constructions,
environment-induced decoherence selects a robust commutative algebra of
observables, while the Egorov property guarantees that the effective dynamics
inside this sector reduces to the intended classical contact evolution.
\end{example}

\section{An Euler-top-type Poisson--Lie pencil as a bi-Lindblad system}
\label{sec:euler-example}

In this section we illustrate how the abstract framework developed above can be
implemented in a concrete low-dimensional model of \emph{Euler-top type}, inspired
by the deformed Euler top of Marrero and collaborators~\cite{ballesteros2017poisson,G.I.M+2023}.
Rather than reproducing every physical feature of the rigid-body Euler top, we
focus on the structural ingredients that are essential for our construction:
(i) a linear Poisson pencil of Poisson--Lie origin, (ii) a bi-Hamiltonian vector
field admitting a Lenard--Magri hierarchy of commuting invariants, (iii) a
homogeneous lift yielding a completely integrable contact system with
nontrivial dissipated quantities, and (iv) an algebraic setting in which one can
build Lindblad generators implementing pure dephasing relative to the quantum
integrals.  In this sense, the example below should be read as a
\emph{deformed-Euler-top-type} testbed tailored to the bi-Lindblad viewpoint.

Let $\mathfrak{g}=\mathfrak{so}(3)$ and identify $\mathfrak{g}^\ast\simeq\mathbb{R}^3$
with coordinates $m=(m_1,m_2,m_3)$.  The standard Lie--Poisson bracket reads
\[
\{m_i,m_j\}_0 = -\varepsilon_{ijk} m_k,
\]
where $\varepsilon_{ijk}$ is the Levi--Civita symbol, {i.e.,
\begin{equation}\label{eq:euler-poisson-0}
\{m_1,m_2\}_0 = -m_3, \qquad
\{m_1,m_3\}_0 = m_2, \qquad
\{m_2,m_3\}_0 = -m_1.
\end{equation}
For comparison, the
classical Euler top is Hamiltonian with respect to this bracket, with Hamiltonian
\[
H_0(m)=\frac12\Big(\frac{m_1^2}{J_1}+\frac{m_2^2}{J_2}+\frac{m_3^2}{J_3}\Big),
\]
where $J_1,J_2,J_3>0$ are the principal moments of inertia.

A bi-Hamiltonian formulation of Euler-top type dynamics is obtained by endowing
$\mathfrak{so}(3)^*\simeq\mathbb{R}^3$, with coordinates
$m=(m_1,m_2,m_3)$, with a second \emph{linear} Poisson bracket
$\{\cdot,\cdot\}_1$ compatible with the standard Lie--Poisson bracket
$\{\cdot,\cdot\}_0$.
A second compatible Poisson structure is given by
\begin{equation}\label{eq:euler-poisson-1}
\{m_1,m_2\}_1 = -m_2, \qquad
\{m_1,m_3\}_1 = m_3, \qquad
\{m_2,m_3\}_1 = -2m_1,
\end{equation}
which is again linear and corresponds to the Lie--Poisson structure of
$\mathfrak{sl}(2,\mathbb{R})$.

The two brackets $\{\cdot,\cdot\}_0$ and $\{\cdot,\cdot\}_1$ are
\emph{compatible}, in the sense that for every $\lambda\in\mathbb{R}$ the linear
combination
\[
\{\cdot,\cdot\}_\lambda
\coloneqq \{\cdot,\cdot\}_1 - \lambda\,\{\cdot,\cdot\}_0
\]
is again a Poisson bracket. Equivalently, the associated bivectors
$\pi_0$ and $\pi_1$ define a Poisson pencil
\[
\pi_\lambda \coloneqq \pi_1 - \lambda \pi_0 .
\]

Rather than considering the full family of Hamiltonian flows associated with the
pencil, we fix a distinguished \emph{bi-Hamiltonian vector field} $X$ on
$\mathbb{R}^3$ defined by the requirement that it be Hamiltonian with respect to
both Poisson structures.  Concretely, $X$ is chosen so that there exist two
functions $C_0,C_1\in \Cinfty(\mathbb{R}^3)$ satisfying
\begin{equation}\label{eq:vector_field_Casimirs}
    X=\pi_0^\sharp\,\dd C_1=\pi_1^\sharp\,\dd C_0 .
\end{equation}
This is the Euler-top-type dynamics that we shall quantize: a coadjoint motion on
$\mathfrak{so}(3)^*$ admitting two compatible Hamiltonian descriptions and a
complete set of commuting integrals generated by the Poisson pencil.

This choice is essential for what follows.  The bi-Hamiltonian nature of $X$
ensures the existence of a Lenard--Magri hierarchy of invariants, which in turn
gives rise—after homogeneous lift—to a contact integrable system with
distinguished dissipated quantities.  These quantities will generate the
commutative classical algebra $\mathcal{A}_{\mathrm{cl}}$ and, upon
quantization, the invariant quantum algebra underlying the dephasing Lindblad
dynamics constructed below.


In coordinates, the Poisson pencil associated with
$\pi_\lambda=\pi_1-\lambda\pi_0$ reads
\begin{equation}\label{eq:euler-pencil}
\{m_1,m_2\}_\lambda = -m_2 + \lambda m_3, \qquad
\{m_1,m_3\}_\lambda = m_3 - \lambda m_2, \qquad
\{m_2,m_3\}_\lambda = (\lambda-2)m_1 .
\end{equation}
A convenient choice of Casimir functions for this pencil is
\begin{equation}\label{eq:Casimirs-C0C1}
C_0(m) = -\tfrac12\big(m_1^2+m_2^2+m_3^2\big), \qquad
C_1(m) = m_1^2 + m_2 m_3 .
\end{equation}
With this choice, we can define the vector field $X$ given by \eqref{eq:vector_field_Casimirs}:
\change{
      $$X = \left(m_3^2 - m_2^2\right) \frac{\partial}{\partial m_1}
      + \left(m_1 m_2 - 2 m_1 m_3\right) \frac{\partial}{\partial m_2} 
      + \left(2m_1 m_2 -  m_1 m_3\right) \frac{\partial}{\partial m_3} \, .$$
}
This vector field is the Euler-top-type dynamics that we fix throughout this
section.  Its trajectories are confined to the common level sets of $(C_0,C_1)$,
and the pair $(C_0,C_1)$ provides a complete set of commuting first integrals for
this three-dimensional system.

More generally, the Poisson pencil $\pi_\lambda$ underlies a Lenard--Magri
hierarchy of commuting invariants $\{H_k\}_{k\ge0}$ satisfying
\[
\pi_0^\sharp \dd H_{k+1} = \pi_1^\sharp \dd H_k ,
\qquad k\ge0,
\]
with $(H_0,H_1)=(C_0,C_1)$.  In the present example the hierarchy truncates after
these two basic integrals, which already capture the full integrable structure
of the flow.

\medskip


To connect the Lie--Poisson dynamics on $\mathfrak{so}(3)^*$ with the
homogeneous/contact framework developed in the previous sections, we proceed
via an exact symplectic realization of the linear Poisson structure
$\pi_0$.
Concretely, we consider a symplectic manifold $(S,\omega)$ with
$\omega=-\dd\Theta$ exact, together with a Poisson map
\[
\Phi:(S,\omega)\longrightarrow(\mathfrak{so}(3)^*,\pi_0).
\]
Such exact symplectic realizations of linear Poisson manifolds are classical and
can be constructed explicitly; see, for instance,
\cite[Section~6.3]{C.F.M2021}. In particular, for
$\mathfrak{so}(3)^*\simeq\mathfrak{su}(2)^*$ one may take
$S=N\times\mathbb{R}^3_m$, where $N$ is a $3$-manifold equipped with a coframe
$\{\vartheta_1,\vartheta_2,\vartheta_3\}$ satisfying the Maurer--Cartan relations
\change{
      $$\dd \vartheta_1 = - \vartheta_2 \wedge \vartheta_3 \, , \quad
      \dd \vartheta_2 = - \vartheta_3 \wedge \vartheta_1 
      \, , \quad
      \dd \vartheta_3 = - \vartheta_1 \wedge \vartheta_2\, ,$$ 
}
and define
\[
\Theta \coloneqq m_1\vartheta_1+m_2\vartheta_2+m_3\vartheta_3,
\qquad
\omega\coloneqq -\dd\Theta\, ,
\]
\change{that is,
$$\omega = m_1 \vartheta_2 \wedge \vartheta_3 + m_2 \vartheta_3 \wedge \vartheta_1 + m_3 \vartheta_1 \wedge \vartheta_2 + \vartheta_1 \wedge \dd m_1 + \vartheta_2 \wedge \dd m_2 + \vartheta_3 \wedge \dd m_3 \, .$$}

\noindent
The projection $\Phi(n,m)=m$ is then a Poisson map onto
$(\mathfrak{so}(3)^*,\pi_0)$.
\change{In fact, taking $N=G=\mathrm{SO}(3)$ with $\{\vartheta_1,\vartheta_2,\vartheta_3\}$ its left-invariant Maurer--Cartan coframe, this realization is precisely the cotangent bundle $(S,\omega)\cong\cT G$ in its left trivialization $\cT G\cong G\times\mathfrak{g}^\ast$ (with $\mathfrak{g}=\mathfrak{so}(3)$): the primitive $\Theta$ is the tautological one-form and $\Phi$ is the body momentum map.}
In this realization, the classical invariants on $\mathfrak{so}(3)^*$ lift
naturally by pullback. In particular, for the Casimir functions
$C_a\in \Cinfty(\mathfrak{so}(3)^*)$ introduced in
\eqref{eq:Casimirs-C0C1}, we define
\[
F_a \coloneqq \Phi^*C_a,\qquad a=0,1,
\] namely,
$$F_0(n,m) = -\tfrac12\big(m_1^2+m_2^2+m_3^2\big), \qquad
F_1(n,m) = m_1^2 + m_2 m_3 .$$
These lifted functions Poisson-commute on $(S,\omega)$.
\change{Since $\Phi$ realizes $\pi_0$, on which $C_0$ is a Casimir, the Hamiltonian vector field of $F_0$ projects to zero on $\mathfrak{so}(3)^*$, whereas that of $F_1$ projects onto the bi-Hamiltonian vector field $X$ of \eqref{eq:vector_field_Casimirs} (up to sign). The Euler-top-type dynamics is thus recovered through the lift $F_1$ of $C_1$, not through $F_0$.}
Note that $\omega = -\dd \Theta$ is an exact symplectic form by definition, and  thus we can define the Liouville vector field as the unique vector field $\Delta$ on $S$ such that
$$\contr{\Delta} \omega = -\Theta\, .$$
\change{Let $\{Y_1, Y_2, Y_3\}$ be the frame on $N$ dual to $\{\vartheta_1,\vartheta_2,\vartheta_3\}$, i.e.,
$$\contr{Y_i} \vartheta_j = \delta_{ij}\, , \quad 1\leq i,j \leq 3\, .$$
By standard Cartan calculus properties, we have that
$$\contr{[Y_i, Y_j]} \theta_k = \contr{\theta_i} \contr{\theta_j} \dd \theta_k = \varepsilon_{ijk}\, , $$
and thus
$$[Y_i, Y_j] =  \varepsilon_{ijk} Y_k\, ,$$
that is,
$$[Y_1, Y_2] = Y_3\, , \quad[Y_1, Y_3] = - Y_2\, , \quad [Y_2, Y_3] = Y_1\, .$$
In the frame $\{Y_1, Y_2, Y_3, \partial_{m_1}, \partial_{m_2}, \partial_{m_3}\}$ the Hamiltonian vector field of an arbitrary $f\in \Cinfty(S)$ reads
\change{$$X_f = \sum_{i=1}^3 \left(\frac{\partial f}{\partial m_i} Y_i- Y_i(f) \frac{\partial}{\partial m_i}\right) + \sum_{i,j,k=1}^3 \varepsilon_{ijk}\, m_i \frac{\partial f}{\partial m_j}\, \frac{\partial}{\partial m_k}\, ,$$}
where $\contr{X_f}\omega = \dd f$.
In particular,
\change{$$X_{m_i} = Y_i + \sum_{j,k=1}^3 \varepsilon_{jik}\, m_j\, \frac{\partial}{\partial m_k}\, , \quad 1\leq i \leq 3\, ,$$
so that $\Phi$ is indeed a Poisson map onto $(\mathfrak{so}(3)^\ast, \pi_0)$, $\{m_i, m_j\}_\omega = -\varepsilon_{ijk} m_k$,} and
$$X_{F_0} = -m_1 Y_1 - m_2 Y_2 - m_3 Y_3\, ,$$
\change{$$\begin{aligned} X_{F_1} = {}&2m_1 Y_1 + m_3 Y_2 + m_2 Y_3 \\ &+ (m_2^2-m_3^2)\, \frac{\partial}{\partial m_1} + (2m_1 m_3 - m_1 m_2)\, \frac{\partial}{\partial m_2} + (m_1 m_3 - 2m_1 m_2)\, \frac{\partial}{\partial m_3}\, .\end{aligned}$$}
\change{Note that $X_{F_0}$ has no $\partial_{m_i}$-components and $\Phi_\ast X_{F_0}=0$, whereas $\Phi_\ast X_{F_1}=-X$ is the bi-Hamiltonian vector field of \eqref{eq:vector_field_Casimirs}.}
}
On each hypersurface $C$ of $S$ which is transversal to $\Delta$, the restriction $\alpha\coloneqq \restr{\Theta}{C}$ is a contact form. We can check that 
$$\Delta = \sum_{i=1}^3 m_i \frac{\partial}{\partial m_i}\, ,$$
and hence $F_0$ and $F_1$ are homogeneous of degree $2$, i.e., $\Delta F_a =2 F_a,\, a =0,1$. Let us define the functions $I_a^\Sigma\colon U \to \RR$, defined on the open subset \change{$U= \{F_1>0\}\subset S$}, given by
$$I_0^\Sigma = \sqrt{-2F_0} = \sqrt{m_1^2+m_2^2+m_3^2}\, , \quad I_1^\Sigma = \sqrt{F_1} = \sqrt{m_1^2 + m_2 m_3 }\, .$$
These functions are homogeneous of degree $1$ and Poisson-commute, namely,
$$\{I_0^\Sigma, I_1^\Sigma\}_\omega = 0\, ,$$
where $\{\cdot, \cdot\}_{\omega}$ denotes the Poisson bracket defined by $\omega$. 
\change{Their Hamiltonian vector fields are given by
$$X_{I_0^\Sigma} = \frac{\dd I_0^\Sigma}{\dd F_0} X_{F_0} = -\frac{1}{\sqrt{-2F_0}} X_{F_0}\, , \quad X_{I_1^\Sigma} = \frac{\dd I_1^\Sigma}{\dd F_1} X_{F_1} = \frac{1}{2\sqrt{F_1}} X_{F_1}\, .$$
}

\change{So far, $I_0^\Sigma$ and $I_1^\Sigma$ furnish only two independent integrals in involution, whereas complete integrability of the six-dimensional symplectic manifold $(S,\omega)$ requires three. No third \emph{basic} integral---a function of $m$ alone---can be added: the Lenard--Magri hierarchy of the pencil truncates at $(C_0,C_1)$, and any $\overline{g}\in\Cinfty(\RR^3_m)$ with $\{F_1,\overline{g}\}_\omega=0$ is a function of $C_0$ and $C_1$. A third integral is instead supplied by the realization $(S,\omega)\cong\cT G$ itself, which carries a second, commuting momentum map---the \emph{spatial} momentum
$$\mu \coloneqq \Ad^\ast_g\, m \in \mathfrak{so}(3)^\ast\, , \qquad g\in N\, .$$
Here $m=(m_1,m_2,m_3)$ and $\mu=(\mu_1,\mu_2,\mu_3)$ are the components of the body and spatial momenta in $\mathfrak{so}(3)^\ast\cong\RR^3$; explicitly, $\mu_i = \sum_{j=1}^3 (\Ad^\ast_g)_{ij}\, m_j$, so each $\mu_i$ is a function on $S$ that is linear in $m$ with coefficients depending on the configuration $g\in N$. By the standard structure of $\cT G$ (see, e.g., \cite[Section~6.3]{C.F.M2021}), it satisfies
$$\{m_i, \mu_j\}_\omega = 0\, , \quad \{\mu_i, \mu_j\}_\omega = \varepsilon_{ijk}\,\mu_k\, , \quad |\mu|^2 = |m|^2\, , \quad 1\leq i,j\leq 3\, .$$
Fix any one component, say $I_2^\Sigma \coloneqq \mu_3$. It is linear in $m$, hence homogeneous of degree $1$, $\Delta I_2^\Sigma = I_2^\Sigma$; and since $I_0^\Sigma, I_1^\Sigma$ are functions of $m$ alone, the relation $\{m_i,\mu_j\}_\omega=0$ gives, by the Leibniz rule,
$$\{I_0^\Sigma, I_2^\Sigma\}_\omega = \{I_1^\Sigma, I_2^\Sigma\}_\omega = 0\, .$$
The differential $\dd I_2^\Sigma$ has a nonzero component along $N$, which $\dd I_0^\Sigma$ and $\dd I_1^\Sigma$ lack; hence the three functions are functionally independent on a dense open subset. Therefore $(I_0^\Sigma, I_1^\Sigma, I_2^\Sigma)$ is a completely integrable system on $(U,\omega)$, in accordance with Assumption~\ref{assumption:homogeneous-lift} (here $\dim S = 2(n+1)$ with $n=2$). Being homogeneous of degree $1$ and pairwise Poisson-commuting, they restrict to $C$ as three independent, Jacobi-commuting functions $I_0, I_1, I_2\coloneqq\restr{\mu_3}{C}$, supplying the $n+1 = 3$ dissipated quantities required by Definition~\ref{def:contact-integrable}. We stress that $I_2^\Sigma$ is \emph{not} a homogeneous lift of a Casimir of the pencil: it originates in the symmetry of the realization $\cT G$, rather than in the bi-Hamiltonian structure on $\mathfrak{so}(3)^\ast$.}

The hypersurface \change{$C=\{m_1=1\}\cap U = \{m_1=1,\,m_2m_3>-1\}\subset U\subset S$} is transverse to $\Delta$. Hence,
$$\alpha\coloneqq \restr{\Theta}{C} = \vartheta_1+m_2\vartheta_2+m_3\vartheta_3$$
is a contact form on $C$. The functions
$$I_0 \coloneqq \restr{I_0^\Sigma}{C} = \sqrt{1+m_2^2+m_3^2}\, , \quad I_1 \coloneqq \restr{I_1^\Sigma}{C} = \sqrt{1 + m_2 m_3 }\, , \quad \change{I_2 \coloneqq \restr{I_2^\Sigma}{C} = \restr{\mu_3}{C}\, ,}$$
\change{pairwise} commute with respect to the Jacobi bracket $\{\cdot, \cdot\}_\alpha$ defined by $\alpha$\change{; the third one, $I_2$, depends on the configuration coordinates on $N$ rather than on $m$ alone}.
Thus, in this Euler-top-type example the Jacobi-commutative algebra of dissipated
quantities is
\[
\mathcal{A}_{\mathrm{cl}}
=
\mathrm{Jacobi\text{-}alg}\langle I_0\coloneqq h,\,I_1\change{,\,I_2}\rangle
\subset \Cinfty(C),
\]
which encodes the integrable structure inherited from the underlying
bi-Hamiltonian Poisson pencil.
\change{Since $C_0$ is a Casimir of $\pi_0$, the field $X_{F_0}$ has no vertical ($\partial_m$) component; hence $X_{I_0^\Sigma}$ is tangent to $C$, and the contact Hamiltonian vector field of $I_0$ is its restriction,
$$X_{I_0} = \frac{Y_1 + m_2 Y_2 + m_3 Y_3}{\sqrt{m_1^2+m_2^2+m_3^2}} \, .$$
The function $C_1$, by contrast, is not a Casimir of $\pi_0$ but the Hamiltonian generating $X$ through it, so $X_{F_1}$ carries the nonzero vertical component $-X$, and
$$X_{I_1^\Sigma} = \frac{1}{2\sqrt{F_1}}\left(2m_1 Y_1 + m_3 Y_2 + m_2 Y_3 - X\right)\, ,$$
with $X$ as in \eqref{eq:vector_field_Casimirs}, viewed as a vertical vector field on $S$ (i.e.\ tangent to the $\RR^3_m$ factor). Consequently $X_{I_1^\Sigma}$ is not tangent to $C$, and the contact Hamiltonian vector field of $I_1$ is its projection along the Liouville field $\Delta$ onto $\T C$; explicitly, on $C=\{m_1=1\}$,
$$X_{I_1} = \frac{2 Y_1 + m_3 Y_2 + m_2 Y_3 + \left(2m_3 - m_2 - m_2^3 + m_2 m_3^2\right)\partial_{m_2} + \left(m_3 - 2m_2 + m_3^3 - m_2^2 m_3\right)\partial_{m_3}}{2\sqrt{1 + m_2 m_3}}\, .$$
}
\change{The vector field of the third integral is $X_{I_2}=\restr{X_{I_2^\Sigma}}{C}$, the restriction of $X_{\mu_3}$. Applying the general formula for $X_f$ to $f=\mu_3$ and using $\{m_i,\mu_3\}_\omega=0$ (so the vertical part $-Y_i(\mu_3)\partial_{m_i}+\varepsilon_{ijk}m_i(\partial\mu_3/\partial m_j)\partial_{m_k}$ cancels), one finds that $X_{\mu_3}$ has no $\partial_{m_i}$-component and reads
$$X_{I_2} = \restr{\left(\sum_{i=1}^3 \frac{\partial \mu_3}{\partial m_i}\, Y_i\right)}{C} = \restr{\left(\sum_{i=1}^3 (\Ad^\ast_g)_{3i}\, Y_i\right)}{C}\, ,$$
generating the spatial rotation about the third axis ($\{\mu_i,\mu_3\}_\omega=\varepsilon_{i3k}\mu_k$). The coefficients $(\Ad^\ast_g)_{3i}$ depend on the configuration $g\in N$ and not on $m$. To exhibit them explicitly, realize $N=G=\mathrm{SO}(3)$ and let $R=R(g)\in\mathrm{SO}(3)$ denote the matrix of $\Ad^\ast_g$ in the basis dual to $\{\vartheta_1,\vartheta_2,\vartheta_3\}$, so that $(\Ad^\ast_g)_{ij}=R_{ij}$. Let $\gamma\coloneqq R^\top e_3=(R_{31},R_{32},R_{33})$ be the third row of $R$, the unit vector on $N$ representing the space-fixed third axis in the body frame (the Poisson vector of rigid-body kinematics); it satisfies
$$|\gamma|=1\, ,\qquad \gamma\cdot m=\mu_3\, ,\qquad Y_i(\gamma_j)=\varepsilon_{ijk}\gamma_k\, .$$
Then $\partial\mu_3/\partial m_i=\gamma_i$, and
$$X_{I_2}=\restr{\left(\sum_{i=1}^3\gamma_i\, Y_i\right)}{C}=\gamma_1 Y_1+\gamma_2 Y_2+\gamma_3 Y_3\, ,\qquad I_2=\restr{\mu_3}{C}=\gamma_1+\gamma_2 m_2+\gamma_3 m_3\, .$$
}

This construction makes explicit the conceptual separation between the
Lie--Poisson data and the contact geometry. The Poisson structures
$\pi_0,\pi_1$ on $\mathfrak{so}(3)^*$ govern the classical integrable content via
the bi-Hamiltonian vector field $X$ and its invariants, while the contact
structure on $C$ arises indirectly from homogeneous symplectization of an exact
symplectic realization. In particular, the contact form $\alpha=\Theta|_C$ is not
an independent ingredient but is induced by the primitive $\Theta$ of the
symplectic form.
Accordingly, the Jacobi-commutative algebra
$\mathcal{A}_{\mathrm{cl}}=\mathrm{Jacobi\text{-}alg}\langle I_0,I_1\change{,I_2}\rangle$ and
the dissipated quantities $I_j$ are determined entirely by the homogeneous lift
of the Poisson-pencil invariants, and do not depend on a special choice of
contact representative. These quantities will play a central role in the
construction of the dephasing Lindblad dynamics in the next step.

\medskip
We fix a semiclassical quantization $(\mathcal{H}_\hbar,Q_\hbar)$ of
$\mathcal{A}_{\mathrm{cl}}$ in the sense of Definition~\ref{def:semiclassical-quantization},
implemented by semiclassical Weyl quantization on the homogeneous symplectic
manifold $(M,\omega)$ after the homogeneous lift. Concretely, if
$f\in\mathcal{A}_{\mathrm{cl}}$ and $f^\Sigma$ denotes its $1$-homogeneous lift
on $M$ (so that $f^\Sigma|_{C}=f$), we set
\[
Q_\hbar(f)\coloneqq \Op_\hbar^{\mathrm{w}}(f^\Sigma),
\qquad
\widehat{f}_\hbar\coloneqq Q_\hbar(f).
\]
In particular, for the Euler-top generators we obtain
\[
\widehat{I}_{0,\hbar}=Q_\hbar(I_0)=\Op_\hbar^{\mathrm{w}}(I_0^\Sigma),
\qquad
\widehat{I}_{1,\hbar}=Q_\hbar(I_1)=\Op_\hbar^{\mathrm{w}}(I_1^\Sigma),
\change{\qquad
\widehat{I}_{2,\hbar}=Q_\hbar(I_2)=\Op_\hbar^{\mathrm{w}}(I_2^\Sigma),}
\qquad
\widehat{h}_\hbar\coloneqq \widehat{I}_{0,\hbar}.
\]

\medskip

Since $\mathcal{A}_{\mathrm{cl}}$ is Jacobi-commutative and $Q_\hbar$ satisfies
the Dirac condition, we have for all $f,g\in\mathcal{A}_{\mathrm{cl}}$
\[
\frac{1}{i\hbar}\big[Q_\hbar(f),Q_\hbar(g)\big]
=
Q_\hbar(\{f,g\}_\alpha)+\mathcal{O}(\hbar),
\]
and therefore
\begin{equation}\label{eq:Euler-approx-commute}
\change{[\widehat{I}_{i,\hbar},\widehat{I}_{j,\hbar}]
=
\mathcal{O}(\hbar^2)\, ,\qquad i,j\in\{0,1,2\},}
\quad\text{in operator norm.}
\end{equation}
In particular, the family $\{\widehat{I}_{0,\hbar},\widehat{I}_{1,\hbar}\change{,\widehat{I}_{2,\hbar}}\}$ is
\emph{asymptotically commuting} in the semiclassical limit.

This family generates a distinguished $C^\ast$-subalgebra
\[
\mathcal{A}_\hbar
\coloneqq
C^\ast(\widehat{I}_{0,\hbar},\widehat{I}_{1,\hbar}\change{,\widehat{I}_{2,\hbar}})
\subset \mathcal{B}(\mathcal{H}_\hbar),
\]
which is abelian up to semiclassical corrections (and exactly abelian if an exact
commuting quantization is chosen).
In the present Euler-top setting, $\mathcal{A}_\hbar$ is the quantum counterpart
of the classical invariant algebra generated by the \change{three} integrals $(I_0,I_1\change{,I_2})$,
namely \change{the quadratic Casimir, the Hamiltonian of the Poisson pencil dynamics, and the spatial-momentum component}.

This algebra will play the role of a \emph{robust commutative sector} of the open
quantum dynamics: it is singled out by the underlying integrable geometry and,
as we show below, can be preserved pointwise by a suitable choice of Lindblad
operators inducing pure dephasing.

\medskip

When, in addition, the quantization is chosen so that
$\widehat{I}_{0,\hbar}$\change{,} $\widehat{I}_{1,\hbar}$ \change{and $\widehat{I}_{2,\hbar}$} commute exactly (for instance,
when $Q_\hbar$ is realized in a representation where $I_0^\Sigma$\change{,}
$I_1^\Sigma$ \change{and $I_2^\Sigma$} quantize to commuting multiplication operators, or when the exact
commutation relations \eqref{eq:exact-commutant-choice} are imposed), the joint
spectral theorem applies.

More precisely, since $\widehat{I}_{0,\hbar}$\change{,} $\widehat{I}_{1,\hbar}$ \change{and $\widehat{I}_{2,\hbar}$} are
bounded self-adjoint operators that commute strongly, there exists a joint
projection-valued spectral measure $E_\hbar(\cdot)$ on $\mathbb{R}^{\change{3}}$ and a
Borel measure $\mu_\hbar$ such that the Hilbert space admits a unitary
decomposition as a direct integral
\[
\mathcal{H}_\hbar
\simeq
\int_{\mathbb{R}^{\change{3}}}^{\oplus}
\mathcal{H}_{\hbar,\nu}\, \dd\mu_\hbar(\nu).
\]
Here $\nu=(\nu_0,\nu_1\change{,\nu_2})$ denotes the joint spectral parameter, and
$\mathcal{H}_{\hbar,\nu}$ is the corresponding spectral multiplicity space.
In this representation, the operators $\widehat{I}_{j,\hbar}$ act as
multiplication operators,
\[
(\widehat{I}_{j,\hbar}\psi)(\nu)=\nu_j\,\psi(\nu),
\qquad j=0,1\change{,2},
\]
so that, fibrewise,
\[
\widehat{I}_{j,\hbar}\big|_{\mathcal{H}_{\hbar,\nu}}=\nu_j\,\mathbf{1}.
\]

Semiclassically, the joint spectral parameter $\nu=(\nu_0,\nu_1\change{,\nu_2})$ labels,
up to the usual Bohr--Sommerfeld quantization conditions
(see e.g.~\cite{GuilleminSternberg1990}),
the common level sets $\{I_0=\nu_0,\ I_1=\nu_1\change{,\ I_2=\nu_2}\}$ of the underlying contact
integrable system.

The direct-integral fibres $\mathcal{H}_{\hbar,\nu}$ may therefore be interpreted as
quantum counterparts of the classical invariant leaves (Liouville
tori/cylinders after reduction), providing the quantum ``fibres'' associated
with the classical invariant foliation.

Even without exact commutation, the asymptotic relation
\eqref{eq:Euler-approx-commute} is sufficient for the present purposes: it
ensures that $\mathcal{A}_\hbar$ behaves as an approximately classical,
asymptotically abelian sector in the semiclassical limit.  It is precisely this
sector that we will preserve pointwise by a suitable choice of Lindblad
operators in the next step.

\medskip

We now construct Lindblad generators whose dissipative part induces pure
dephasing with respect to the joint spectral data of
$(\widehat{I}_{0,\hbar},\widehat{I}_{1,\hbar}\change{,\widehat{I}_{2,\hbar}})$, in a way compatible with the
Weyl-type assumptions of Theorem~\ref{prop:egorov-criterion}.
Fix $N\in\mathbb{N}$, choose rates $\gamma_k>0$, and let
$\varphi_k\colon\mathbb{R}^{\change{3}}\to\mathbb{R}$ be bounded smooth functions.
We define the dissipative operators by joint functional calculus as
\begin{equation}\label{eq:Euler-Lindblad-operators}
L_{k,\hbar}
\coloneqq 
\sqrt{\gamma_k}\,
\varphi_k\!\big(\widehat{I}_{0,\hbar},\widehat{I}_{1,\hbar}\change{,\widehat{I}_{2,\hbar}}\big),
\qquad k=1,\dots,N.
\end{equation}
By construction, $L_{k,\hbar}\in\mathcal{A}_\hbar$, and therefore each Lindblad
operator commutes with every element of $\mathcal{A}_\hbar$. In particular,
\begin{equation}\label{eq:Euler-commutant-relations}
[L_{k,\hbar},\widehat{I}_{j,\hbar}]
=
[L_{k,\hbar}^\dagger,\widehat{I}_{j,\hbar}]
=0,
\qquad j\in\{0,1\change{,2}\},\ \ k=1,\dots,N.
\end{equation}
Equivalently, $L_{k,\hbar}\in\mathcal{A}_\hbar\subset\mathcal{A}_\hbar'$, since
$\mathcal{A}_\hbar$ is abelian.

\medskip

We choose the Hamiltonian part so as to quantize the contact Hamiltonian
$h=I_0$, namely
\begin{equation}\label{eq:Euler-Hamiltonian-choice}
\widehat{H}_\hbar
=
\widehat{h}_\hbar+\mathcal{O}(\hbar^2)
\qquad\text{in operator norm}.
\end{equation}
If this is strengthened to the exact commutation condition
\begin{equation}\label{eq:Euler-exact-commutation}
[\widehat{H}_\hbar,\widehat{I}_{0,\hbar}]
=
[\widehat{H}_\hbar,\widehat{I}_{1,\hbar}]
\change{=
[\widehat{H}_\hbar,\widehat{I}_{2,\hbar}]}
=
0,
\end{equation}
then the hypotheses of Lemma~\ref{lem:kernel-condition-contact-compatible} are
satisfied together with \eqref{eq:Euler-commutant-relations}. As a consequence,
the Heisenberg adjoint of the resulting GKSL generator satisfies
\begin{equation}\label{eq:Euler-constants-of-motion}
\mathcal{L}_\hbar^\dagger(\widehat{I}_{0,\hbar})=0,
\qquad
\mathcal{L}_\hbar^\dagger(\widehat{I}_{1,\hbar})=0\change{,
\qquad
\mathcal{L}_\hbar^\dagger(\widehat{I}_{2,\hbar})=0}.
\end{equation}
Thus the Euler-top quantum integrals $\widehat{I}_{0,\hbar}$\change{,}
$\widehat{I}_{1,\hbar}$ \change{and $\widehat{I}_{2,\hbar}$} are exact constants of motion (not merely asymptotically
conserved).

\medskip

Assume for simplicity that $\widehat{I}_{0,\hbar}$\change{,} $\widehat{I}_{1,\hbar}$ \change{and $\widehat{I}_{2,\hbar}$}
commute exactly. Then the joint spectral theorem yields a decomposition of the
Hilbert space into simultaneous spectral subspaces,
\[
\mathcal{H}_\hbar
=
\bigoplus_{\nu\in\mathrm{Spec}(\widehat{I}_{0,\hbar},\widehat{I}_{1,\hbar}\change{,\widehat{I}_{2,\hbar}})}
\mathcal{H}_{\hbar,\nu},
\qquad
\widehat{I}_{j,\hbar}\big|_{\mathcal{H}_{\hbar,\nu}}
=
\nu_j\,\mathbf{1},
\quad j\in\{0,1\change{,2}\}.
\]
Let $P_\nu$ denote the orthogonal projector onto $\mathcal{H}_{\hbar,\nu}$ and
write $\rho_{\nu\mu}\coloneqq P_\nu\rho P_\mu$ for the corresponding blocks of a density
matrix. Since $L_{k,\hbar}\in\mathcal{A}_\hbar$, each Lindblad operator acts as a
scalar on each joint eigenspace,
\[
L_{k,\hbar}P_\nu
=
\ell_{k,\nu}\,P_\nu,
\qquad
\ell_{k,\nu}
\coloneqq
\sqrt{\gamma_k}\,\varphi_k(\nu_0,\nu_1\change{,\nu_2})\in\mathbb{R}.
\]
A direct computation shows that the dissipator acts diagonally on the block
decomposition. For $\nu\neq\mu$ one finds
\begin{equation}\label{eq:Euler-dephasing-rate}
\frac{\dd}{\dd t}\rho_{\nu\mu}(t)\Big|_{\rm diss}
=
-\frac12
\sum_{k=1}^N
\big(\ell_{k,\nu}-\ell_{k,\mu}\big)^2\,
\rho_{\nu\mu}(t),
\end{equation}
while the diagonal blocks $\rho_{\nu\nu}$ are unaffected by the dissipator.
Hence the off-diagonal coherences decay exponentially, that is, there exists a constant $C>0$ such that, for all $t\ge 0$ and $\nu\neq\mu$,
\[
\|\rho_{\nu\mu}(t)\|
\;\le\;
C\,
\exp\!\Big(
-\tfrac12\,t\sum_{k=1}^N(\ell_{k,\nu}-\ell_{k,\mu})^2
\Big)\,
\|\rho_{\nu\mu}(0)\|\, ,
\]
whereas the populations $\Tr(\rho P_\nu)$ are preserved.
In other words, the open-system dynamics destroys coherences between states with
different values of the Euler-top invariants $(I_0,I_1\change{,I_2})$ without inducing
transitions between the corresponding sectors. This is precisely pure
dephasing relative to the commutative algebra $\mathcal{A}_\hbar$.

\medskip

Finally, the bi-Hamiltonian origin of the deformed Euler top provides two
(distinct) effective Hamiltonians associated with the underlying Poisson pencil,
namely the pair of classical functions $(C_0,C_1)$ satisfying
\[
X=\pi_0^\sharp \dd C_1=\pi_1^\sharp \dd C_0 .
\]
After homogeneous lift and restriction to the contact hypersurface $C$, this
yields two effective contact Hamiltonians
\[
H_0^{\mathrm{eff}},\; H_1^{\mathrm{eff}} \;\in\; \mathcal{A}_{\mathrm{cl}} .
\]
In the present minimal Euler-top pencil, one may simply take
$H_0^{\mathrm{eff}}=I_0$ and $H_1^{\mathrm{eff}}=I_1$, up to irrelevant affine
reparametrizations. Let $\widehat{H}_{0,\hbar}$ and $\widehat{H}_{1,\hbar}$ be the
corresponding semiclassical quantizations.

We define two Lindblad generators on $\mathcal{B}(\mathcal{H}_\hbar)$ by
\[
\mathcal{L}^{(a)}_\hbar(\rho)
=
-\frac{i}{\hbar}[\widehat{H}_{a,\hbar},\rho]
+\sum_{k=1}^N\Big(
L_{k,\hbar}\rho L_{k,\hbar}^\dagger
-\tfrac12\{L_{k,\hbar}^\dagger L_{k,\hbar},\rho\}
\Big),
\qquad a\in\{0,1\},
\]
using the \emph{same} family of dephasing operators $L_{k,\hbar}$ defined in
\eqref{eq:Euler-Lindblad-operators}.

By construction, the Lindblad operators are of Weyl type and their principal
symbols depend only on the homogeneous lifts $(I_0^\Sigma,I_1^\Sigma\change{,I_2^\Sigma})$. Hence
condition~(3) of Theorem~\ref{prop:egorov-criterion} is satisfied for both
generators. Moreover, the commutation relations
\eqref{eq:Euler-commutant-relations} imply that $\widehat{I}_{0,\hbar}$\change{,}
$\widehat{I}_{1,\hbar}$ \change{and $\widehat{I}_{2,\hbar}$} are common quantum constants of motion for both dynamics.
The semiclassical limit of the Hamiltonian commutator then yields the
corresponding contact dynamics in the sense of {\rm\ref{assumption_CC2}}.
\change{In particular, for the second generator the Egorov limit of the
Heisenberg evolution generated by $\widehat{H}_{1,\hbar}=Q_\hbar(I_1)$ reproduces
the contact Hamiltonian flow of $I_1$ on $C$, that is, the flow of
$X_{I_1}=\restr{X_{I_1^\Sigma}}{C}$ with $X_{I_1^\Sigma}=\tfrac{1}{2\sqrt{F_1}}X_{F_1}$.
Since $\Phi_\ast X_{F_1}=-X$, this flow projects under $\Phi$ onto the Euler-top
vector field $X$ of \eqref{eq:vector_field_Casimirs}, up to sign and the positive
conformal factor $1/(2\sqrt{F_1})$. The bi-Hamiltonian dynamics is therefore
recovered in the semiclassical limit through the second member
$\mathcal{L}^{(1)}_\hbar$ of the pencil, and not through $\mathcal{L}^{(0)}_\hbar$,
whose Hamiltonian part quantizes the Casimir lift $h=I_0$ and hence projects
trivially ($\Phi_\ast X_{F_0}=0$).}
Therefore, the families $\{\mathcal{L}^{(0)}_\hbar\}$ and
$\{\mathcal{L}^{(1)}_\hbar\}$ form a bi-Lindblad pair in the sense of
Definition~\ref{def:bi-Lindblad}.

Moreover, for any $\lambda\in[0,1]$, the convex combination
\[
\mathcal{L}^{(\lambda)}_\hbar
\coloneqq 
(1-\lambda)\mathcal{L}^{(0)}_\hbar+\lambda\mathcal{L}^{(1)}_\hbar
\]
is again a GKSL generator and remains contact-compatible with the same invariant
commutative sector $\mathcal{A}_\hbar$. In this way, the present
Euler-top-type Poisson--Lie model realizes a genuine
\emph{bi-Lindblad pencil}: a one-parameter family of open quantum dynamics
sharing the same dephasing mechanism and the same quantum integrals, whose
classical shadows recover different points of the underlying Poisson pencil.

\medskip

In this explicit testbed, integrability, dissipation, and decoherence coexist in
a geometrically controlled way. Decoherence selects the eigenstates of the
quantized classical integrals as pointer states, while the underlying
bi-Hamiltonian structure organizes both the classical and quantum dynamics.

More broadly, this Euler-top-type example illustrates a general mechanism
underlying the framework developed in this work. A classical bi-Hamiltonian
structure, once lifted to a contact system, singles out a distinguished family
of dissipated quantities whose Jacobi-commutative algebra encodes the integrable
geometry of the flow. Upon semiclassical quantization, these quantities generate
a commutative $C^\ast$-subalgebra of observables. Choosing the Lindblad data
inside its commutant produces a purely dephasing dissipative dynamics,
dynamically enforcing a superselection rule aligned with the classical
invariants.

In this sense, decoherence does not compete with integrability but rather
\emph{selects} it. The bi-Lindblad structure reflects the Poisson--Lie and
bi-Hamiltonian origin of the model: two compatible open quantum evolutions share
the same invariant commutative sector, and Egorov-type convergence guarantees
that, within this robust sector, the semiclassical Heisenberg dynamics reduces
to the intended contact Hamiltonian flow.

\begin{remark}
In many standard Lindblad models of dephasing, one starts from a prescribed
commutative algebra of observables (typically generated by the Hamiltonian or by
a chosen family of projectors) and constructs dissipative operators so as to
enforce invariance of that algebra under the Heisenberg evolution. From a
mathematical viewpoint, this corresponds to selecting a maximal abelian
$C^\ast$-subalgebra of $\mathcal{B}(\mathcal{H})$ and designing a GKSL generator
whose adjoint vanishes on it pointwise.

In the present setting, the direction of the construction is reversed. The
commutative $C^\ast$-algebra
\[
\mathcal{A}_\hbar
=
C^\ast(\widehat{I}_{0,\hbar},\dots,\widehat{I}_{n,\hbar})
\]
is not chosen \emph{a priori}, but is determined intrinsically by the classical
geometry: it arises from a Jacobi-commutative algebra of dissipated quantities
selected by a bi-Hamiltonian Poisson--Lie structure, lifted to a contact system
and subsequently quantized semiclassically. The Lindblad operators are then
chosen inside the commutant $\mathcal{A}_\hbar'$, which guarantees exact
invariance of $\mathcal{A}_\hbar$ and yields pure dephasing relative to its joint
spectral decomposition.

Thus, in the Euler-top-type Poisson--Lie model discussed above, the invariant
quantum algebra is not an external input but the quantization of a
Jacobi-commutative algebra of classical dissipated quantities. The associated
pointer sectors are therefore in canonical correspondence with the leaves of the
classical contact foliation singled out by the integrable structure.

This provides a natural setting in which decoherence, integrability, and the
semiclassical limit are aligned by construction rather than enforced by
model-dependent assumptions: decoherence dynamically selects the algebra
that encodes the classical integrable geometry.
\hfill$\diamond$
\end{remark}

\section{Conclusions}

We conclude by clarifying the physical meaning of the Lindblad constructions
introduced in this work, and in particular their relation to quantum
decoherence, integrability, and the emergence of classical behavior in the
semiclassical limit.

In the theory of open quantum systems, \emph{decoherence} refers to the dynamical
suppression of quantum coherences due to the interaction with an environment.
Mathematically, this manifests itself as the decay of off-diagonal matrix
elements of the density operator $\rho(t)$ in a preferred basis, while the
corresponding diagonal components (populations) remain invariant or evolve
autonomously. The states defining this preferred basis are known as
\emph{pointer states}.

In the framework developed here, the role of pointer states is played by the
joint eigenspaces of a distinguished commuting family of quantum observables,
namely the quantizations
\[
\{\widehat{I}_{0,\hbar},\widehat{I}_{1,\hbar},\dots,\widehat{I}_{n,\hbar}\}
\]
of the classical dissipated quantities singled out by the underlying contact
integrable system. By construction, the Lindblad operators $L_{k,\hbar}$ are
chosen as Weyl-type operators whose principal symbols depend only on the
homogeneous lifts of these quantities, and are implemented as functions of the commuting operators $\widehat{I}_{j,\hbar}$. As a result,
the dissipative part of the Lindblad generator leaves invariant the commutative
$C^\ast$-subalgebra
\[
\mathcal{A}_\hbar \coloneqq C^\ast(\widehat{I}_{0,\hbar},\dots,\widehat{I}_{n,\hbar}),
\]
while suppressing coherences between distinct joint eigenspaces of
$\mathcal{A}_\hbar$.

More precisely, upon decomposing the Hilbert space as a direct sum (or, more
generally, a direct integral) of joint spectral subspaces
\[
\mathcal{H}_\hbar
=
\bigoplus_{\lambda}
\mathcal{H}_{\hbar,\lambda},
\qquad
\widehat{I}_{j,\hbar}\big|_{\mathcal{H}_{\hbar,\lambda}}
= \lambda_j\,\mathbf{1},
\]
the Lindblad dynamics exponentially damps the off-diagonal blocks
$\rho_{\lambda\lambda'}(t)$ for $\lambda\neq\lambda'$, while leaving the diagonal
blocks $\rho_{\lambda\lambda}(t)$ invariant. In this sense, the dissipative
dynamics implements an environment-induced \emph{superselection rule} associated
with the integrals of motion.

This mechanism is the direct quantum analog of the classical notion of
\emph{dissipated quantities} in contact geometry. Although the classical contact
flow is irreversible, the Jacobi-commutative algebra generated by the quantities
$I_j$ organizes the dynamics and labels its invariant geometric structures.
At the quantum level, this organization reappears as decoherence towards the joint
eigenspaces of the corresponding observables.

From a physical perspective, several points are worth emphasising. First, the
Lindblad dynamics constructed here describe a class of open quantum systems in
which dissipation does \emph{not} destroy integrability. Instead, dissipation is
aligned with the classical invariants, leading to pure dephasing rather than
relaxation. While non-generic, this situation is physically relevant in several
contexts, including weak coupling to classical noise, measurement-induced
decoherence, and engineered reservoirs.

Second, the contact-geometric formulation clarifies the semiclassical meaning of
this phenomenon. The contact Hamiltonian flow encodes irreversible classical
dynamics, while the Jacobi algebra of dissipated quantities captures the
surviving integrable structure. The Lindblad generators constructed in this work
provide a quantum realization of this picture: in the Heisenberg picture, the
quantum evolution converges in the semiclassical limit to the contact dynamics,
while the dissipated quantities become genuine quantum constants of motion.

Third, the bi-Lindblad structures introduced here reflect the underlying
bi-Hamiltonian geometry. Just as a classical bi-Hamiltonian system admits a
pencil of compatible Poisson structures and Hamiltonians, the quantum system
admits a family of compatible Lindblad generators whose convex combinations
remain completely positive and trace preserving. This establishes a precise
correspondence between classical integrable hierarchies and families of
admissible open quantum dynamics.

The Euler-top-type Poisson--Lie example analyzed in this work should therefore
be understood as a structural testbed rather than as a literal model of rigid
body dynamics. Its role is to demonstrate, in a concrete and low-dimensional
setting, how integrability, dissipation, and decoherence can coexist in a
geometrically controlled way. More broadly, our results suggest that contact
geometry provides a natural classical language for describing the semiclassical
limit of certain Markovian open quantum systems, in which decoherence selects
classical invariants rather than equilibria.

This perspective opens the door to systematic constructions of open quantum
models in which dissipation, integrability, and semiclassical consistency are
aligned by design, rather than imposed by ad hoc assumptions.

\medskip

\section*{Declarations}

\subsection*{Acknowledgments}
The authors are grateful to the referee for their valuable comments and suggestions.

\subsection*{Funding}

L.~Colombo acknowledges financial support from Grant PID2022-137909-NB-C22 funded by the Spanish Ministry of Science and Innovation.

\subsection*{Competing Interests}

The authors declare no known competing financial interests or personal relationships that could have appeared to influence the work reported in this manuscript.

\subsection*{Data Availability Statement} No datasets were generated or analyzed during the current study. All mathematical derivations are contained within the article.

\addcontentsline{toc}{section}{References}
\printbibliography

@book{Arnold1978,
  title = {Mathematical {{Methods}} of {{Classical Mechanics}}},
  author = {Arnol'd, V. I.},
  year = {1978},
  series = {Graduate {{Texts}} in {{Mathematics}}},
  publisher = {Springer-Verlag},
  address = {New York},
  doi = {10.1007/978-1-4757-1693-1},
  isbn = {978-1-4757-1693-1},
  langid = {english}
}

@article{ballesteros2017poisson,
  title={Poisson--Lie groups, bi-Hamiltonian systems and integrable deformations},
  author={Ballesteros, Angel and Marrero, Juan C and Ravanpak, Zohreh},
  journal={J. Phys. A},
  volume={50},
  number={14},
  pages={145204},
  year={2017},
  publisher={IOP Publishing},
  doi={10.1088/1751-8121/aa617b}
}

@book{Audin2004,
  title = {Torus {{Actions}} on {{Symplectic Manifolds}}},
  author = {Audin, Mich{\`e}le},
  year = {2004},
  publisher = {Birkh{\"a}user Basel},
  address = {Basel},
  doi = {10.1007/978-3-0348-7960-6},
  isbn = {978-3-0348-9637-5 978-3-0348-7960-6},
  langid = {english}
}

@book{Zworski2012,
  author    = {Maciej Zworski},
  title     = {Semiclassical Analysis},
  publisher = {American Mathematical Society},
  year      = {2012},
  series    = {Graduate Studies in Mathematics},
  volume    = {138},
  isbn = {978-0-8218-8320-4},
  doi = {10.1090/gsm/138}
}

@book{DimassiSjoestrand1999, 
place={Cambridge}, series={London Mathematical Society Lecture Note Series}, title={Spectral Asymptotics in the Semi-Classical Limit}, publisher={Cambridge University Press}, author={Dimassi, M. and Sj{\"o}strand, J.}, year={1999}, collection={London Mathematical Society Lecture Note Series},
isbn={
9780511662195},
doi={10.1017/CBO9780511662195}
}

@book{Folland1989,
  author    = {Gerald B. Folland},
  title     = {Harmonic Analysis in Phase Space},
  publisher = {Princeton University Press},
  year      = {1989}
}

@book{GuilleminSternberg1990,
  title={Semi-Classical Analysis},
  author={Guillemin, V. and Sternberg, S.},
  publisher={International Press},
  place = {Boston, MA},
  year={2013},
  isbn ={978-1-57146-276-3}
}

@book{Martinez2002,
  author    = {Andr{\'e} Martinez},
  title     = {An Introduction to Semiclassical and Microlocal Analysis},
  publisher = {Springer},
  year = {2002},
  doi = {10.1007/978-1-4757-4495-8}
}

@article{Gutt2000,
  author  = {Simone Gutt},
  title   = {An explicit $*$-product on the cotangent bundle of a Lie group},
  journal = {Lett. Math. Phys.},
  volume  = {7},
  number = {3},
  year    = {1983},
  pages   = {249--258},
  doi ={10.1007/BF00400441}
}

@article{B.G.G2017a,
  title = {Remarks on {{Contact}} and {{Jacobi Geometry}}},
  author = {Bruce, Andrew James and Grabowska, Katarzyna and Grabowski, Janusz},
  year = {2017},
  oldmonth = jul,
  journal = {SIGMA Symmetry Integrability Geom. Methods Appl.},
  volume = {13},
  pages = {059},
  publisher = {{SIGMA. Symmetry, Integrability and Geometry: Methods and Applications}},
  issn = {18150659},
  doi = {10.3842/SIGMA.2017.059},
  langid = {english}
}

@article{B.K.M2022,
  title = {Applications of {{Nijenhuis}} Geometry: Non-Degenerate Singular Points of {{Poisson-Nijenhuis}} Structures},
  shorttitle = {Applications of {{Nijenhuis}} Geometry},
  author = {Bolsinov, Alexey V. and Konyaev, Andrey Yu. and Matveev, Vladimir S.},
  year = {2022},
  journal = {Eur. J. Math.},
  volume = {8},
  number = {4},
  pages = {1355--1376},
  issn = {2199-675X,2199-6768},
  doi = {10.1007/s40879-020-00429-6},
  mrnumber = {4510609}
}

@article{B.K.M2022a,
  title = {Nijenhuis Geometry},
  author = {Bolsinov, Alexey V. and Konyaev, Andrey Yu. and Matveev, Vladimir S.},
  year = {2022},
  journal = {Adv. Math.},
  volume = {394},
  pages = {Paper No. 108001, 52},
  issn = {0001-8708,1090-2082},
  doi = {10.1016/j.aim.2021.108001},
  mrnumber = {4355721}
}

@article{C.d.L+2023,
  title = {Liouville-{{Arnold}} Theorem for Homogeneous Symplectic and Contact {{Hamiltonian}} Systems},
  author = {Colombo, Leonardo and {de Le{\'o}n}, Manuel and Lainz, Manuel and {L{\'o}pez-Gord{\'o}n}, Asier},
  year = {2025},
  oldmonth = May,
  journal = {Geom. Mech.},
  volume = {02},
  number = {03},
  pages = {275--307},
  doi = {10.1142/S2972458925400039},
}

@article{C.D.N2010,
  title = {On {{Jacobi}} Quasi-{{Nijenhuis}} Algebroids and {{Courant}}--{{Jacobi}} Algebroid Morphisms},
  author = {Caseiro, Raquel and De Nicola, Antonio and {Nunes da Costa}, J.},
  year = {2010},
  oldmonth = jun,
  journal = {J. Geom. Phys.},
  volume = {60},
  number = {6-8},
  pages = {951--961},
  issn = {0393-0440},
  doi = {10.1016/j.geomphys.2010.02.011}
}

@book{C.F.M2021,
  title = {Lectures on {{Poisson}} Geometry},
  author = {Crainic, Marius and Fernandes, Rui Loja and Mărcuţ, Ioan},
  date = {2021},
  series = {Graduate {{Studies}} in {{Mathematics}}},
  volume = {217},
  publisher = {American Mathematical Society, Providence, RI},
  doi = {10.1090/gsm/217},
  isbn = {978-1-4704-6430-1},
  mrnumber = {4328925},
  pagetotal = {xix+479}
}

@incollection{C.M.P1993,
  title = {The {{Bihamiltonian Approach}} to {{Integrable Systems}}},
  booktitle = {Modern {{Group Analysis}}: {{Advanced Analytical}} and {{Computational Methods}} in {{Mathematical Physics}}: {{Proceedings}} of the {{International Workshop Acireale}}, {{Catania}}, {{Italy}}, {{October}} 27--31, 1992},
  author = {Casati, Paolo and Magri, Franco and Pedroni, Marco},
  editor = {Ibragimov, N. H. and Torrisi, M. and Valenti, A.},
  year = {1993},
  pages = {101--110},
  publisher = {Springer Netherlands},
  address = {Dordrecht},
  doi = {10.1007/978-94-011-2050-0_10},
  isbn = {978-94-011-2050-0},
  langid = {english}
}

@article{d.L.M+2003,
  title = {On the Computation of the {{Lichnerowicz-Jacobi}} Cohomology},
  author = {{de Le{\'o}n}, Manuel and L{\'o}pez, Bel{\'e}n and Marrero, Juan C. and Padr{\'o}n, Edith},
  year = {2003},
  journal = {J. Geom. Phys.},
  volume = {44},
  number = {4},
  pages = {507--522},
  issn = {0393-0440,1879-1662},
  doi = {10.1016/S0393-0440(02)00056-6},
  mrnumber = {1943175}
}

@article{D.L.M1991,
  title = {Structure Locale Des Vari{\'e}t{\'e}s de {{Jacobi}}},
  author = {Dazord, Pierre and Lichnerowicz, Andr{\'e} and Marle, Charles-Michel},
  year = {1991},
  journal = {J. Math. Pures Appl. (9)},
  volume = {70},
  number = {1},
  pages = {101--152},
  issn = {0021-7824},
  mrnumber = {1091922}
}

@article{d.L2019a,
  title = {A Review on Contact {{Hamiltonian}} and {{Lagrangian}} Systems},
  author = {{de Le{\'o}n}, Manuel and Lainz, Manuel},
  year = {2019},
  journal = {Rev. Acad. Canaria Cienc.},
  volume = {XXXI},
  eprint = {2011.05579},
  pages = {1--46},
  archiveprefix = {arXiv}
}

@article{d.L2019b,
  title = {Contact {{Hamiltonian}} Systems},
  author = {{de Le{\'o}n}, Manuel and Lainz Valc{\'a}zar, Manuel},
  year = {2019},
  oldmonth = oct,
  journal = {J. Math. Phys.},
  volume = {60},
  number = {10},
  pages = {102902},
  publisher = {American Institute of Physics},
  issn = {0022-2488},
  doi = {10.1063/1.5096475}
}

@article{D.M.P2011,
  title = {Reduction of {{Poisson-Nijenhuis Lie}} Algebroids to Symplectic-{{Nijenhuis Lie}} Algebroids with a Nondegenerate {{Nijenhuis}} Tensor},
  author = {De Nicola, Antonio and Marrero, Juan Carlos and Padr{\'o}n, Edith},
  year = {2011},
  journal = {J. Phys. A},
  volume = {44},
  number = {42},
  pages = {425206, 35},
  issn = {1751-8113,1751-8121},
  doi = {10.1088/1751-8113/44/42/425206},
  mrnumber = {2845004}
}

@inproceedings{F.M.P2000,
  title = {A Bihamiltonian Approach to Separation of Variables in Mechanics},
  booktitle = {Proceedings of the {{Workshop}} on {{Nonlinearity}}, {{Integrability}} and {{All That}}: {{Twenty Years}} after {{NEEDS}} '79 ({{Gallipoli}}, 1999)},
  author = {Falqui, G. and Magri, F. and Pedroni, M.},
  year = {2000},
  pages = {258--266},
  publisher = {World Sci. Publ., River Edge, NJ},
  mrnumber = {1772987}
}

@article{Fernandes1994,
  title = {Completely Integrable Bi-{{Hamiltonian}} Systems},
  author = {Fernandes, Rui L.},
  year = {1994},
  oldmonth = jan,
  journal = {J. Dynam. Differential Equations},
  volume = {6},
  number = {1},
  pages = {53--69},
  issn = {1572-9222},
  doi = {10.1007/BF02219188},
  langid = {english}
}

@article{G.G2022a,
  title = {A Geometric Approach to Contact {{Hamiltonians}} and Contact {{Hamilton}}--{{Jacobi}} Theory},
  author = {Grabowska, Katarzyna and Grabowski, Janusz},
  year = {2022},
  oldmonth = nov,
  journal = {J. Phys. A: Math. Theor.},
  volume = {55},
  number = {43},
  pages = {435204},
  publisher = {IOP Publishing},
  issn = {1751-8121},
  doi = {10.1088/1751-8121/ac9adb},
  langid = {english}
}

@article{G.G2023,
  title = {Reductions: Precontact versus Presymplectic},
  shorttitle = {Reductions},
  author = {Grabowska, Katarzyna and Grabowski, Janusz},
  year = {2023},
  volume = {202},
  pages={2803-2839},
  oldmonth = jun,
  journal = {Ann. Mat. Pura Appl.},
  issn = {1618-1891},
  doi = {10.1007/s10231-023-01341-y},
  langid = {english}
}

@article{G.G2024a,
  title = {Contact Geometric Mechanics: The {{Tulczyjew}} Triples},
  shorttitle = {Contact Geometric Mechanics},
  author = {Grabowska, Katarzyna and Grabowski, Janusz},
  year = {2024},
  oldmonth = sep,
  journal = {Adv. Theor. Math. Phys.},
  volume = {28},
  number = {2},
  pages = {599--654},
  publisher = {International Press of Boston},
  issn = {1095-0753},
  doi = {10.4310/ATMP.240914022224},
  langid = {english}
}

@article{G.I.M+2004,
  title = {Poisson-{{Jacobi}} Reduction of Homogeneous Tensors},
  author = {Grabowski, J. and Iglesias, D. and Marrero, J. C. and Padr{\'o}n, E. and Urba{\'n}ski, P.},
  year = {2004},
  journal = {J. Phys. A: Math. Gen.},
  volume = {37},
  number = {20},
  pages = {5383--5399},
  issn = {0305-4470,1751-8121},
  doi = {10.1088/0305-4470/37/20/010},
  mrnumber = {2065677}
}

@article{G.M.S2003,
  title = {Bi-{{Hamiltonian}} Partially Integrable Systems},
  author = {Giachetta, G. and Mangiarotti, L. and Sardanashvily, G.},
  year = {2003},
  journal = {J. Math. Phys.},
  volume = {44},
  number = {5},
  pages = {1984--1997},
  issn = {0022-2488,1089-7658},
  doi = {10.1063/1.1566453},
  mrnumber = {1972759}
}

@article{G.L.R2023,
  title = {Symmetries, {{Conservation}} and {{Dissipation}} in {{Time-Dependent Contact Systems}}},
  author = {Gaset, Jordi and López-Gordón, Asier and Rivas, Xavier},
  date = {2023},
  journaltitle = {Fortschritte der Phys.},
  volume = {71},
  number = {8--9},
  pages = {2300048},
  issn = {1521-3978},
  doi = {10.1002/prop.202300048},
  url = {https://onlinelibrary.wiley.com/doi/abs/10.1002/prop.202300048},
  langid = {english}
}

@article{Grabowski2013a,
  title = {Graded Contact Manifolds and Contact {{Courant}} Algebroids},
  author = {Grabowski, Janusz},
  year = {2013},
  journal = {J. Geom. Phys.},
  volume = {68},
  pages = {27--58},
  issn = {0393-0440,1879-1662},
  doi = {10.1016/j.geomphys.2013.02.001},
  mrnumber = {3035113}
}

@article{I.L.M+1997,
  title = {Co-Isotropic and {{Legendre}} - {{Lagrangian}} Submanifolds and Conformal {{Jacobi}} Morphisms},
  author = {Ib{\'a}{\~n}ez, Ra{\'u}l and de Le{\'o}n, Manuel and Marrero, Juan C. and {Mart{\'i}n de Diego}, David},
  year = {1997},
  oldmonth = aug,
  journal = {J. Phys. A: Math. Gen.},
  volume = {30},
  number = {15},
  pages = {5427--5444},
  issn = {0305-4470},
  doi = {10.1088/0305-4470/30/15/027},
  langid = {english}
}

@article{K.M1990,
  title = {Poisson-{{Nijenhuis}} Structures},
  author = {{Kosmann-Schwarzbach}, Yvette and Magri, Franco},
  year = {1990},
  journal = {Ann. Inst. H. Poincar{\'e} Phys. Th{\'e}or.},
  volume = {53},
  number = {1},
  pages = {35--81},
  issn = {0246-0211},
  mrnumber = {1077465},
  url={https://www.numdam.org/item/AIHPA_1990__53_1_35_0/
}
}

@book{L.M1987,
  title = {Symplectic {{Geometry}} and {{Analytical Mechanics}}},
  author = {Libermann, Paulette and Marle, Charles-Michel},
  year = {1987},
  publisher = {Springer Netherlands},
  address = {Dordrecht},
  doi = {10.1007/978-94-009-3807-6},
  isbn = {978-90-277-2439-7 978-94-009-3807-6},
  langid = {english}
}

@article{Lichnerowicz1977a,
  title = {Vari{\'e}t{\'e}s de {{Jacobi}} et Alg{\`e}bres de {{Lie}} Associ{\'e}es},
  author = {Lichnerowicz, Andr{\'e}},
  year = {1977},
  journal = {C. R. Acad. Sci. Paris S{\'e}r. A-B},
  volume = {285},
  number = {6},
  pages = {A455--A459},
  issn = {0151-0509},
  fjournal = {Comptes Rendus Hebdomadaires des S{\'e}ances de l'Acad{\'e}mie des Sciences. S{\'e}ries A et B},
  mrclass = {58F05},
  mrnumber = {455037}
}

@article{Lichnerowicz1978,
  title = {Les Vari{\'e}t{\'e}s de {{Jacobi}} et Leurs Alg{\`e}bres de {{Lie}} Associ{\'e}es},
  author = {Lichnerowicz, Andr{\'e}},
  year = {1978},
  journal = {J. Math. Pures Appl. (9)},
  volume = {57},
  number = {4},
  pages = {453--488},
  issn = {0021-7824},
  mrnumber = {524629}
}

@phdthesis{Lopez-Gordon2024,
  title = {The Geometry of Dissipation},
  author = {{L{\'o}pez-Gord{\'o}n}, Asier},
  year = {2024},
  oldmonth = sep,
  eprint = {2409.11947},
  archiveprefix = {arXiv},
  copyright = {All rights reserved},
  school = {Universidad Aut{\'o}noma de Madrid},
  doi = {10.48550/arXiv.2409.11947}
}

@inproceedings{M.C.F+1997,
  title = {Eight Lectures on Integrable Systems},
  booktitle = {Integrability of {{Nonlinear Systems}}},
  author = {Magri, F. and Casati, P. and Falqui, G. and Pedroni, M.},
  editor = {{Kosmann-Schwarzbach}, Y. and Grammaticos, B. and Tamizhmani, K. M.},
  year = {1997},
  series = {Lecture {{Notes}} in {{Physics}}},
  pages = {256--296},
  publisher = {Springer},
  address = {Berlin, Heidelberg},
  doi = {10.1007/BFb0113698},
  isbn = {978-3-540-69521-9},
  langid = {english}
}

@article{M.M.P1999,
  title = {Jacobi---{{Nijenhuis}} Manifolds and Compatible {{Jacobi}} Structures},
  author = {Marrero, Juan C and Monterde, Juan and Padr{\'o}n, Edith},
  year = {1999},
  oldmonth = jan,
  journal = {C. R. Acad. Sci. Paris S{\'e}r. I Math.},
  volume = {329},
  number = {9},
  pages = {797--802},
  issn = {07644442},
  doi = {10.1016/S0764-4442(99)90010-1},
  langid = {english}
}

@unpublished{M.M1984,
  title = {A Geometrical Characterization of Integrable {{Hamiltonian}} Systems through the Theory of {{Poisson}} - {{Nijenhuis}} Manifolds},
  author = {Magri, Franco and Morosi, C},
  year = {1984},
  oldmonth = jan,
  note= {{Quaderno} S 19, Universit\`a degli Studi di Milano}
}

@article{M.P1996,
  title = {On the {{Euler}} Equation: Bi-{{Hamiltonian}} Structure and Integrals in Involution},
  shorttitle = {On the {{Euler}} Equation},
  author = {Morosi, Carlo and Pizzocchero, Livio},
  year = {1996},
  journal = {Lett. Math. Phys.},
  volume = {37},
  number = {2},
  pages = {117--135},
  issn = {0377-9017,1573-0530},
  doi = {10.1007/BF00416015},
  mrnumber = {1391194}
}

@article{N.P2002,
  title = {Reduction of {{Jacobi-Nijenhuis}} Manifolds},
  author = {{Nunes da Costa}, Joana Margarida and Petalidou, Fani},
  year = {2002},
  journal = {J. Geom. Phys.},
  volume = {41},
  number = {3},
  pages = {181--195},
  issn = {0393-0440,1879-1662},
  doi = {10.1016/S0393-0440(01)00054-7},
  mrnumber = {1877605}
}

@inproceedings{NunesdaCosta1998,
  title = {Bi-{{Hamiltonian}} Manifolds---an Application to the Relativistic {{Toda}} Lattice and a Generalization to the {{Jacobi}} Manifolds},
  booktitle = {Proceedings of the 1st {{International Meeting}} on {{Geometry}} and {{Topology}} ({{Braga}}, 1997)},
  author = {{Nunes da Costa}, Joana Margarida},
  year = {1998},
  pages = {149--156},
  publisher = {Cent. Mat. Univ. Minho, Braga},
  mrnumber = {1694950}
}

@article{NunesdaCosta1998a,
  title = {Compatible {{Jacobi}} Manifolds: Geometry and Reduction},
  shorttitle = {Compatible {{Jacobi}} Manifolds},
  author = {{Nunes da Costa}, Joana Margarida},
  year = {1998},
  oldmonth = jan,
  journal = {J. Phys. A: Math. Gen.},
  volume = {31},
  number = {3},
  pages = {1025--1033},
  issn = {0305-4470},
  doi = {10.1088/0305-4470/31/3/013},
  langid = {english}
}

@article{P.N2003,
  title = {Local Structure of {{Jacobi}}--{{Nijenhuis}} Manifolds},
  author = {Petalidou, Fani and {Nunes da Costa}, Joana Margarida},
  year = {2003},
  oldmonth = mar,
  journal = {J. Geom. Phys.},
  volume = {45},
  number = {3-4},
  pages = {323--367},
  issn = {03930440},
  doi = {10.1016/S0393-0440(01)00074-2},
  langid = {english}
}

@article{Petalidou2000,
  title = {Sur La Symplectisation de Structures Bihamiltoniennes},
  author = {Petalidou, Fani},
  year = {2000},
  journal = {Bull. Sci. Math.},
  volume = {124},
  number = {4},
  pages = {255--286},
  issn = {0007-4497,1952-4773},
  doi = {10.1016/S0007-4497(00)00132-9},
  mrnumber = {1771937}
}

@article{Petalidou2002,
  title = {On a New Relation between {{Jacobi}} and Homogeneous {{Poisson}} Manifolds},
  author = {Petalidou, Fani},
  year = {2002},
  journal = {J. Phys. A},
  volume = {35},
  number = {10},
  pages = {2505--2518},
  issn = {0305-4470,1751-8121},
  doi = {10.1088/0305-4470/35/10/314},
  mrnumber = {1909408}
}

@book{BreuerPetruccione,
  title = {The {{Theory}} of {{Open Quantum Systems}}},
  author = {Breuer, Heinz-Peter and Petruccione, Francesco},
  date = {2007-01-25},
  publisher = {Oxford University Press},
  doi = {10.1093/acprof:oso/9780199213900.001.0001},
  url = {https://doi.org/10.1093/acprof:oso/9780199213900.001.0001},
  isbn = {978-0-19-921390-0}
}

@inproceedings{Colombo2025,
  title = {Homogeneous {{Bi-Hamiltonian Structures}} and~{{Integrable Contact Systems}}},
  booktitle = {Geometric {{Science}} of {{Information}}},
  author = {Colombo, Leonardo and family=León, given=Manuel, prefix=de, useprefix=true and Eyrea Irazú, María Emma and López-Gordón, Asier},
  editor = {Nielsen, Frank and Barbaresco, Frédéric},
  date = {2025-10-20},
  pages = {30--39},
  publisher = {Springer Nature Switzerland},
  location = {Cham},
  doi = {10.1007/978-3-032-03924-8_4},
  isbn = {978-3-032-03924-8},
  langid = {english}
}

@article{Drinfeld1988,
  title = {Quantum Groups},
  author = {Drinfel'd, V. G.},
  date = {1988-04-01},
  journaltitle = {J. Soviet Math.},
  volume = {41},
  number = {2},
  pages = {898--915},
  issn = {1573-8795},
  doi = {10.1007/BF01247086},
  url = {https://doi.org/10.1007/BF01247086},
  langid = {english}
}

@incollection{Drinfeld1990,
  title = {Hamiltonian Structures on {{Lie}} Groups, {{Lie}} Bialgebras and the Geometric Meaning of the Classical {{Yang-Baxter}} Equations},
  booktitle = {Yang-{{Baxter Equation}} in {{Integrable Systems}}},
  author = {Drinfel'd, V. G.},
  date = {1990-03},
  series = {Advanced {{Series}} in {{Mathematical Physics}}},
  volume = {Volume 10},
  number = {Volume 10},
  pages = {222--225},
  publisher = {WORLD SCIENTIFIC},
  doi = {10.1142/9789812798336_0009},
  url = {https://www.worldscientific.com/doi/10.1142/9789812798336_0009},
  isbn = {978-981-02-0120-3}
}

@article{G.G.M+2020a,
  title = {New Contributions to the {{Hamiltonian}} and {{Lagrangian}} Contact Formalisms for Dissipative Mechanical Systems and Their Symmetries},
  author = {Gaset, Jordi and Gràcia, Xavier and Muñoz-Lecanda, Miguel C. and Rivas, Xavier and Román-Roy, Narciso},
  date = {2020-05},
  journaltitle = {Int. J. Geom. Methods Mod. Phys.},
  volume = {17},
  number = {06},
  pages = {2050090},
  publisher = {World Scientific Publishing Co.},
  issn = {0219-8878},
  doi = {10.1142/S0219887820500905},
  url = {https://www.worldscientific.com/doi/abs/10.1142/S0219887820500905}
}

@article{G.I.M+2023,
  title = {Unimodularity and Invariant Volume Forms for {{Hamiltonian}} Dynamics on {{Poisson}}–{{Lie}} Groups},
  author = {Gutiérrez-Sagredo, I and Iglesias Ponte, D and Marrero, J C and Padrón, E and Ravanpak, Z},
  date = {2023-01},
  journaltitle = {J. Phys. A: Math. Theor.},
  volume = {56},
  number = {1},
  pages = {015203},
  publisher = {IOP Publishing},
  issn = {1751-8121},
  doi = {10.1088/1751-8121/acb116},
  url = {https://doi.org/10.1088/1751-8121/acb116},
  langid = {english}
}

@incollection{KS2004,
  title = {Lie {{Bialgebras}}, {{Poisson Lie Groups}}, and {{Dressing Transformations}}},
  booktitle = {Integrability of {{Nonlinear Systems}}},
  author = {Kosmann-Schwarzbach, Yvette},
  editor = {Kosmann-Schwarzbach, Yvette and Tamizhmani, K. M. and Grammaticos, Basil},
  date = {2004},
  pages = {107--173},
  publisher = {Springer},
  location = {Berlin, Heidelberg},
  doi = {10.1007/978-3-540-40962-5_5},
  url = {https://doi.org/10.1007/978-3-540-40962-5_5},
  isbn = {978-3-540-40962-5},
  langid = {english}
}

@article{KS2008,
  title = {Poisson {{Manifolds}}, {{Lie Algebroids}}, {{Modular Classes}}: A {{Survey}}},
  shorttitle = {Poisson {{Manifolds}}, {{Lie Algebroids}}, {{Modular Classes}}},
  author = {Kosmann-Schwarzbach, Yvette},
  date = {2008-01-16},
  journaltitle = {SIGMA Symmetry Integrability Geom. Methods Appl.},
  volume = {4},
  pages = {005},
  issn = {1815-0659},
  doi = {10.3842/SIGMA.2008.005},
  url = {http://arxiv.org/abs/0710.3098}
}

@article{LuWeinstein1990,
  title = {Poisson {{Lie}} Groups, Dressing Transformations, and {{Bruhat}} Decompositions},
  author = {Lu, Jiang-Hua and Weinstein, Alan},
  date = {1990-01},
  journaltitle = {J. Differential Geom.},
  volume = {31},
  number = {2},
  pages = {501--526},
  publisher = {Lehigh University},
  issn = {0022-040X},
  doi = {10.4310/jdg/1214444324},
  url = {https://projecteuclid.org/journals/journal-of-differential-geometry/volume-31/issue-2/Poisson-Lie-groups-dressing-transformations-and-Bruhat-decompositions/10.4310/jdg/1214444324.full}
}

@article{Lindblad1976,
  title = {On the Generators of Quantum Dynamical Semigroups},
  author = {Lindblad, G.},
  date = {1976-06-01},
  journaltitle = {Commun. Math. Phys.},
  volume = {48},
  number = {2},
  pages = {119--130},
  issn = {1432-0916},
  doi = {10.1007/BF01608499},
  url = {https://doi.org/10.1007/BF01608499},
  langid = {english}
}

@article{STS1985,
  title = {Dressing {{Transformations}} and {{Poisson Group Actions}}},
  author = {Semenov-Tian-Shansky, M. A.},
  date = {1985-12-31},
  journaltitle = {Publ. Res. Inst. Math. Sci.},
  volume = {21},
  number = {6},
  pages = {1237--1260},
  issn = {0034-5318},
  doi = {10.2977/prims/1195178514},
  url = {https://ems.press/journals/prims/articles/3265},
  langid = {english}
}

@article{STS1983,
  title = {What Is a Classical R-Matrix?},
  author = {Semenov-Tian-Shansky, M. A.},
  date = {1983-10-01},
  journaltitle = {Funct. Anal. Its Appl.},
  volume = {17},
  number = {4},
  pages = {259--272},
  issn = {1573-8485},
  doi = {10.1007/BF01076717},
  url = {https://doi.org/10.1007/BF01076717},
  langid = {english}
}

\end{document}